\documentclass[11pt,a4paper]{article}
\pdfoutput=1

\usepackage{jheppub-nosort,amsthm}

\usepackage{amsmath,amssymb,amsthm,amscd}
\usepackage{tikz}
\tikzset{node distance=2cm, auto}
\usepackage{epsfig}
\usepackage{psfrag,comment}
\usepackage{graphicx}

\newcommand{\CH}{{\cal H}}

\newcommand{\CN}{{\cal N}}

\def\BN{{\mathbb N}}
\def\BZ{{\mathbb Z}}

\def\BC{{\mathbb C}}
\def\BP{{\mathbb P}}

\newcommand{\be}{\begin{equation}}
\newcommand{\ee}{\end{equation}}
\newcommand{\ba}{\begin{aligned}}
\newcommand{\ea}{\end{aligned}}
\newcommand{\bea}{\begin{eqnarray}}
\newcommand{\eea}{\end{eqnarray}}
\newcommand{\bean}{\begin{eqnarray*}}
\newcommand{\eean}{\end{eqnarray*}}

\def\r{\right\rangle}

\def\1{\mathbf{1}}
\def\0{|\1\r}

\newcommand{\rme}{{\rm e}}
\newcommand{\rmi}{{\rm i}}
\newcommand{\rmd}{{\rm d}}

\def\XXint#1#2#3{{\setbox0=\hbox{$#1{#2#3}{\int}$}
     \vcenter{\hbox{$#2#3$}}\kern-.5\wd0}}

\newtheorem{thm}{Theorem}
\newtheorem{lem}{Lemma}

\newdimen\tableauside\tableauside=1.0ex
\newdimen\tableaurule\tableaurule=0.4pt
\newdimen\tableaustep
\def\phantomhrule#1{\hbox{\vbox to0pt{\hrule height\tableaurule width#1\vss}}}
\def\phantomvrule#1{\vbox{\hbox to0pt{\vrule width\tableaurule height#1\hss}}}
\def\sqr{\vbox{%
  \phantomhrule\tableaustep
  \hbox{\phantomvrule\tableaustep\kern\tableaustep\phantomvrule\tableaustep}%
  \hbox{\vbox{\phantomhrule\tableauside}\kern-\tableaurule}}}
\def\squares#1{\hbox{\count0=#1\noindent\loop\sqr
  \advance\count0 by-1 \ifnum\count0>0\repeat}}
\def\tableau#1{\vcenter{\offinterlineskip
  \tableaustep=\tableauside\advance\tableaustep by-\tableaurule
  \kern\normallineskip\hbox
    {\kern\normallineskip\vbox
      {\gettableau#1 0 }%
     \kern\normallineskip\kern\tableaurule}%
  \kern\normallineskip\kern\tableaurule}}
\def\gettableau#1{\ifnum#1=0\let\next=\null\else
\squares{#1}\let\next=\gettableau\fi\next}
\title{Resurgent Transseries and the Holomorphic Anomaly}

\author[a]{Ricardo~Couso--Santamar\'\i a,}
\affiliation[a]{Departamento de F\'\i sica de Part\'\i culas and IGFAE, Universidade de Santiago de Compostela,\\ E--15782 Santiago de Compostela, Spain\\}
\emailAdd{ricardo.couso@usc.es}

\author[a,b]{Jos\'e~D.~Edelstein,}
\affiliation[b]{Centro de Estudios Cient\'\i ficos, CECs, Casilla 1469, Valdivia, Chile\\}
\emailAdd{jose.edelstein@usc.es}

\author[c]{Ricardo~Schiappa,}
\affiliation[c]{CAMGSD, Departamento de Matem\'atica, Instituto Superior T\'ecnico,\\ Av. Rovisco Pais 1, 1049--001 Lisboa, Portugal\\}
\emailAdd{schiappa@math.ist.utl.pt}

\author[d]{Marcel~Vonk}
\affiliation[d]{Institute for Theoretical Physics, University of Amsterdam,\\ Science Park 904, 1090--GL Amsterdam, The Netherlands\\}
\emailAdd{m.l.vonk@uva.nl}

\abstract{
The gauge theoretic large $N$ expansion yields an asymptotic series which requires a nonperturbative completion in order to be well defined. Recently, within the context of random matrix models, it was shown how to build resurgent transseries solutions encoding the full nonperturbative information beyond the 't~Hooft genus expansion. On the other hand, via large $N$ duality, random matrix models may be holographically described by B--model closed topological strings in local Calabi--Yau geometries. This raises the question of constructing the corresponding holographically dual resurgent transseries, tantamount to nonperturbative topological string theory. This paper addresses this point by showing how to construct resurgent transseries solutions to the holomorphic anomaly equations. These solutions are built upon (generalized) multi--instanton sectors, where the instanton actions are holomorphic. The asymptotic expansions around the multi--instanton sectors have both holomorphic and anti--holomorphic dependence, may allow for resonance, and their structure is completely fixed by the holomorphic anomaly equations in terms of specific polynomials multiplied by exponential factors and up to the holomorphic ambiguities---which generalizes the known perturbative structure to the full transseries. In particular, the anti--holomorphic dependence has a somewhat universal character. Furthermore, in the nonperturbative sectors, holomorphic ambiguities may be fixed at conifold points. This construction shows the nonperturbative integrability of the holomorphic anomaly equations, and sets the ground to start addressing large--order analysis and resurgent nonperturbative completions within closed topological string theory.
}


\arxivnumber{1308.1695}

\begin{document}

\maketitle

\vfill

\eject

\allowdisplaybreaks

\section{Introduction and Summary}

For over a decade, the idea of large $N$ duality \cite{m97} has been a remarkable source of progress within the nonperturbative development of strongly coupled gauge theoretic systems. In this framework one starts off with the partition function of some nonabelian gauge theory, $Z$, whose 't~Hooft large $N$ limit \cite{th74} produces an \textit{asymptotic} expansion, in $1/N$. The topological structure of the $1/N$ expansion is that of a genus expansion, which leads to the holographically dual closed string theory. In this way, one can think of closed string theory as an asymptotic (large $N$) approximation to some nonabelian gauge theory. This idea has been applied with success in very large classes of examples, see, \textit{e.g.}, \cite{agmoo99} for an older review. In the following we shall consider one of the simplest classes of nonabelian gauge theories, namely, random matrix models \cite{fgz93, m04}.

Due to its aforementioned asymptotic nature, the large $N$ expansion by itself is not enough to define either the gauge theory or the dual string theory. Indeed, in the large $N$ limit, the gauge theoretic free energy $F = \log Z$ has a genus expansion in $1/N^2$ where its genus $g$ perturbative contributions\footnote{In here $t = N g_s$ is the 't~Hooft coupling. When considering the closed string dual, $t$ will instead be part of the geometric moduli associated to the background Calabi--Yau geometry---but more on this in the following.} $F_g (t)$ have large--order behavior $F_g \sim (2g)!$, rendering the large $N$ expansion as an asymptotic expansion with zero radius of convergence\footnote{The situation is different at \textit{fixed} genus: although the free energy $F_g (t)$ may also be computed perturbatively---now in the 't~Hooft coupling---one finds a milder expansion with finite (non--zero) convergence radius \cite{knn77, th82}.} \cite{s90}---which in the string theoretic context is not even Borel summable \cite{gp88}. In order to go beyond this state of affairs, one has to take into account all possible nonperturbative corrections to the large $N$ expansion, typically of the form $\sim \exp \left( - N \right)$, into what is known as a transseries expansion; see, \textit{e.g.}, \cite{cnp93, ss03, e0801}. Because they include all possible nonperturbative corrections, transseries precisely incorporate Stokes phenomena and, as such, allow us to go beyond the large $N$ expansion obtaining, upon median resummation \cite{as13}, analytic results anywhere in the complex $N$--plane. Within the large $N$ context, explicit examples have been constructed addressing random matrix models and their double--scaling limits \cite{m08, gikm10, asv11, sv13}. One interesting point found in these analyses is that while one could have \textit{na\"\i vely} expected that all the large $N$ nonperturbative content within the transseries would be associated to multi--instantons, this expectation was actually shown to be \textit{incomplete}: one further needs to introduce \textit{new} nonperturbative sectors to fully describe the large $N$ gauge theory beyond its genus expansion. In fact, an important aspect of these transseries solutions is their resurgent nature; see, \textit{e.g.}, \cite{asv11}. That these transseries are resurgent means that, deep in their large--order behavior, every (generalized) multi--instanton sector is related to each other. Verifying the validity of any new nonperturbative sector may be thus thoroughly checked in many examples by exploring the information encoded in the large--order behavior of perturbation theory around any chosen multi--instanton sector \cite{gikm10, asv11, sv13}. Once this is done, the final resurgent transseries solution may be constructed yielding results valid for any $N$.

Our concern in the present work is that most studies carried through so far on the resurgent nature of the large $N$ limit have precisely been done in the gauge theory side, starting off with random matrix integrals \cite{m06, msw07, m08, msw08, gm08, gikm10, kmr10, asv11, sv13}. But would it be possible to address these issues \textit{directly} in the closed string side? If so, what would this imply towards the very nature of the large $N$ holographic duality, from a nonperturbative point--of--view?

We shall be interested in the rather complete picture of \cite{dv02, dv02a}, which addresses the 't~Hooft large $N$ limit of general matrix models. Let us start off on the gauge theoretic side with a hermitian one--matrix model with polynomial potential $V(z)$ of degree $n+1$. In this case it is rather well--known that, in the planar limit, saddle--points of the matrix integral are described by spectral curves \cite{bipz78}. Further, the 't~Hooft $1/N$ asymptotic expansion follows via a recursive calculation on this curve \cite{biz80, ackm93, eo07}. The key point of \cite{dv02} is to notice that the spectral curve is actually part of a Calabi--Yau threefold, and that the 't~Hooft expansion can actually be computed holographically by B--model topological string theory on this specific local geometry.

Let us be a bit more precise, as this will also help explain our main theme. Consider the $n$ critical points of the potential $V'(z^*)=0$ and distribute the $N$ eigenvalues across these critical points, $\{ N_i \}_{i=1}^n$. This saddle configuration is then described by a hyperelliptic curve \cite{bipz78, eo07}
\be\label{hyper}
y^2 = V'(x)^2 - t R(x),
\ee
\noindent
where $t=g_s N$ is the 't~Hooft coupling, and $R(x)$ is a polynomial of degree $n-1$ which effectively opens up the critical points $z^*_i$ into cuts A$_i$. Different eigenvalue partitions, corresponding to distinct classical vacua of the gauge theory, lead to geometrically distinct spectral curves. Conversely, given the curve (\ref{hyper}) all relevant information follows; \textit{e.g.}, the number of eigenvalues in each cut is given by the partial 't~Hooft couplings $t_i = g_s N_i$ as the A--periods
\be\label{hyperAcycle}
t_i = \frac{1}{4\pi\rmi} \oint_{\mathrm{A}_i} \rmd x\, y(x),
\ee
\noindent
while the tree--level (genus zero) free energy follows from the integration of the spectral curve along the dual B--cycles\footnote{This relation expresses the change in the planar free energy as we vary the eigenvalues around the cuts---in this case as we remove an eigenvalue from the cut A$_i$ all the way up to infinity. The energy cost of tunneling an eigenvalue from the cut A$_i$ to the cut A$_j$ is given by the instanton action $A = \int_{\mathrm{B}_{ij}} \rmd x\, y(x)$ \cite{d91, d92, msw07, msw08}.}
\be\label{hyperBcycle}
\frac{\partial F_0}{\partial t_i} = \int_{\mathrm{B}_i} \rmd x\, y(x).
\ee
\noindent
Further, the topological recursion \cite{eo07} algorithmically computes all higher--genus contributions in the 't~Hooft expansion. At the conceptual level one may now ask: why is the matrix model solved by a geometrical construction? As explained in \cite{dv02}, this can be best understood by first reinterpreting the spectral curve solution of the matrix model as being described by a Calabi--Yau geometry in disguise, and not just a Riemann surface; \textit{i.e.}, (\ref{hyper}) is actually a subset of the local (non--compact) Calabi--Yau threefold
\be\label{Xdef}
X_{\text{def}} = \left\{ u,v \in \BC,\, x,y \in \BC \,|\, u v + \CH (x,y) = 0 \right\},
\ee
\noindent
where $\CH (x,y)$ is the polynomial whose zero locus precisely yields the hyperelliptic curve (\ref{hyper}). In this way, the spectral curve clearly encodes all the non--trivial information about this Calabi--Yau geometry. Its A--cycles are projections of the Calabi--Yau compact three--cycles; the same holding for the non--compact canonically conjugated B--cycles. The one--form $y\, \rmd x$ is a projection of the holomorphic three--form $\Omega$, such that one may write (\ref{hyperAcycle}--\ref{hyperBcycle}) instead as threefold periods\footnote{In particular identifying Calabi--Yau periods with 't~Hooft moduli.}
\be\label{special}
t_i = \frac{1}{4\pi\rmi} \oint_{\widehat{\mathrm{A}}_i} \Omega, \qquad
\frac{\partial F_0}{\partial t_i} = \int_{\widehat{\mathrm{B}}_i} \Omega. 
\ee
\noindent
These are the special geometry relations of the Calabi--Yau (with $F_0$ the prepotential). In other words, the special geometry of $X_{\text{def}}$, which solves the tree--level closed topological B--model on this background, further yields the planar solution to the hermitian one--matrix model.

So why does this geometry naturally emerge from the matrix eigenvalue dynamics? Geometrically, the local Calabi--Yau (\ref{Xdef}) is a \textit{deformation} of the singular geometry $uv + y^2 - V'(x)^2 = 0$. One can now understand the nature of the geometric solution to the hermitian one--matrix model by thinking about the \textit{resolved} geometry $X_{\text{res}}$ instead, obtained by blow--up of the aforementioned singular geometry (see, \textit{e.g.}, \cite{dv02, m04} for the precise transition functions, which explicitly depend upon $V(z)$). As it turns out, the string field theory of \textit{open} topological B--strings on $X_{\text{res}}$ \textit{localizes} into the hermitian one--matrix model with potential $V(z)$ \cite{dv02, m04}. The emergence of the special geometry (\ref{special}) can now be understood as the result of the geometric transition $X_{\text{res}} \to X_{\text{sing}} \to X_{\text{def}}$ \cite{r04} implementing a large $N$ duality \cite{gv98b}. In this way, the 't~Hooft resummation of the matrix model precisely matches the closed topological B--model on $X_{\text{def}}$.

To be fully precise, the results in \cite{dv02} only show that closed topological strings on Calabi--Yau threefolds are large $N$ solutions to matrix models at the \textit{planar} level. An extension of this derivation to higher genera was later presented in \cite{emo07}. Generically, one may compute the topological B--model free energy using the holomorphic anomaly equations \cite{bcov93}, which control the $\bar{t}_i$--dependence\footnote{From this closed string perspective, both $t_i$ and $\bar{t}_i$ are now geometric moduli associated to the ambient Calabi--Yau geometry. We shall be more precise about their exact definition in the main body of the paper.} of the closed string amplitudes $F_g ( t_i, \bar{t}_i)$ and, when combined with extra boundary conditions, turn out to be completely integrable for non--compact Calabi--Yau manifolds \cite{hkr08, alm08}. What \cite{emo07} shows is that the holomorphic limit of these B--model closed string amplitudes $F_g ( t_i ) \equiv \lim_{\bar{t}_i \to +\infty} F_g ( t_i, \bar{t}_i)$ exactly matches the components of the large $N$ 't~Hooft expansion as computed within the matrix model using the topological recursion \cite{eo07}. Conversely, the matrix model free energy at genus $g$, may be extended non--holomorphically (by imposing adequate requirements of modular invariance and modifying the topological recursion accordingly) in such a way as to satisfy the holomorphic anomaly equations on the corresponding local Calabi--Yau. In this way, \cite{emo07} sets up the complete perturbative equivalence of \cite{dv02}. Setting up the \textit{nonperturbative} equivalence between large $N$ expansions of hermitian one--matrix models and closed B--model strings on local Calabi--Yau threefolds is, of course, much harder, as one immediately stumbles upon the problems of first defining exactly what one means by either nonperturbative topological strings or the nonperturbative 't~Hooft expansion of a matrix model. This is one of the main themes in the present paper, following previous work in \cite{m06, msw07, m08, msw08, em08, gm08, ps09, mpp09, gikm10, kmr10, asv11, sv13, as13} (a more detailed description of these recent developments may be found either briefly in the introduction of \cite{asv11}, or, in greater detail, in the excellent review \cite{m10}).

One may now return to our motivating question: can one address the issues of resurgent analysis and transseries solutions \textit{directly} in the closed string framework? As should be clear from the previous paragraph, answering this question entails the construction of resurgent transseries solutions to the holomorphic anomaly equations along with verifying resurgent large--order relations \textit{directly} in the closed string sector, and this is what we shall address throughout this work. It is important to stress that, over the years, there have been many other proposals to define nonperturbative topological strings, \textit{e.g}, \cite{osv04, ajs05, cdv10, lv12, hmmo13}. Yet, in one way or another, at the end of the day all these proposals have to rely upon large $N$ duality. As was hopefully made clear above, our main goal in this work is to get rid of this requirement by addressing this question strictly within closed string theory. In this way, besides eventually validating the nonperturbative duality to random matrix models we described above, our approach may also serve the purpose of validating, within closed strings, the other aforementioned proposals.

Schematically, we may represent the motivation for our research in the following diagram:
\be\label{antiholoinstaction}
\begin{tikzpicture}
  \node (Fttb) {$F_g (t_i, \bar{t}_i)$};
  \node (Ft) [right of=Fttb] {$F_g (t_i)$};
  \node (Attb) [below of=Fttb] {$A (t_i, \bar{t}_i)$};
  \node (At) [below of=Ft] {$A (t_i)$};
  \draw[->] (Fttb) to node {} (Ft);
  \draw[->, dashed] (Attb) to node {} (Fttb);
  \draw[->, dashed] (Attb) to node {} (At);
  \draw[->] (At) to node {} (Ft);
\end{tikzpicture}
\ee
\noindent
In words, we know that the perturbative closed string amplitudes $F_g ( t_i, \bar{t}_i)$ have a natural holomorphic limit and that, in this limit, the (leading) large--order growth of the matrix model amplitudes $F_g ( t_i )$ is controlled by a (holomorphic) instanton action, $A (t_i)$ \cite{msw07}. In this way, we might be tempted to conclude that, on the closed string side, there should be an instanton action $A (t_i, \bar{t}_i)$, with both holomorphic and anti--holomorphic dependence, controlling the (leading) large--order growth of the $F_g ( t_i, \bar{t}_i)$ amplitudes. As we shall show in this paper, it turns out that this expectation has \textit{too much} structure; the real answer is much simpler:
\be\label{justholoinstaction}
\begin{tikzpicture}
  \node (Fttb) {$F_g (t_i, \bar{t}_i)$};
  \node (Ft) [right of=Fttb] {$F_g (t_i)$};
  \node (At) [below of=Ft] {$A (t_i)$};
  \draw[->] (Fttb) to node {} (Ft);
  \draw[->] (At) to node {} (Fttb);
  \draw[->] (At) to node {} (Ft);
\end{tikzpicture}
\ee
\noindent
In fact, the holomorphic anomaly equations will precisely demand that, in the closed string sector, the instanton action is \textit{holomorphic}. Of course as we move towards the nonperturbative multi--instanton sectors, the corresponding amplitudes will have both holomorphic and anti--holomorphic dependence, making the resurgent behavior of the closed string amplitudes naturally different from that of the matrix model amplitudes. But it will still be overall much simpler than one might have expected at the beginning and we shall discuss this in detail later on. In fact, one might say that the instanton action itself is part of the holomorphic ambiguities, and that all ``non--trivial'' large--order behavior is thus encoded in these ambiguities. What exactly we mean by this will also be made clear as we develop our analysis in the main body of the text.

This paper is organized as follows. We begin by discussing the holomorphic anomaly equations in section \ref{sec:holoeqs}. These equations are usually written recursively for the genus $g$ free energy, but we shall see how they may be simply rewritten so that they become equations for the full free energy and thus equations allowing for a transseries \textit{ansatz}. In order to be self--contained, we briefly introduce transseries and resurgence in section \ref{sec:transseriesintro}. Having these two pieces of information at hand, one may then proceed and solve the holomorphic anomaly equations with transseries, which we do in section \ref{sec:1parameter} in the context of so--called one--parameter transseries. This is the simplest setting and we thus address it first, further motivating what one expects to be the structure of the general solution. This general solution is then constructed in section \ref{sec:multiparameter} for multi--parameter resurgent transseries, also including cases where the transseries might be resonant---and we discuss such features and such implications. The relation to large--order analysis (and resurgence) is then explored in section \ref{sec:largeorder}. On the one hand, this gives us strong perturbative support for our nonperturbative construction. On the other hand, part of these results are fundamental to explain how to fix the remaining holomorphic ambiguities of our solutions, which we do in section \ref{sec:fixingambiguities}---thus concluding the construction of general nonperturbative transseries solutions to the holomorphic anomaly equations. We end with an outlook towards specific examples and other open problems. In a couple of appendices we include the generalizations of our results in section \ref{sec:multiparameter} to multi--dimensional complex moduli spaces, and a proof of a Lemma used in the main text.

\section{The Holomorphic Anomaly Equations}\label{sec:holoeqs}

Let us begin by setting the stage and briefly reviewing how closed B--model topological string theory is solved by the holomorphic anomaly equations \cite{bcov93b, bcov93, o94}. In the process we shall also very briefly review topological strings in local Calabi--Yau geometries, but we refer the reader to one of the many excellent reviews on the subject for further details \cite{m05, m04, nv04, v05, a12}. After reviewing the basics, we will discuss how to adapt the holomorphic anomaly equations---usually written in a suitable fashion for perturbative calculations---to the present problem which asks for nonperturbative results.

\subsection{Reviewing the Background Calabi--Yau Geometry}

Closed type B topological string theory describes constant maps from Riemann surfaces of genus $g$, $\Sigma_g$, to some Calabi--Yau threefold, $X$. The theory is topological on the world--sheet, but it still depends upon the complex structure of the Calabi--Yau, parametrized by moduli $\{ z_i \}_{i=1}^{h^{2,1}}$. While at first one might have expected that the genus $g$ free energies were holomorphic functions of the complex moduli as $F_g \equiv F_g (z_i)$, it was shown in \cite{bcov93b, bcov93} that, instead, there is an anomaly in the holomorphic dependence of the free energies which arises when one couples the B--model to gravity. In other words, it is actually the case that $F_g \equiv F_g (z_i, \bar{z}_i)$. This complex conjugate dependence of the free energies was derived in \cite{bcov93}, by studying contributions to $F_g$ which arise from boundary contributions to the integral over the moduli space of maps from Riemann surfaces of genus $g$ to the Calabi--Yau manifold $X$. In particular, the authors of \cite{bcov93} obtained a set of equations encoding this anti--holomorphic dependence, known as the holomorphic anomaly equations:
\be
\label{eq:ttbarholanomeqs}
\partial_{\bar{i}} F_g = \frac{1}{2} {\bar{C}_{\bar{i}}}^{\ jk} \left( D_j D_k F_{g-1} + \sum_{h=1}^{g-1} D_j F_{g-h} D_k F_h \right), \qquad g\geq 2.
\ee
\noindent
In here $\partial_{\bar{i}} = \frac{\partial}{\partial \bar{z}_i}$, while the tensor ${\bar{C}_{\bar{i}}}^{\ jk}$ and the covariant derivative $D_j$ will be explained below. The holomorphic anomaly equation at genus $g=1$ was given in \cite{bcov93b}, while the one for genus $g=0$ is trivial since the prepotential, $F_0$, is holomorphic. Together, these equations allow for a recursive integration of the genus $g$ free energies, as the right--hand--side of \eqref{eq:ttbarholanomeqs} only depends on previous genera, already calculated in the recursion process. The outcome of the integration is the desired free energy, $F_g (z_i, \bar{z}_i)$, which can be determined up to a holomorphic \textit{integration constant}, $f_g (z_i)$, the holomorphic ambiguity. We shall discuss how this is fixed in section \ref{sec:fixingambiguities}.

Let us focus on these equations, \eqref{eq:ttbarholanomeqs}. They are valid for both compact and non--compact Calabi--Yau manifolds, but we shall only be interested in the latter. In this case, the free energies are functions (sections) on the complex moduli space of the Calabi--Yau manifold. This moduli space is actually K\"ahler, with metric
\be
G_{i\bar{j}} = \partial_i \partial_{\bar{j}} K, \qquad K = \log \rmi \int_X \Omega \wedge \bar{\Omega},
\ee
\noindent
where $K$ is the K\"ahler potential and $\Omega$ is the non-vanishing $(3,0)$--form on $X$ which determines the complex structure. With this structure, the covariant derivative $D_j$ is given by
\bea
D_j F_g &=& \partial_j F_g, \\
D_j D_k F_g &=& \partial_j \partial_k F_g - \Gamma^i_{jk}\, \partial_i F_g - \left( 2 - 2g \right) \partial_j K\, \partial_k F_g,
\eea
\noindent
where $\Gamma^i_{jk} = G^{i\bar{i}}\partial_j G_{k\bar{i}}$ is the Chistoffel symbol, and where $\partial_j K$ acts as the connection for the line bundle $\mathcal{L}^{2-2g}$ associated to the multiplicative freedom in the normalization of $\Omega$. Finally, the last ingredient is
\be
{\bar{C}_{\bar{i}}}^{\ jk} = \rme^{2K}\, \bar{C}_{\bar{i}\bar{j}\bar{k}}\, G^{j\bar{j}}\, G^{k\bar{k}},
\ee
\noindent
where $\bar{C}_{\bar{i}\bar{j}\bar{k}} = \overline{C_{ijk}}$ is the complex conjugate of the Yukawa couplings, which are given by the third holomorphic derivatives of the prepotential, $C_{ijk} = \partial_i \partial_j \partial_k F_0$.

It turns out that the integration of the equations \eqref{eq:ttbarholanomeqs} is quite complicated as they stand. The calculation may be made easier with the introduction of propagators, which essentially are potentials for the $\bar{C}_{\bar{i}\bar{j}\bar{k}}$ tensors thus simplifying the integration. These propagators were already introduced in \cite{bcov93}, and their interest became explicit in \cite{yy04, al07}. In fact, in \cite{yy04} it was shown that, in the example of the local quintic, the anti--holomorphic dependence of the free energies at every genus may be captured by a finite number of generators, in such a way that this dependence is polynomial. In \cite{al07} this was extended to general compact Calabi--Yau geometries. There, it was shown that if one defines the propagators $S^{ij}$, $S^i$, $S$, by
\be
\label{eq:propagatordefinitions}
\partial_{\bar{i}} S^{jk} = {\bar{C}_{\bar{i}}}^{\ jk}, \qquad \partial_{\bar{i}} S^j = G_{i\bar{i}}\, S^{ij}, \qquad \partial_{\bar{i}} S = G_{i\bar{i}}\, S^i,
\ee
\noindent
akin to \cite{bcov93}, then one can show that any covariant derivative of these propagators, $S^{ij}$, $S^i$, $S$, and also of $\partial_i K$, may be expressed back in terms of these objects and the Yukawa couplings, $C_{ijk}$. For example,
\be
D_i S^{jk} = \delta^j_i S^k + \delta^k_i S^j - C_{i\ell m}S^{\ell j}S^{mk} + f^{jk}_i,
\ee
\noindent
where $f^{jk}_i$ is a holomorphic function related to the fact that $S^{jk}$ in  \eqref{eq:propagatordefinitions} is only defined up to a holomorphic quantity.  As we turn to the case of non--compact geometries, one can consistently turn off\footnote{Strictly speaking this is only true in the holomorphic limit.} the propagators $S^i$, $S$, $\partial_i K$, and work only with $S^{jk}$ \cite{alm08}. The Christoffel symbols are also expressable in terms of the propagators. Restricting to this simplified scenario, one finds:
\be
\Gamma^i_{jk} = - C_{jk\ell} S^{\ell i} + \tilde{f}^i_{jk},
\ee 
\noindent
where $\tilde{f}^i_{jk}$ is another holomorphic function. Both $f^{jk}_i$ and $\tilde{f}^i_{jk}$ can be determined within specific examples, partly by choice and partly by internal consistency of the equations.

The main idea behind the above discussion is that one may now use the propagators $S^{ij}$ to take up the role of the anti--holomorphic variables $\bar{z}_i$. Using the chain rule,
\be
\partial_{\bar{i}} = \frac{\partial S^{jk}}{\partial \bar{z}_i}\, \frac{\partial}{\partial S^{jk}} = {\bar{C}_{\bar{i}}}^{\ jk}\, \frac{\partial}{\partial S^{jk}}.
\ee
\noindent
It follows that the holomorphic anomaly equations \eqref{eq:ttbarholanomeqs} become
\be
\label{eq:propagatorsholanomeqs}
\frac{\partial F_g}{\partial S^{ij}} = \frac{1}{2} \left( D_i D_j F_{g-1} + \sum_{h=1}^{g-1} D_i F_{g-h} \, D_j F_h \right), \qquad g\geq 2,
\ee
\noindent
where we have been able to get rid of ${\bar{C}_{\bar{i}}}^{\ jk}$. This will be our starting point. Later, in section \ref{sec:1parameter}, we shall see how these equations are integrated and what is the structure of the resulting solution with respect to the propagators. Let us just mention for the moment that because $D_j F_1 = \partial_j F_1$ is \textit{linear} in the propagators, the recursion \eqref{eq:propagatorsholanomeqs} yields a polynomial dependence in the $S^{ij}$. Integration also produces a holomorphic ambiguity, whose fixing is discussed in section \ref{sec:fixingambiguities}.

Before concluding this subsection, let us point out that the recursion \eqref{eq:propagatorsholanomeqs} starts at $g=2$. The free energies $F_0$ and $F_1$ are of a different sort and they are handled separately. On what concerns $F_0$, the prepotential, it carries strong geometrical information. The derivative $\partial_i F_0$ is actually a period, \textit{i.e.}, a cycle integral of the $(3,0)$--form $\Omega$, and taking two more derivatives produces the Yukawa couplings. Also note that $F_0$ does not appear explicitly in the holomorphic anomaly equations, except through the Yukawa couplings. As to $F_1$, it appears in \eqref{eq:propagatorsholanomeqs} through derivatives and it has a geometrical interpretation in terms of Ray--Singer torsion. All we will need later on is that
\be
\partial_i F_1 = \frac{1}{2} C_{ijk} S^{jk} + \alpha_i,
\ee
\noindent
where $\alpha_i$ is a holomorphic quantity that can be fixed for each example.

\subsection{Rewriting the Holomorphic Anomaly Equations}

As they stand, the holomorphic anomaly equations essentially yield a recursive procedure to compute the \textit{perturbative} free energy. Indeed, they are equations precisely for the perturbative $F_g$. But, in order to address nonperturbative issues, one would need to write these equations either for the full free energy, $F$, or even the partition function $Z = \exp F$. That equivalent formulations of the holomorphic anomaly equations may be found for $F$ or $Z$ was already anticipated early on \cite{bcov93b}, as a \textit{master equation} for the partition function. Shortly before \cite{bcov93}, the following explicit master equation was suggested in \cite{w93}, for $Z = \exp \sum_{g=0}^{+\infty} g_s^{2g-2} F_g$,
\be\label{eq:Witteneq}
\left( \partial_{\bar{i}} - \frac{1}{2} g_s^2\, {\bar{C}_{\bar{i}}}^{\ jk} D_j D_k \right) Z = 0.	
\ee
\noindent
Unfortunately, as it stands, this equation does not quite reproduce the holomorphic anomaly equations as the upper and lower summation  bounds in the resulting quadratic part are not as in \eqref{eq:propagatorsholanomeqs}. Recall that the proper limits, from $h=1$ to $h=g-1$, are necessary in order to have a recursion in $g$. In \cite{bcov93} another version of a master equation was considered, this time for $\mathcal{F} \simeq \sum_{g=1}^{+\infty} g_s^{2g-2} F_g$ (note that the sum starts at $g=1$, so that $F_0$ is not included). This equation is
\be\label{eq:BCOVmastereq}
\left( \partial_{\bar{i}} - \partial_{\bar{i}} F_1 \right) \rme^{\mathcal{F}} = \frac{1}{2} g_s^2\, {\bar{C}_{\bar{i}}}^{\ jk} D_j D_k \rme^{\mathcal{F}},
\ee
\noindent
which now does give back the holomorphic anomaly equations for $g\geq 2$. As before, $F_0$ has no direct role in the equations other than through the Yukawa couplings. Later on, in \cite{emo07}, an equivalent version of the holomorphic anomaly master equation was discussed which is closer to what we are interested in. Defining
\be
Z \simeq \exp \sum_{g=0}^{+\infty} g_s^{2g-2} F_g, \qquad \widetilde{Z} = \rme^{-\frac{1}{g_s^{2}} F_0 - F_1}\, Z, \qquad \widehat{Z} = \rme^{- \frac{1}{g_s^{2}} F_0}\,Z,
\ee
\noindent
it is possible to show that the holomorphic anomaly equations may be written in the form\footnote{\cite{emo07} works in the language of matrix models, where $g_s$ is identified with $N^{-1}$.}
\be\label{eq:EMOeq}
\frac{1}{\widetilde{Z}}\, \partial_{\bar{i}} \widetilde{Z} = \frac{1}{2} g_s^2\, \frac{1}{\widehat{Z}}\, {\bar{C}_{\bar{i}}}^{\ jk} D_j D_k \widehat{Z}.	
\ee
\noindent
It is noteworthy that these master equations resemble generalized versions of the heat equation, and this has been explored in \cite{w93, gnp06, nw07} along with the possibility that the partition function is actually described by a Riemann theta function.

Our goal here is to write a master equation for the partition function $Z \simeq \exp \sum_{g=0}^{+\infty} g_s^{2g-2}F_g$, including both $F_0$ and $F_1$, with the property that it will yield the tower of holomorphic anomaly equations for genus $g \ge 2$ in the spirit of the above examples. Let us emphasize that such an equation should be written in terms of $Z$ alone, so that later on one is able to promote $Z$ to a fully nonperturbative partition function. This is in contrast with, \textit{e.g.}, \eqref{eq:EMOeq} above, which further involves the functions $\widetilde{Z}$ and $\widehat{Z}$. The price of working only with $Z$ will be that both $F_0$ and $F_1$ will explicitly appear in the resulting equation. However, this is not a problem since one should regard $F_0$ and $F_1$ as geometrical data associated to the specific model under study. In this sense, let us stress that \eqref{eq:EMOeq} is naturally \textit{already} the equation we are looking for; we just need to rewrite it in a more convenient way.

Before deriving the appropriate form of the master equation, we shall switch to propagator variables for the anti--holomorphic dependence. Now, let us first state the master equation, and then show that it reproduces the familiar holomorphic anomaly equations \eqref{eq:propagatorsholanomeqs}. The master equation is:
\be\label{eq:UVWZeq}
\left( \frac{\partial}{\partial S^{ij}} + \frac{1}{2} \left( U_i D_j + U_j D_i \right) - \frac{1}{2} g_s^2\, D_i D_j \right) Z = \left( \frac{1}{g_s^{2}}\, W_{ij} + V_{ij} \right) Z.
\ee
\noindent
Here $U_i$, $V_{ij}$, and $W_{ij}$ are functions involving $F_0$ and $F_1$ which ensure that we will obtain the correct holomorphic anomaly equations and nothing else. They can be thought of as geometrical data. Switching from the partition function, $Z$, to the free energy, $F= \log Z$, one may write \eqref{eq:UVWZeq} as
\be\label{eq:UVWFeq}
\frac{\partial F}{\partial S^{ij}} + \frac{1}{2} \left( U_i D_j F + U_j D_i F \right) - \frac{1}{2} g_s^2 \left( D_i D_j F+ D_i F D_j F \right) = \frac{1}{g_s^{2}}\, W_{ij}+V_{ij}.
\ee
\noindent
Let us now use a perturbative \textit{ansatz} for the free energy
\be
\label{perturbativeF}
F_{\text{pert}} \equiv F^{(0)} \simeq \sum_{g=0}^{+\infty} g_s^{2g-2} F^{(0)}_{g}
\ee
\noindent
and find out what are the values of $U_i$, $V_{ij}$, $W_{ij}$ which yield back \eqref{eq:propagatorsholanomeqs}. We have explicitly added a superscript $(0)$ to stress that the free energies are perturbative, anticipating that we shall soon work with nonperturbative corrections. The quadratic term $D_i F^{(0)}D_j F^{(0)}$ yields
\be
D_i F^{(0)} D_j F^{(0)} = \sum_{g=0}^{+\infty} g_s^{2g-4} \sum_{h=0}^g D_i F^{(0)}_{g-h} D_j F^{(0)}_h,
\ee
\noindent
where the limits in the $h$--sum are not quite the ones we need. It will thus be the role of $U_i$ to remove the first and last terms in the sum. Doing the rest of the calculation and assembling equal powers of $g_s$ together, we find
\bea
&& \frac{1}{g_s^{2}} \left( \frac{\partial F^{(0)}_0}{\partial S^{ij}} + \frac{1}{2} U_i D_j F^{(0)}_0 + \frac{1}{2} U_j D_i F^{(0)}_0 - \frac{1}{2} D_i F^{(0)}_0 D_j F^{(0)}_0 \right) + \nonumber \\
&+&
g_s^0 \left( \frac{\partial F^{(0)}_1}{\partial S^{ij}} + \frac{1}{2} U_i D_j F^{(0)}_1 + \frac{1}{2} U_j D_i F^{(0)}_1 - \frac{1}{2} D_i D_j F^{(0)}_0 - \frac{1}{2} D_i F^{(0)}_0 D_j F^{(0)}_1 - \frac{1}{2} D_i F^{(0)}_1 D_j F^{(0)}_0 \right) + \nonumber\\
&+& \sum_{g=2}^{+\infty} g_s^{2g-2} \left( \frac{\partial F^{(0)}_g}{\partial S^{ij}} + \frac{1}{2} U_i D_j F^{(0)}_g + \frac{1}{2} U_j D_i F^{(0)}_g - \frac{1}{2} D_i D_j F^{(0)}_{g-1} - \frac{1}{2} \sum_{h=0}^g D_i F^{(0)}_h D_j F^{(0)}_{g-h} \right) = \nonumber \\
&=& \frac{1}{g_s^{2}}\, W_{ij} + g_s^0\, V_{ij}.
\eea
\noindent
The terms at order $g_s^{-2}$ and $g_s^0$ dictate the values of $W_{ij}$ and $V_{ij}$, respectively; and we fix $U_i$ by imposing that the $h=0$ and $h=g$ are removed in the $g_s^{2g-2}$ term, for each $g\geq 2$. It finally follows 
\bea
\label{eq:Uconstraint}
U_i &=& D_i F^{(0)}_0, \\
\label{eq:Vconstraint}
V_{ij} &=& \frac{\partial F^{(0)}_1}{\partial S^{ij}} - \frac{1}{2}D_i D_j F^{(0)}_0, \\
\label{eq:Wconstraint}
W_{ij} &=& \frac{\partial F^{(0)}_0}{\partial S^{ij}} + \frac{1}{2} D_i F^{(0)}_0 D_j F^{(0)}_0.
\eea
\noindent
These expressions imply the tower of holomorphic anomaly equations \eqref{eq:propagatorsholanomeqs}. Let us stress again that \eqref{eq:UVWZeq} is precisely the well--known master equation version of the holomorphic anomaly equations---as mentioned equivalent to \eqref{eq:EMOeq}---simply rewritten in a version  prepared to accept a nonperturbative \textit{ansatz} for the partition function.

To conclude, we specialize our version of the master equation \eqref{eq:UVWFeq} to the case when the background geometry has a complex moduli space of dimension one. This means that we only have a single holomorphic coordinate, call it $z$, and a single propagator, $S^{zz}$. In this case, the master equation for the holomorphic anomaly equations has the form
\be\label{eq:1parameterUVWeq}
\partial_{S^{zz}} F + U\, D_z F - \frac{1}{2} g_s^2 \left( D_z D_z F + D_z F D_z F \right) = \frac{1}{g_s^{2}}\, W + V,
\ee
\noindent
where we have removed unnecessary indices in $U$, $V$ and $W$.

\section{A Word on the Transseries Framework}\label{sec:transseriesintro}

The concept of a \textit{resurgent transseries} plays a central role in this paper. As such, in an effort to be self--contained, we briefly review it in this section. We follow a practical approach: we do not give the most general definitions possible, but rather introduce transseries based on examples, with our later applications in mind. The reader who wants to know more about the general theory of transseries can find a good starting point in \cite{e0801, asv11}.

In physics, one often encounters perturbative series of the general form
\be
F(x) \simeq \sum_{g=0}^{+\infty} F_g \, x^g,
\label{eq:fseries}
\ee
\noindent
with $F$ some physical quantity and $x$ a (perturbative) coupling. In many examples, the quantity $F(x)$ is known to satisfy some (non--linear) equation although solving such equation exactly and finding a function $F(x)$ is not possible. What one often \textit{can} do instead is to solve the relevant equation order by order in $x$, and in this way inductively find the perturbative coefficients $F_g$.

As is well known, the perturbative coefficients $F_g$ do not always tell the full story about a problem. A famous example is the function
\be
F(x) = \rme^{-A/x} \qquad (x \geq 0),
\ee
\noindent
which is clearly a non--trivial function, yet at $x=0$ it has a Taylor expansion of the form \eqref{eq:fseries} with $F_g=0$ for all $g$. Such ``nonperturbative'' functions appear frequently in physics; for example, instanton contributions in quantum field theories are of this form, where the parameter $x$ is Planck's constant $\hbar$ and $A$ is the instanton action. This example motivates the introduction of more general objects, beyond power series, called \textit{transseries}. A transseries in the variable $x$ is a formal series in several ``building blocks'', each of which is a function of $x$. The canonical example to have in mind is a transseries of the form
\be
F(x,\sigma) = \sum_{n=0}^{+\infty} \sum_{g=0}^{+\infty} \sigma^n\, \rme^{-nA/x}\, F^{(n)}_g\, x^g.
\label{eq:1parts}
\ee
\noindent
Here, the two building blocks are the original perturbative parameter, $x$, and the ``instanton factor'', $\rme^{-A/x}$. One also introduces an ``instanton counting parameter'' $\sigma$. This makes it very explicit how the new building block augments simple power series. In fact, the transseries is a perturbative expansion in both of these building blocks, with coefficients $F^{(n)}_g$. Just like in the ordinary power series case, one can plug such a transseries \textit{ansatz} into an equation and try to formally solve the equation order by order in $x$ and in $\rme^{-A/x}$.

The transseries parameter $\sigma$ turns out to be very useful when one constructs a transseries solution to either a differential or a finite difference equation (see \cite{m08, gikm10} for the first applications in string theory and \cite{asv11, sv13} for further developments), since in these cases $\sigma$ is often an integration constant. This means that the transseries one finds is a formal solution to the equation in question for \textit{any} value of $\sigma$. Note that changing $\sigma$ does not change the ordinary perturbative sub--series given by the coefficients $F^{(0)}_g$; the different formal transseries one finds by changing $\sigma$ only differ in their \textit{nonperturbative} content. Of course, in problems described by higher--order equations, one expects to find several integration constants. As a result, one needs \textit{multi--parameter} transseries in these problems. For example, a two--parameters transseries, with instanton counting parameters $\sigma_1$ and $\sigma_2$, could be of the form
\be
F(x,\sigma_1,\sigma_2) = \sum_{n=0}^{+\infty} \sum_{m=0}^{+\infty} \sum_{g=0}^{+\infty} \sigma_1^n \sigma_2^m\, \rme^{-\left( n A_1 + m A_2 \right)/x}\, F^{(n|m)}_g\, x^g.
\label{eq:2parts}
\ee
\noindent
Examples for which one can inductively construct such formal two--parameter transseries solutions, the Painlev\'e I and II equations and the string equation for the hermitean matrix model with quartic potential, were studied in detail in \cite{gikm10, asv11, sv13}.

One further example of the type of transseries structure we are interested in appears in situations where two of the instanton actions are opposite. For example, if we plug a transseries \textit{ansatz} of the form (\ref{eq:2parts}) into a certain class of equations, known as \textit{resonant}, we may find that $A_2 = -A_1$. This fact may actually lead to problems when one inductively solves for the transseries coefficients $F^{(n|m)}_g$, since the coefficients in the recursive equation one has to solve,  appearing in front of $F^{(n|m)}_g$, depend on linear combinations of the $A_\alpha$. As such, these coefficients may now vanish making the system unsolvable. This problem, called \textit{resonance}, was originally encountered in \cite{gikm10} and also played a crucial role in the examples in \cite{asv11, sv13}. It was found that it can be solved by introducing an extra ``building block'' into the transseries---a block of the form $\log x$. Its derivatives will lead to additional terms in the recursive equations that in general solve the resonance problem. Thus, apart from transseries of the above form, in problems with resonance one often encounters transseries of the form\footnote{In the  examples addressed in \cite{asv11, sv13}, the sum over logarithmic sectors is actually \textit{finite}.}
\be
F(x,\sigma_1,\sigma_2) = \sum_{n=0}^{+\infty} \sum_{m=0}^{+\infty} \sum_{k=0}^{k_{\text{max}}} \sum_{g=0}^{+\infty} \sigma_1^n \sigma_2^m\, \rme^{- \left( n A_1 + m A_2 \right)/x}\, \log^k \left(x\right) F^{(n|m)[k]}_g\, x^g,
\label{eq:2partslog}
\ee
\noindent
where we used the two--parameter setting as a concrete example. We refer the reader to \cite{asv11} for more details on the resonance phenomenon.

The transseries structure appearing in \eqref{eq:2partslog} is the most general form of transseries that we will discuss in the following, except for one final detail: in the above examples, all of the sums have started at index $0$. This is generally sufficient for the $n$, $m$ and $k$--sums, since one usually looks for a transseries solution that starts off with an ordinary series solution, with coefficients $F^{(0|0)[0]}_g$. However, this restriction is in general too strong for the $g$--sums (see, \textit{e.g.}, \cite{asv11, sv13}). Some leading coefficients in these sums may vanish, or it may happen that we need the $g$--sums to be Laurent series starting at a negative value of $g$. We shall use the generic variable $b$, with the appropriate indices attached, for the starting value of $g$ in a particular sector. For example, with an arbitrary starting order $b^{(n|m)[k]}$ our two--parameters example above becomes
\be
F(x,\sigma_1,\sigma_2) = \sum_{n=0}^{+\infty} \sum_{m=0}^{+\infty} \sum_{k=0}^{k_{\text{max}}} \sum_{g=0}^{+\infty} \sigma_1^n \sigma_2^m\, \rme^{- \left( n A_1 + m A_2 \right)/x}\, \log^k \left(x\right) F^{(n|m)[k]}_g\, x^{g+b^{(n|m)[k]}}.
\label{eq:2partslogb}
\ee
\noindent
This transseries structure (and its straightforward multi--parameter generalizations) will be the most general structure we address in this paper. For more general examples of transseries structures and the general theory behind those, we refer the reader to \cite{e0801}.

So far we have considered completely general, formal transseries, where the coefficients $F^{(n|m)}_g$ (we will now revert to the log--free case to keep the notation simple) can be arbitrary numbers. However, in most interesting cases transseries are \textit{resurgent}, \textit{i.e.}, the transseries coefficients have structural properties encoded in their growth\footnote{We shall say a bit more about how exactly these coefficients grow in the resurgent case in section \ref{sec:largeorder}.} as $g \to +\infty$. To understand what a resurgent transseries is, let us first note that they satisfy so--called \textit{bridge equations} (this is particularly clear when one starts off with differential or finite difference equations), such as 
\be
\dot{\Delta}_\omega F(x,\sigma_1,\sigma_2) = S_\omega \left( \sigma_1, \sigma_2 \right) \frac{\partial F}{\partial \sigma_1} + \widetilde{S}_\omega \left( \sigma_1, \sigma_2 \right) \frac{\partial F}{\partial \sigma_2},
\label{eq:bridgeeq}
\ee
\noindent
where, for concreteness, we once again chose a two--parameter example. Let us briefly explain the contents of this equation. On the left--hand--side, the so--called \textit{alien derivative} $\dot{\Delta}_\omega$ appears, implying that \eqref{eq:bridgeeq} is a ``bridge'' between alien and ordinary differential calculus. This is a differential operator mapping transseries to transseries, which plays an important role in \'Ecalle's theory of resurgence \cite{e81}. In particular, its exponential is used to calculate the discontinuity in a transseries when Stokes phenomenon occurs (see \cite{asv11, as13} for very explicit formulae). The alien derivative is non--zero at singular points $\omega$ in the Borel complex plane, which encode the nonperturbative content of a given transseries. On the right--hand--side, $S_\omega$ and $\widetilde{S}_\omega$ are the Stokes factors. These factors completely encode Stokes phenomenon, \textit{i.e.}, they carry the complete information on how to connect transseries sectorial solutions across the complex plane. Fortunately for us, neither the precise definition of the alien derivative nor the explicit form of the Stokes coefficients is important to convey the main message of the bridge equation \eqref{eq:bridgeeq}. All we need to know is that the alien derivative $\dot{\Delta}_\omega$ only acts on power series of $x$ and, as such, it does not change the powers of $\sigma_i$ appearing in the transseries. On the other hand, the $\sigma_i$--derivatives on the right--hand--side of \eqref{eq:bridgeeq} clearly \textit{do} change the powers of $\sigma_i$ appearing in a transseries. The result of this is that the bridge equation, when written out in components, relates coefficients $F^{(n|m)}_g$ to coefficients $F^{(n'|m')}_{g'}$ with \textit{different} $n'$, $m'$ and $g'$. In particular, it gives us relations between perturbative and nonperturbative coefficients (again, see \cite{asv11} for very explicit formulae).

This is also the origin of the name resurgent: at each singular point $\omega$, a specific nonperturbative sector will see the resurgence of other, different nonperturbative sectors due to the nature of the bridge equations. The relations one thus obtains between different coefficients in different sectors are so stringent that we could, in principle, determine \textit{all} nonperturbative coefficients if the perturbative coefficients and the Stokes factors are known. In section \ref{sec:largeorder}, we shall see that the resulting relations take the form of large--order relations, and discuss how this allows us to test the resurgent properties of the transseries appearing in topological string theory.

Before proceeding, let us stress one important point. In the above discussion we have implicitly assumed the existence of a (non--linear) differential or finite--difference equation in the perturbative coupling, $x$, with the transseries parameterizing families of solutions to such an equation. This is for instance the case in earlier examples dealing with matrix models (where there is a string equation) and their double--scaling limits (described by Painlev\'e--type differential equations) \cite{m08, gikm10, asv11, sv13}. From a string theoretic point--of--view, one may think of these equations as equations in $g_s$, implying one can further determine important structural data. For instance, such equations precisely compute the instanton actions, determine the number of the transseries parameters, or the starting order $b$. Let us stress, however, that this is \textit{not} the case in our problem of the topological string\footnote{Finding an equation in $g_s$ for diverse topological string backgrounds is in itself a fascinating research problem.}. Rather, the holomorphic anomaly equations are a set of consistency relations for the topological string free energy and, as such, turn out to be equations in the moduli space coordinates. It is thus quite remarkable that it is still possible to solve these equations with a transseries \textit{ansatz}, although, consequently, it will turn out that just as the (perturbative) holomorphic ambiguity is not determined, neither will we be able to determine the complete structural data for the transseries nor the associated instanton actions. Nonetheless, as we shall see, a very large number of properties may still be obtained!

Finally, note that in the derivation of the net of resurgence relations, as briefly discussed above, a bridge equation is needed and in turn this is based on the existence of some equation in the perturbative expansion parameter, $g_s$. The fact that we do not know of such an equation for the topological string does not imply that the free energy is not a resurgent function; quite to the contrary: it is in fact expected to be, based on the large accumulated experience with matrix models and their double--scaling limits.  Of course some deviations from these examples can still occur, and our strategy will be adjusted as they may appear.

\section{One--Parameter Resurgent Transseries}\label{sec:1parameter}

Our main goal in this paper is to construct transseries solutions to the holomorphic anomaly equations. In the previous two sections we have seen how to write these equations in a fashion adapted to nonperturbative solutions, and we have discussed what exactly is the structure of transseries solutions. We will now show how to construct one--parameter transseries solutions to the holomorphic anomaly equations, and how their anti--holomorphic dependence is fixed via exponentials and polynomials in the propagators both at perturbative and nonperturbative levels.

\subsection{A Closer Look at the Perturbative Sector}

Before diving into the details of a transseries solution for the topological string free energy, let us review what is the structure of the \textit{perturbative} solution, $F^{(0)}_g$, as integrated from the holomorphic anomaly equations. We are particularly interested in the anti--holomorphic moduli dependence of these free energies, \textit{i.e.}, how they depend on the propagators.

Let us start with \cite{yy04} where it was pointed out that the holomorphic anomaly equations may be efficiently integrated to yield the string perturbative expansion. Within the example of the quintic in $\mathbb{P}^4$, that work introduced a finite set of generators and proved that the genus $g$ free energies are polynomials of a certain degree in the generators. These generators captured the anti--holomorphic dependence of the $F^{(0)}_g$ and the degree of the polynomials was found to be $3g-3$. Later, in \cite{al07}, this result was extended to arbitrary compact Calabi--Yau threefolds with any number of complex moduli. In particular, it was shown that the propagators $S^{ij}$, $S^i$, $S$ and $\partial_i K$, generate a differential ring, \textit{i.e.}, any covariant derivative of a polynomial in these generators is again a polynomial in these generators; and it was found that the free energy at genus $g$ is a polynomial of total degree $3g-3$ precisely in the propagators, provided one assigns degree $+1$ to $S^{ij}$ and $\partial_i K$, degree $+2$ to $S^i$, and degree $+3$ to $S$. The same result, but for local Calabi--Yau geometries, was obtained in \cite{hkr08, alm08}. Another approach to integration based on the differential ring of quasi--modular forms was followed in \cite{gkmw07}.

Since we shall later study the structure of higher--instanton free energies $F^{(n)}_g$ in great detail, it will be useful to first derive the polynomial structure of the perturbative free energies, $F^{(0)}_g$. A nice and elegant way to show that $F^{(0)}_g$ is a polynomial of degree $3g-3$ is the one used originally in \cite{al07}. Given the holomorphic anomaly equations \eqref{eq:propagatorsholanomeqs}, one assigns a degree to each object in the following way
\be
F^{(0)}_g \mapsto d(g), \quad D_i \mapsto +1, \quad S^{ij} \mapsto +1.
\ee
\noindent
Note that this degree will coincide with the usual notion of polynomial degree in the $S^{ij}$--variables. Note also that since
\be
\label{eq:covariantderivativeonpropagator}
D_i S^{jk}  = - C_{i\ell m}\, S^{\ell j}S^{mk} + f^{jk}_i,
\ee
\noindent
it is consistent to assign degree $+1$ to $D_i$ (and zero to any holomorphic quantity like the Yukawa couplings). The idea now is to prove by induction that $d(g) = 3g-3$. The base case of the induction is simply proved by direct integration of \eqref{eq:propagatorsholanomeqs} for $g=2$, where one finds that $F^{(0)}_2$ is a polynomial in $S^{ij}$ of degree $3$. Then, assuming $d(h)=3h-3$ for $h<g$, one finds that the right--hand--side of \eqref{eq:propagatorsholanomeqs} has degree $3 \left( g-1 \right)-3+1+1=3g-4$, as it should be, thus proving the assumption that the left--hand--side has degree $3g-3-1=3g-4$.

A more pedestrian but equivalent way to obtain the same conclusion is again to prove by induction that
\be
F^{(0)}_g = \textrm{Pol} \left( \{S^{ij}\}; 3g-3 \right) \equiv \textrm{Pol} \left( 3g-3 \right), \qquad g\geq 2,
\ee
\noindent
but this time around without relying directly on the notion of a degree. Instead, one may perform the integration directly once the structure of the right--hand--side of \eqref{eq:propagatorsholanomeqs} is known. In here, $\textrm{Pol} \left( \{S^{ij}\}; d \right)$ stands for a polynomial of total degree $d$ in the propagators, and whose coefficients are functions of the holomorphic complex moduli. The base case of the induction is the same as before, proved by explicit integration. Then, assuming $F^{(0)}_h = \textrm{Pol} \left( 3h-3 \right)$ for $h<g$, and since
\bea
D_i F^{(0)}_h &=& \textrm{Pol} \left( 3h-3+1 \right), \\
D_i D_j F^{(0)}_h &=& \textrm{Pol} \left( 3h-3+2 \right),
\eea
\noindent
where we used \eqref{eq:covariantderivativeonpropagator}, it follows
\bea
\textrm{RHS of }\eqref{eq:propagatorsholanomeqs} &=& \textrm{Pol} \left( 3 \left( g-1 \right) -3+2 \right) + \sum_{h=1}^{g-1} \textrm{Pol} \left( 3 \left( g-h \right) -3+1 \right) \times \textrm{Pol} \left( 3h-3+1 \right) = \nonumber \\
&=& \textrm{Pol} \left( 3g-4 \right).
\eea
\noindent
Integration with respect to the propagators raises the polynomial degree by one, so we find that indeed $F^{(0)}_g = \textrm{Pol} \left( 3g-3 \right)$ which is what we wanted to prove.

In the analysis of the structure of the higher--instanton free energies we shall make use of this latter strategy. The reason for not using the initial degree approach is, as we shall see, the appearance of exponentials of the propagators in the solutions, which would make such approach impracticable (but in section \ref{sec:multiparameter} we will fully explain this point).

\subsection{Setting the Transseries Construction}

We shall now address the construction of transseries solutions, within the context of one--parameter transseries and for Calabi--Yau geometries with a complex moduli space of dimension one. Along this paper we shall lift these restrictions and construct completely general solutions but, to introduce the overall strategy, it proves useful to start with a simplified setting. Previously, we discussed how to appropriately write the holomorphic anomaly equations as a single equation for the full partition function, instead of an infinite tower of equations for the perturbative genus $g$ free energies. In this way, in the dimension one case, we obtained equation (\ref{eq:1parameterUVWeq}) with constraints (\ref{eq:Uconstraint}), (\ref{eq:Vconstraint}) and (\ref{eq:Wconstraint}) (with $i=j=z$). As already mentioned, while this equation is just a simple repackaging of the holomorphic anomaly equations, it also allows for a solution beyond the perturbative level: we may now try to solve it with a transseries \textit{ansatz} of the form
\begin{equation}\label{eq:1parameterF}
F(\sigma,g_s) = \sum_{n=0}^{+\infty} \sigma^n\, \rme^{- A^{(n)}/g_s}\, F^{(n)} (g_s),
\end{equation}
\noindent
where the transseries parameter $\sigma$ keeps track of the instanton number $n$. Further, $A^{(n)} := n A$, where $A\equiv A(z,S^{zz})$ is the instanton action, and the asymptotic perturbative expansions around each fixed multi--instanton sector ($n \neq 0$; the perturbative case is described separately in \eqref{perturbativeF}) take the form
\begin{equation}\label{eq:1parameteransatz}
F^{(n)}(g_s) \simeq \sum_{g=0}^{+\infty} g_s^{g+b^{(n)}} F^{(n)}_g (z,S^{zz}).
\end{equation}
\noindent
In this expression $b^{(n)}$, the starting order, is a characteristic exponent which depends upon the instanton number but which is not determined by the equations themselves. When we plug the above asymptotic series (\ref{eq:1parameteransatz}) into the holomorphic anomaly equations (\ref{eq:1parameterUVWeq}) we will find an expansion in both $\sigma$ and $g_s$ whose coefficients constitute the tower of equations we wish to solve\footnote{In particular we recover the familiar ``perturbative'' tower of holomorphic anomaly equations.}. However, we will further find that the particular equations we obtain do depend on the actual values of $b^{(n)}$, and as such we must find a procedure to fix them.

It is both pedagogical and convenient to perform the expansion of the holomorphic anomaly equation with the \textit{ansatz} (\ref{eq:1parameterF}) first in $\sigma$, dealing dealing with functions of $g_s$, $F^{(n)}(g_s)$, and after that to expand the result in $g_s$. For each $n\geq 1$ we find  
\bea
&&
\left( \partial_{S^{zz}} - \frac{1}{g_s} \partial_{S^{zz}} A^{(n)} \right) F^{(n)} - \frac{1}{2} g_s^2 \left( D_z - \frac{1}{g_s} \partial_z A^{(n)} + 2\, \partial_z \widetilde{F}^{(0)} \right) \left( \partial_z - \frac{1}{g_s} \partial_z A^{(n)} \right) F^{(n)} = \nonumber \\ 
&&
= \frac{1}{2}g_s^2\, \sum_{m=1}^{n-1} \left( \partial_z - \frac{1}{g_s} \partial_z A^{(m)} \right) F^{(m)} \left( \partial_z - \frac{1}{g_s} \partial_z A^{(n-m)} \right) F^{(n-m)}.
\label{eq:1parameterpreNPeq}
\eea
\noindent
Here we have defined $\widetilde{F}^{(0)}(g_s) := F^{(0)}(g_s) - \frac{1}{g_s^2}\, F^{(0)}_0$. The differential operator $\partial_z - \frac{1}{g_s}\, \partial_z A^{(n)}$ appears in (\ref{eq:1parameterpreNPeq}) because of the identity
\be
\partial_z^k \left( \rme^{-A^{(n)}/g_s} F^{(n)} \right) = \rme^{-A^{(n)}/g_s} \left(\partial_z -\frac{1}{g_s}\, \partial_z A^{(n)} \right)^k F^{(n)}, \qquad k=1,2,\ldots.
\ee
\noindent
The right--hand--side of (\ref{eq:1parameterpreNPeq}) comes from the quadratic term in (\ref{eq:1parameterUVWeq}). We have separated the terms corresponding to values of $m$ equal to $0$ and $n$, which involve $F^{(0)}$ and $F^{(n)}$, and moved them to the left;  these and the $U_i\, D_j F$ terms build up the $\widetilde{F}^{(0)}$ term in (\ref{eq:1parameterpreNPeq}). In particular, the left--hand--side of (\ref{eq:1parameterpreNPeq}) may be written as
\be
\mathcal{D}^{(n)}(g_s)\, F^{(n)}(g_s),
\ee
\noindent
where $\mathcal{D}^{(n)}(g_s)$ is a $g_s$--dependent differential operator
\be\label{eq:1parameterDoperator}
\mathcal{D}^{(n)}(g_s) \equiv \partial_{S^{zz}} - \frac{1}{g_s} \partial_{S^{zz}} A^{(n)} - \frac{1}{2} g_s^2 \left( D_z - \frac{1}{g_s} \partial_z A^{(n)} + 2\, \partial_z \widetilde{F}^{(0)} \right) \left( \partial_z - \frac{1}{g_s} \partial_z A^{(n)} \right).
\ee

Let us turn to the $g_s$--expansion, using (\ref{eq:1parameteransatz}) and $\mathcal{D}^{(n)}(g_s)=\sum_{g=-1}^{+\infty} g_s^g\, \mathcal{D}^{(n)}_g$, with $\mathcal{D}^{(n)}_{-1} := -\partial_{S^{zz}} A^{(n)}$. We find
\bea\label{eq:1paramgssumeq}
\sum_{g=-1}^{+\infty} g_s^g \left\{ -\partial_{S^{zz}} A^{(n)} F^{(n)}_{g+1} + \sum_{h=0}^g \mathcal{D}^{(n)}_h F^{(n)}_{g-h} \right\} &=& \\
&&
\hspace{-190pt} = \sum_{m=1}^{n-1} \sum_{g=0}^{+\infty} g_s^{g+B(n,m)}\, \frac{1}{2}\, \sum_{h=0}^g \left( \partial_z F^{(m)}_{h-1} - \partial_z A^{(m)} F^{(m)}_h \right) \left(\partial_z F^{(n-m)}_{g-1-h} - \partial_z A^{(n-m)} F^{(n-m)}_{g-h} \right), \nonumber
\eea
\noindent
where we have defined $B(n,m):=b^{(m)}+b^{(n-m)}-b^{(n)}$ for $n,m \ge 1$. As we have commented earlier, the holomorphic anomaly equations do not compute the values of the starting genus $b^{(n)}$. This is not a surprise since the values of $b^{(n)}$ are very much dependent upon the model and the geometry one is working with. Note, however, that the values of $B(n,m)$ determine the structure of the equations in the following way. Looking at \eqref{eq:1paramgssumeq} one sees that if $B(n,m)$ is sufficiently negative we will have terms on the right--hand--side of the equation but not on the left. This is in contrast to what happens when $B$ is sufficiently positive. In that case, after collecting terms with similar powers in $g_s$, we find an equation for $F^{(n)}_g$ which depends on previously computed data (see below). If $B(n,m)\geq 0$ for every $n,m \ge 1$, we obtain a system of equations which can be solved recursively in $n$ and $g$, and which has no extra constraint equations. If $B(n,m)<0$ for some $n$, $m$, some such constraints will appear and they will have to be treated separately. Note that while \eqref{eq:1paramgssumeq} imposes $B(n,m)$ to be an integer, all it imposes on the starting orders $b^{(n)}$ is that this constraint is satisfied; in particular the $b^{(n)}$ are not required to be integers.

Let us start with the simplest situation, $B(n,m)\geq 0$, and in particular work out the example $B(n,m)=+1$ for each $n$, $m$. After solving this case we shall see how it generalizes to any other $b^{(n)}$ such that $B(n,m)\geq 0$ with $n,m \ge 1$. Setting $B=+1$ and collecting terms with the same power in $g_s$, we obtain the tower of equations
\bea
\label{eq:1parameteralmostNPeq}
\partial_{S^{zz}} A^{(n)} F^{(n)}_{g+1} - \sum_{h=0}^g \mathcal{D}^{(n)}_h F^{(n)}_{g-h} + && \\
&&
\hspace{-70pt} + \frac{1}{2}\,\sum_{m=1}^{n-1} \sum_{h=0}^{g-1} \left( \partial_z F^{(m)}_{h-1} - \partial_z A^{(m)} F^{(m)}_h \right) \left( \partial_z F^{(n-m)}_{g-2-h} - \partial_z A^{(n-m)} F^{(n-m)}_{g-1-h} \right) = 0,
\nonumber
\eea
\noindent
for $g=-1,0,1,2,\ldots$. In the sums over $h$ we are using the convention that $F^{(m)}_h$ is zero whenever $h<0$. The explicit form the $\mathcal{D}^{(n)}_g$ operators, which we defined earlier as coefficients in the $g_s$--expansion of (\ref{eq:1parameterDoperator}), is
\bea
\label{eq:Dn0}
\mathcal{D}^{(n)}_0 &=& \partial_{S^{zz}} - \frac{1}{2} \left( \partial_z A^{(n)} \right)^2, \\
\label{eq:calD1}
\mathcal{D}^{(n)}_1 &=& \frac{1}{2}\, D_z^2 A^{(n)} + \partial_z A^{(n)} \left( D_z + \partial_z F^{(0)}_1 \right), \\
\label{eq:calD2}
\mathcal{D}^{(n)}_2 &=& - \frac{1}{2}\, D_z^2 - \partial_z F^{(0)}_1\, D_z, \\
\mathcal{D}^{(n)}_{2h-1} &=& \partial_z A^{(n)}\, \partial_z F^{(0)}_h, \qquad h = 2, 3, \ldots, \\
\label{eq:calD2h}
\mathcal{D}^{(n)}_{2h} &=& -\partial_z F^{(0)}_h\, D_z, \qquad h = 2, 3, \ldots.
\eea
\noindent
Let us have a closer look at the equations (\ref{eq:1parameteralmostNPeq}). For $g=-1$ there is only one term,
\be
\partial_{S^{zz}} A^{(n)}\, F^{(n)}_0 = 0.
\ee
\noindent
Since $F^{(n)}_0\neq 0$ by definition, and $A^{(n)}=n A$, we arrive at the important equation
\be
\partial_{S^{zz}} A =0.
\ee
\noindent
This means that the instanton action, $A$, that we presupposed dependent on both $z$ and $S^{zz}$ turns out \textit{not} to depend on the propagator $S^{zz}$. In fact, the instanton action is \textit{holomorphic}. Note, however, that our equations do not determine its $z$--functional dependence. This suggests that one may think of the instanton action, $A(z)$, as a holomorphic ambiguity, much in the same way as the familiar holomorphic ambiguities that appear in the integration of the free energies. In order to fix it, first recall the experience from matrix models \cite{msw07, ps09, dmp11} which tells us that the instanton action is generically given by a combination of periods of the spectral curve \cite{kk04}. One may then forget about matrix models \cite{bkmp07, dmp11}---albeit keeping its spectral geometry structure---by considering toric Calabi--Yau threefolds, whose mirror is essentially described by a Riemann surface, the mirror curve \cite{hv00}, and take the instanton action as an appropriate combination of periods of this curve, \textit{i.e.}, as an appropriate combination of Calabi--Yau periods.

Using this result back in (\ref{eq:1parameteralmostNPeq}) simplifies the equations to
\bea
\label{eq:1parameterNPeq}
\left( \partial_{S^{zz}} - \frac{1}{2} \left( \partial_z A^{(n)} \right)^2 \right) F^{(n)}_g &=& \\
&&
\hspace{-130pt} = - \sum_{h=1}^g \mathcal{D}^{(n)}_h F^{(n)}_{g-h} + \frac{1}{2}\, \sum_{m=1}^{n-1} \sum_{h=0}^{g-1} \left( \partial_z F^{(m)}_{h-1} - \partial_z A^{(m)} F^{(m)}_h \right) \left( \partial_z F^{(n-m)}_{g-2-h} - \partial_z A^{(n-m)} F^{(n-m)}_{g-1-h} \right). \nonumber
\eea
\noindent
The operator that appears on the left--hand--side is $\mathcal{D}^{(n)}_0$, \eqref{eq:Dn0}. Note that the right--hand--side of (\ref{eq:1parameterNPeq}) depends only on \textit{previous} instanton sectors and genera, and it is, therefore, a completely known object in the recursive integration of the equations. Further, the set of operators $\mathcal{D}^{(n)}_h$ appearing on the right--hand--side of (\ref{eq:1parameterNPeq}) will always include a second--order covariant derivative, as long as $g \ge 2$ (just look at \eqref{eq:calD2} above). In other words, \eqref{eq:1parameterNPeq} is, effectively, the nonperturbative analogue of the holomorphic anomaly equations of \cite{bcov93}.

As we explained before, \eqref{eq:1parameterNPeq} are the resulting ``nonperturbative'' holomorphic anomaly equations when $b^{(n)}=+1$, $\forall n$ (or, more generally, when $B(n,m)=+1$, $\forall n,m$). Otherwise, the calculation simply generalizes to
\bea
\label{eq:B1parameteralmostNPeq}
\partial_{S^{zz}} A^{(n)} F^{(n)}_{g+1} - \sum_{h=0}^g \mathcal{D}^{(n)}_h F^{(n)}_{g-h} + && \\
&&
\hspace{-140pt} + \frac{1}{2}\,\sum_{m=1}^{n-1} \sum_{h=0}^{g-B(n,m)} \left( \partial_z F^{(m)}_{h-1} - \partial_z A^{(m)} F^{(m)}_h \right) \left( \partial_z F^{(n-m)}_{g-1-B(n,m)-h} - \partial_z A^{(n-m)} F^{(n-m)}_{g-B(n,m)-h} \right) = 0.
\nonumber
\eea
\noindent
Let us stress that, independently of the value of $B(n,m)$, the equation for $g=-1$ and $n=1$ is
\be
\partial_{S^{zz}} A^{(1)} F^{(1)}_0 =0,
\ee
\noindent
which implies that $A^{(1)} \equiv A$ is holomorphic also in this general case.

As a final remark, let us note that the one--parameter transseries we addressed may be thought of as a multi--parameters transseries when all sectors but one are switched off. That is, in a general setting, where (generalized) instanton sectors are labelled by integers $(n_1 |\cdots | n_\alpha |\cdots | n_k)$, the $n$--instanton sector of an one--parameter transseries is embedded as $(0 |\cdots | n | \cdots | 0)$. In this way, it is natural to expect that much of the structure we have uncovered so far will have a natural generalization within multi--parameter transseries.

\subsection{Structure of the Transseries Solution}

So far we have obtained an infinite tower of equations; a nonperturbative extension of the holomorphic anomaly equations. The key question is, of course, whether we can solve/integrate them explicitly, or how much can be said about their (generic) solutions. Similar to what was done in the perturbative case \cite{yy04, al07, alm08}, it turns out that also in the present nonperturbative case we can be very specific about the general structure of the solutions.

Let us for the moment stick to the case $b^{(n)}=1$, $\forall n$, and let us denote by $G^{(n)}_g$ the full right--hand--side of \eqref{eq:1parameterNPeq}. Because of the simple structure of this equation, it is straightforward to integrate in the propagator, $S^{zz}$, and obtain
\begin{equation}\label{eq:1parameterintegratingG}
F^{(n)}_g (z,S^{zz}) = \rme^{\frac{1}{2} \left( \partial_z A^{(n)} \right)^2 S^{zz}} \left( f^{(n)}_g (z) + \int^{S^{zz}} \rmd \widetilde{S}^{zz}\, \rme^{- \frac{1}{2} \left( \partial_z A^{(n)} \right)^2 \widetilde{S}^{zz}}\, G^{(n)}_g (z,\widetilde{S}^{zz}) \right),
\end{equation}
\noindent
where $f^{(n)}_g (z)$ is the holomorphic ambiguity of the $n$th multi--instanton sector at order $g$. Notice the nonperturbative novelty: the appearance of exponentials in $S^{zz}$ in the solution \eqref{eq:1parameterintegratingG}. One can further anticipate the presence of different exponentials arising from previous instanton sectors building up the actual expression for $F^{(n)}_g$. The reason for these exponentials factors is, of course, the fact that in the nonperturbative setting the differential operator we are integrating is no longer just a derivative, $\partial_{S^{zz}}$, but a derivative with an extra linear term in the instanton action.

Let us first have a look at the actual form of the equations for the one--instanton sector, $n=1$. In this sector, the second term of the right--hand--side of (\ref{eq:1parameterNPeq}) vanishes, so we have
\be
G^{(1)}_g = - \sum_{h=1}^g \mathcal{D}^{(1)}_h F^{(1)}_{g-h}.
\ee
\noindent
Starting at $g=0$, as $G^{(1)}_0 = 0$ the equation for $F^{(1)}_0$ is simply
\be
\left( \partial_{S^{zz}} - \frac{1}{2} (\partial_z A)^2 \right) F^{(1)}_0 = 0,
\ee
\noindent
which can be immediately integrated with respect to $S^{zz}$ to give
\be
F^{(1)}_0 = \rme^{\frac{1}{2} (\partial_z A)^2 S^{zz}}\, f^{(1)}_0(z).
\ee
\noindent
Here, $f^{(1)}_0(z)$ is the holomorphic ambiguity. At next order, for $g=1$, we now have
\be
G^{(1)}_1 = - \mathcal{D}^{(1)}_1 F^{(1)}_0,
\ee
\noindent
where $\mathcal{D}^{(1)}_1$ is given by (\ref{eq:calD1}) with $n=1$. The exponential which is part of $F^{(1)}_0$ naturally survives the action of $\mathcal{D}^{(1)}_1$, while the $z$--derivative acting on the propagator $S^{zz}$, that sits in the exponential, gives back an $(S^{zz})^2$ term (recall \eqref{eq:covariantderivativeonpropagator}). In this way, $G^{(1)}_1$ is the product of the exponential $\exp \frac{1}{2} \left( \partial_z A \right)^2 S^{zz}$ by a polynomial of degree $2$ in $S^{zz}$. When we plug this expression into the general form of the solution (\ref{eq:1parameterintegratingG}), the exponential inside the integral cancels that of $G^{(1)}_1$ and we are left with the integral of a polynomial (of known degree). Thus, from (\ref{eq:1parameterintegratingG}) we obtain that $F^{(1)}_1$ is the product of $\exp \frac{1}{2} \left( \partial_z A \right)^2 S^{zz}$ by a polynomial of degree $3$.

This suggests that, in general, $F^{(1)}_g$ is of the form
\be\label{eq:1parameterhigherinstantonsolution}
F^{(1)}_g = \rme^{\frac{1}{2}(\partial_z A)^2 S^{zz}}\, \textrm{Pol} \left( S^{zz};3g \right),
\ee
\noindent
where $\textrm{Pol} \left(S^{zz};d\right)$ stands for a polynomial of degree $d$ in $S^{zz}$, with coefficients which are holomorphic functions of $z$. It can be checked by explicit computation that this is indeed the case for low $g$. To prove the general case one proceeds by induction, studying the structure of $G^{(1)}_g$ and finding that it is the product of a polynomial of degree $3g-1$ by the usual exponential. Integration then yields the result. We shall be more thorough when looking at the structure of the free energies for the multi--parameter transseries, as the strategy for that case will be the same as in here.

Having understood the $n=1$ sector we may move on to the $n=2$ equations. Now the quadratic terms in (\ref{eq:1parameterNPeq}) are present and the structure of the solution is more complicated. This contribution,
\be
\sum_{h=0}^{g-1} \left( \partial_z F^{(1)}_{h-1} - \partial_z A^{(1)} F^{(1)}_h \right) \left( \partial_z F^{(1)}_{g-2-h} - \partial_z A^{(1)} F^{(n-m)}_{g-1-h} \right),
\ee
\noindent
will produce terms with $\exp \frac{1}{2} \left(1+1\right) \left( \partial_z A \right)^2 S^{zz} = \exp \frac{1}{2}\, 2 \left( \partial_z A \right)^2 S^{zz}$. There will also be terms with $\exp \frac{1}{2} \left( \partial_z A^{(2)} \right)^2 S^{zz} = \exp \frac{1}{2}\, 4 \left( \partial_z A \right)^2 S^{zz}$. In this way, in order to specify the structure of the $F^{(2)}_g$ free energies, one needs to know the degrees, $d_1$ and $d_2$, of the polynomials in $S^{zz}$ that multiply these exponentials:
\be
F^{(2)}_g = \rme^{\frac{1}{2}\, 2 \left( \partial_z A \right)^2 S^{zz}}\, \textrm{Pol}\left(S^{zz};d_1\right) + \rme^{\frac{1}{2}\, 4 \left( \partial_z A \right)^2 S^{zz}}\, \textrm{Pol} \left(S^{zz};d_2 \right).
\ee
\noindent
Computing the first free energies for $g=0,1,2,\ldots,$ suggests $d_1=3(g+1-2)$ and $d_2 =3(g+1-1)$.

Mechanically computing case by case, we find that the next instanton sector, $n=3$, will have exponentials of the type $\exp \frac{1}{2}\, a \left( \partial_z A \right)^2 S^{zz}$ with $a \in \{ 3,5,9 \}$ and the associated polynomials will have degrees $d = 3 \left( g+1-\lambda \right)$, with $\lambda \in \{ 3,2,1 \}$, respectively. These numbers get more complicated as the instanton number increases. For example, for $n=7$, $a \in \{ 7, 9, 11, 13, 15, 17, 19, 21, 25, 27, 29, 37, 49 \}$ and $\lambda \in \{Ê7, 6, 5, 4, 4, 3, 3, 3, 2, 3, 2, 2, 1 \}$, respectively.

The nature of the numbers $a$ and $\lambda$ is purely combinatoric, and they actually have an interpretation\footnote{This interpretation becomes too cumbersome when we have a multi--parameters transseries.} in terms of integer partitions of the instanton number $n$. In the next section we shall prove that these numbers are encoded in the generating function
\be\label{eq:1parambeta1genfunction}
\Phi = \prod_{m=1}^{+\infty} \frac{1}{1 - \varphi\, E^{m^2}\, \rho^m} = \sum_{n=0}^{+\infty} \sum_{\{\gamma_n\}} E^{a \left(n;\gamma_n\right)}\, \varphi^{\lambda \left(n;\gamma_n\right)} \left( 1 + \mathcal{O}(\varphi) \right),
\ee
\noindent
where $\varphi$, $E$, and $\rho$ are formal variables, and $\{\gamma_n\}$ is a set of indices; but see section \ref{sec:multiparameter} for a detailed explanation and origin of this generating function.

It is also important to note from the general form of the solution (\ref{eq:1parameterintegratingG}) that the holomorphic ambiguity, $f^{(n)}_g$, always sits in the polynomial that accompanies $\exp \frac{1}{2}\, n^2 \left( \partial_z A \right)^2 S^{zz}$. Interestingly, this has a theta--function flavor as in \cite{em08}. The corresponding polynomial is of degree $3(g+1-1)=3g$, which is the highest since $\lambda=1$ only for $a = n^2$ (see section \ref{sec:multiparameter} for further details).

In summary, we found that the higher--instanton free energies have the form
\be
F^{(n)}_g = \sum_{\{\gamma_n\}} \rme^{\frac{1}{2}\, a \left( n;\gamma_n \right) \left( \partial_z A \right)^2 S^{zz}}\, \textrm{Pol} \left(S^{zz}; 3 \left( g+1-\lambda \left( n;\gamma_n \right) \right)\right),
\ee
\noindent
where the combinatorial data $\{ a, \lambda, \gamma \}$ is encoded in the generating function \eqref{eq:1parambeta1genfunction}. Similarly, one can derive an analogous structure when $b^{(n)}$ is more general than just $+1$. Assuming the restriction that $B(n,m) \geq 0$, $\forall n,m$, one would find (and again we refer the reader to  section \ref{sec:multiparameter} for a detailed proof)
\be
F^{(n)}_g = \sum_{\{\gamma_n\}} \rme^{\frac{1}{2}\, a \left( n;\gamma_n \right) \left( \partial_z A \right)^2 S^{zz}}\, \textrm{Pol} \left( S^{zz}; 3\left( g + b^{(n)} - \lambda_b \left( n;\gamma_n \right) \right) \right),
\ee
\noindent
where this time around the combinatorial data is stored in
\be\label{eq:1parambetagenfunction}
\Phi_b = \prod_{m=1}^{+\infty} \frac{1}{1 - \varphi^{b^{(m)}}\, E^{m^2}\, \rho^m} = \sum_{n=0}^{+\infty} \sum_{\{\gamma_n\}} E^{a \left(n;\gamma_n\right)}\, \varphi^{\lambda_b \left(n;\gamma_n\right)} \left( 1 + \mathcal{O}(\varphi) \right).
\ee
\noindent
Do note that $a$ and $\gamma$ are independent of the value of the starting genus $b^{(n)}$.

A few comments are now in order. The fact that we have not approached the determination of the structure of the $F^{(n)}_g$ by assigning degrees to the various objects of the corresponding equations---as was earlier done in the perturbative sector---is because now we do not only have polynomials. What we actually have is a sum of different exponentials times polynomials of various degrees, and tracking all those degrees at the same time while capturing all the combinatorics of the problem makes the approach much less clear than using a more hands--on calculation as above. Also, as we mentioned before, a detailed analysis of the structure of the solution and a recipe to calculate the $a$ and $\lambda$ numbers will be studied in the multi--parameters transseries case. The proofs will rely on induction, carefully studying the exponential and polynomial structure of the $G^{(n)}_g$ in (\ref{eq:1parameterintegratingG}) and tracking what happens when integrating. Further, it is important to stress that the holomorphic anomaly equations in (\ref{eq:1parameterNPeq}) determine the functional dependence of the free energies $F^{(n)}_g(z,S^{zz})$, up to a holomorphic ambiguity, $f^{(n)}_g(z)$, just as in the perturbative case (see section \ref{sec:fixingambiguities}). The anti--holomorphic dependence, encoded in the propagator $S^{zz}$, is in the form of polynomials and exponentials. The numbers $a$ and $\lambda$, depending on the instanton number $n$, determine the coefficients in the exponentials and the degree of the polynomials, respectively.

Finally, let us point out that in the above transseries computations we have focused only on the propagators $S^{ij}$, having switched off the rest, $S^i$, $S$ and $\partial_i K$. For the moment this makes the calculations easier, while still preserving non--holomorphic information in the problem and, as we shall see later on, still being able to fix the holomorphic ambiguities (since the other propagators vanish in the holomorphic limit). One can easily check that our computations can be extended without major difficulties in order to include every other generator, and it turns out that the corresponding results do not differ much from the ones we have found. In particular, one can easily check that the instanton action is still independent of \textit{all} propagators, and hence holomorphic. The resulting transseries structure is very similar, with the same exponential terms in $S^{ij}$, but where the associated polynomials will now comprise a few more terms while maintaining the very \textit{same} total degree. In particular, $S^i$ would contribute with degree $+2$, $S$ with degree $+3$, and $\partial_i K$ with degree $+1$ (also see \cite{al07} for the discussion of how to compute the degree when including all generators). We believe this structure is generic, and will assume this in the following---thus restricting our analysis to the sole inclusion of $S^{ij}$ throughout our work.

\section{Multi--Parameter Resurgent Transseries}\label{sec:multiparameter}

In the previous section we have shown how to construct one--parameter transseries solutions to the holomorphic anomaly equations. These solutions will be adequate nonperturbative solutions whenever there is a single instanton action in the game. However, most of the time this is not the case: either there are different instanton actions associated to different (physical) instanton effects, or resurgence demands that further ``generalized'' instanton actions should come into play when considering the full nonperturbative grand--canonical partition function (see, \textit{e.g.}, \cite{gikm10, asv11, sv13}). In this case, multi--parameter transseries solutions must be considered. We shall now see how our previous discussion generalizes to the multi--parameters case; in fact this section will follow closely its analog for the one--parameter transseries.

\subsection{Setting the Generalized Transseries Construction}

Let us consider a general multi--parameters transseries \textit{ansatz}, with the form
\begin{equation}\label{eq:preNPansatz}
F ( \boldsymbol{\sigma}, g_s) = \sum_{\boldsymbol{n} \in \BN^\kappa_0} \boldsymbol{\sigma}^{\boldsymbol{n}}\, \rme^{-A^{(\boldsymbol{n})}/g_s}\, F^{(\boldsymbol{n})} (g_s).
\end{equation}
\noindent
We are assuming one has $\kappa$ distinct instanton actions, $A_\alpha$, and we are using the shorthand notation $A^{(\boldsymbol{n})} = \sum_{\alpha=1}^\kappa n_\alpha A_\alpha$ to identify all (generalized) multi--instanton sectors. At this moment we let $A_\alpha = A_\alpha (z,S^{zz})$ for absolute generality, but we shall see in the following that all these instanton actions are independent of $S^{zz}$, \textit{i.e.}, they are all \textit{holomorphic} functions. At each of the multi--instanton sectors, one further finds asymptotic perturbative expansions which for the moment we just denote by $F^{(\boldsymbol{n})} (g_s)$, \textit{i.e.}, in the first few steps it proves useful not to look at the particular asymptotic expansions in $g_s$. Finally, the transseries parameters encoding the nonperturbative ambiguities are assembled in the expression above as $\boldsymbol{\sigma}^{\boldsymbol{n}} := \prod_{\alpha=1}^\kappa \sigma_\alpha^{n_\alpha}$.

If we plug our transseries \textit{ansatz} (\ref{eq:preNPansatz}) into the holomorphic anomaly equation\footnote{We are still working with complex moduli space of dimension one. This requirement is lifted in appendix \ref{app:multidimcomplmodspace}.} \eqref{eq:1parameterUVWeq}, and we collect terms with the same powers in $\boldsymbol{\sigma}$, we find
\bea
&&
\left( \partial_{S^{zz}} - \frac{1}{g_s} \partial_{S^{zz}} A^{(\boldsymbol{n})} \right) F^{(\boldsymbol{n})} - \frac{1}{2} g_s^2 \left( D_z - \frac{1}{g_s} \partial_z A^{(\boldsymbol{n})} + 2\, \partial_z \widetilde{F}^{(\boldsymbol{0})} \right) \left( \partial_z - \frac{1}{g_s} \partial_z A^{(\boldsymbol{n})} \right) F^{(\boldsymbol{n})} = \nonumber \\ 
&&
= \frac{1}{2}g_s^2\, \sideset{}{'}\sum_{\boldsymbol{m}=\boldsymbol{0}}^{\boldsymbol{n}} \left( \partial_z - \frac{1}{g_s} \partial_z A^{(\boldsymbol{m})} \right) F^{(\boldsymbol{m})} \left( \partial_z - \frac{1}{g_s} \partial_z A^{(\boldsymbol{n}-\boldsymbol{m})} \right) F^{(\boldsymbol{n}-\boldsymbol{m})},
\label{eq:preNPeq}
\eea
\noindent
where again $\widetilde{F}^{(\boldsymbol{0})}(g_s) := F^{(\boldsymbol{0})}(g_s)-\frac{1}{g_s^2}\, F^{(\boldsymbol{0})}_0$. The same comments we made to (\ref{eq:1parameterpreNPeq}) apply now in here. The only novelty is the use of a primed sum over the previous instanton numbers, meaning
\be
\sideset{}{'}\sum_{\boldsymbol{m}=\boldsymbol{0}}^{\boldsymbol{n}} \equiv\,  \sum_{\substack{m_1,\ldots,m_\kappa=0 \\ \boldsymbol{m} \neq \boldsymbol{0},\, \boldsymbol{m}\neq \boldsymbol{n}}}^{n_1,\ldots,n_\kappa}.
\ee
\noindent
That is, the sum does not involve either $F^{(\boldsymbol{0})}$ or $F^{(\boldsymbol{n})}$. Instead, those terms have been moved to the left--hand--side of (\ref{eq:preNPeq}). Again, as in the previous section, we can define a differential operator $\mathcal{D}^{(\boldsymbol{n})}$ and its $g_s$--expansion; see (\ref{eq:1parameterDoperator}).

Let us now consider the perturbative $g_s$ asymptotic expansions of each (generalized) multi--instanton sector,
\be\label{eq:multiparameteransatz}
F^{(\boldsymbol{n})} \simeq \sum_{g=0}^{+\infty} g_s^{g+b^{(\boldsymbol{n})}} F^{(\boldsymbol{n})}_g,
\ee
\noindent
where $b^{(\boldsymbol{n})}$ is a characteristic exponent that generalizes the one in the last section. Plugging this expansion (\ref{eq:multiparameteransatz}) back into (\ref{eq:preNPeq}), we arrive at
\bea
\sum_{g=-1}^{+\infty} g_s^g \left\{ -\partial_{S^{zz}} A^{(\boldsymbol{n})} F^{(\boldsymbol{n})}_{g+1} + \sum_{h=0}^g \mathcal{D}^{(\boldsymbol{n})}_h F^{(\boldsymbol{n})}_{g-h} \right\} &=& \\
&&
\hspace{-190pt} = \sideset{}{'}\sum_{\boldsymbol{m}=\boldsymbol{0}}^{\boldsymbol{n}}\, \sum_{g=0}^{+\infty} g_s^{g+B(\boldsymbol{n},\boldsymbol{m})}\, \frac{1}{2}\, \sum_{h=0}^{g} \left( \partial_z F^{(\boldsymbol{m})}_{h-1} - \partial_z A^{(\boldsymbol{m})} F^{(\boldsymbol{m})}_h \right) \left(\partial_z F^{(\boldsymbol{n}-\boldsymbol{m})}_{g-1-h} - \partial_z A^{(\boldsymbol{n}-\boldsymbol{m})} F^{(\boldsymbol{n}-\boldsymbol{m})}_{g-h} \right), \nonumber
\eea
\noindent
where $B(\boldsymbol{n},\boldsymbol{m}) := b^{(\boldsymbol{m})}+b^{(\boldsymbol{n}-\boldsymbol{m})}-b^{(\boldsymbol{n})}$. In the following we shall carry out the calculation for a general $b^{(\boldsymbol{n})}$, with the condition $B(\boldsymbol{n},\boldsymbol{m}) \geq 0$ to avoid the case of constraints, but it might be useful to have the example $B(\boldsymbol{n},\boldsymbol{m})=+1$, $\forall\, \boldsymbol{n},\boldsymbol{m}$, in mind. Collecting terms with the same power in $g_s$ we find the equations
\begin{eqnarray}
\label{eq:almostNPeq}
\partial_{S^{zz}}A^{(\boldsymbol{n})} F^{(\boldsymbol{n})}_{g+1} - \sum_{h=0}^g \mathcal{D}^{(\boldsymbol{n})}_h F^{(\boldsymbol{n})}_{g-h} &+& \\
&&
\hspace{-150pt} + \frac{1}{2}\, \sideset{}{'}\sum_{\boldsymbol{m}=\boldsymbol{0}}^{\boldsymbol{n}}\, \sum_{h=0}^{g-B(\boldsymbol{n},\boldsymbol{m})} \left( \partial_z F^{(\boldsymbol{m})}_{h-1} - \partial_z A^{(\boldsymbol{m})} F^{(\boldsymbol{m})}_h \right) \left( \partial_z F^{(\boldsymbol{n}-\boldsymbol{m})}_{g-1-B(\boldsymbol{n},\boldsymbol{m})-h} - \partial_z A^{(\boldsymbol{n}-\boldsymbol{m})} F^{(\boldsymbol{n}-\boldsymbol{m})}_{g-B(\boldsymbol{n},\boldsymbol{m})-h} \right) = 0,
\nonumber
\end{eqnarray}
\noindent
for $g=-1,0,1,2,\ldots$. For the one--instanton sectors, that is, those of the form $\boldsymbol{n} = (0|\cdots|1|\cdots|0)$, the quadratic term in \eqref{eq:almostNPeq} drops (for any value of $B$). If we further focus on $g=-1$, also the second term drops,  and we are left with
\be
\partial_{S^{zz}} A^{(0|\cdots|1|\cdots|0)} = 0
\ee
\noindent
for every one--instanton sector, since $F^{(0|\cdots|1|\cdots|0)}_0 \neq 0$. We therefore conclude
\be
\partial_{S^{zz}} A_\alpha = 0, \qquad \alpha =1,\ldots, \kappa,
\ee
\noindent
\textit{i.e.}, \textit{every} instanton action is holomorphic. This is a fundamental point within the nonperturbative structure of topological string theory, and we shall return to it later in the paper.

Finally, if one rewrites \eqref{eq:almostNPeq} above as an equation for the unknowns $F^{(\boldsymbol{n})}_g$---or, more precisely, for $\mathcal{D}^{(\boldsymbol{n})}_0 F^{(\boldsymbol{n})}_g$ as we are dealing with a differential equation---it follows that
\begin{eqnarray}
\label{eq:NPeq}
\left( \partial_{S^{zz}} - \frac{1}{2} \left( \partial_z A^{(\boldsymbol{n})} \right)^2 \right) F^{(\boldsymbol{n})}_g &=& - \sum_{h=1}^g \mathcal{D}^{(\boldsymbol{n})}_h F^{(\boldsymbol{n})}_{g-h} + \\
&&
\hspace{-130pt} + \frac{1}{2}\, \sideset{}{'}\sum_{\boldsymbol{m}=\boldsymbol{0}}^{\boldsymbol{n}}\, \sum_{h=0}^{g-B(\boldsymbol{n},\boldsymbol{m})} \left( \partial_z F^{(\boldsymbol{m})}_{h-1} - \partial_z A^{(\boldsymbol{m})} F^{(\boldsymbol{m})}_h \right) \left( \partial_z F^{(\boldsymbol{n}-\boldsymbol{m})}_{g-1-B(\boldsymbol{n},\boldsymbol{m})-h} - \partial_z A^{(\boldsymbol{n}-\boldsymbol{m})} F^{(\boldsymbol{n}-\boldsymbol{m})}_{g-B(\boldsymbol{n},\boldsymbol{m})-h} \right), \nonumber
\end{eqnarray}
\noindent
where it should be clear that the right--hand--side depends only on \textit{previous} instanton sectors and genera, and it is, therefore, a known object in the recursive integration of the equations. The $\mathcal{D}^{(\boldsymbol{n})}_h$ operators can be read from (\ref{eq:calD1}--\ref{eq:calD2h}), trivially changing $n$ by $\boldsymbol{n}$ and $0$ by $\boldsymbol{0}$.

\subsection{Structure of the Generalized Transseries Solution}

If we denote by $G^{(\boldsymbol{n})}_g$ the right--hand--side of (\ref{eq:NPeq}), it is straightforward to find the solution to the differential equation:
\begin{equation}\label{eq:integratingG}
F^{(\boldsymbol{n})}_g (z,S^{zz}) = \rme^{\frac{1}{2} \left( \partial_z A^{(\boldsymbol{n})} \right)^2 S^{zz}} \left( f^{(\boldsymbol{n})}_g (z) + \int^{S^{zz}} \rmd \widetilde{S}^{zz}\, \rme^{- \frac{1}{2} \left( \partial_z A^{(\boldsymbol{n})} \right)^2 \widetilde{S}^{zz}}\, G^{(\boldsymbol{n})}_g (z,\widetilde{S}^{zz}) \right),
\end{equation}
\noindent
where $f^{(\boldsymbol{n})}_g (z)$ is the holomorphic ambiguity of the multi--instanton sector $\boldsymbol{n}$, at order $g$. This expression, \eqref{eq:integratingG}, is very reminiscent of similar formulae we have obtained in section \ref{sec:1parameter} for one--parameter transseries solutions to the holomorphic anomaly equations. As such, it is without surprise that we will find also in here, in the multi--parameter transseries context, that it is possible to fix the polynomial structure of the many nonperturbative sectors---precisely as finite linear combinations of exponential terms multiplied by suitable polynomials. 

Let us make the aforementioned expectation precise. If correct, this will imply that the structure of the solution will be specified by the coefficients in the exponentials and the degree of the polynomials. In this case, one finds expressions of the form
\begin{equation}\label{eq:Fstructure}
F^{(\boldsymbol{n})}_g = \sum_{\left\{ \gamma_{\boldsymbol{n}} \right\}} \exp \left( \frac{1}{2} \sum_{\alpha,\beta=1}^\kappa a_{\alpha\beta} \left( \boldsymbol{n}; \gamma_{\boldsymbol{n}} \right) \partial_z A_\alpha\, \partial_z A_\beta\, S^{zz} \right) \textrm{Pol} \left( S^{zz}; d_b \left( \boldsymbol{n}; g; \gamma_{\boldsymbol{n}} \right) \right).
\end{equation}
\noindent
Here $\gamma_{\boldsymbol{n}}$ is an index running in some finite set (it may be regarded as a label); $a_{\alpha\beta} \left( \boldsymbol{n}; \gamma_{\boldsymbol{n}} \right)$ are non--negative integers (we are already anticipating here that these numbers will not depend on $g$) and they are the generalization of the $a$--numbers we found at the end of section \ref{sec:1parameter}; and $d_b(\boldsymbol{n};g;\gamma_{\boldsymbol{n}})$ is the degree of the corresponding polynomial in $S^{zz}$ (whose coefficients will have a holomorphic dependence on $z$). In the formula for the degree $d_b$ we shall encounter the generalization of the $\lambda_b$--numbers of section \ref{sec:1parameter}; see \eqref{eq:degreeformula} in the following and note the $b$--dependence.

It turns out that, as we shall see in detail below, all these integers may be codified by a single generating function
\begin{equation}\label{eq:generatingfunction}
\Phi_b \equiv \Phi_b \left( \varphi, E, \rho \right) := \sideset{}{'}\prod_{\boldsymbol{m}=\boldsymbol{0}}^{+\infty} \frac{1}{1 - \varphi^{b^{(\boldsymbol{m})}}\, \prod_{\alpha,\beta=1}^\kappa E_{\alpha\beta}^{m_\alpha m_\beta} \prod_{\alpha=1}^\kappa \rho_\alpha^{m_\alpha}},
\end{equation}
\noindent
where $E_{\alpha\beta}$ is taken to be symmetric in $(\alpha,\beta)$. The prime, $(\:')$, in the infinite product means that $\boldsymbol{m}=\boldsymbol{0}$ is not considered (if it were we would get an extra term $(1-\varphi)^{-1}$, so we omit it). Here $\varphi$, $E_{\alpha\beta}$, and $\rho_\alpha$ are formal variables, so that $\Phi_b$ may be expanded in power series using the formal identity $(1-x)^{-1} = 1 + x + x^2 + x^3 + \cdots$ as
\begin{equation}
\Phi_b = \sideset{}{'}\prod_{\boldsymbol{m}=\boldsymbol{0}}^{+\infty}\, \sum_{r=0}^{+\infty} \varphi^{r\, b^{(\boldsymbol{m})}} \prod_{\alpha,\beta=1}^\kappa E_{\alpha\beta}^{r\, m_\alpha m_\beta}\, \prod_{\alpha=1}^\kappa \rho_\alpha^{r\, m_\alpha}.
\end{equation}
\noindent
Next, we explicitly evaluate the $\boldsymbol{m}$--products (relabeling $r\rightarrow r_{\boldsymbol{m}}$),
\begin{equation}
\Phi_b = \sum_{\{r_{\boldsymbol{m}}\}} \varphi^{\sum_{\boldsymbol{m}=\boldsymbol{0}}^{+\infty} r_{\boldsymbol{m}}\, b^{(\boldsymbol{m})}}\, \prod_{\alpha,\beta=1}^\kappa E_{\alpha\beta}^{\sum_{\boldsymbol{m}=\boldsymbol{0}}^{+\infty} r_{\boldsymbol{m}}\, m_\alpha m_\beta}\, \prod_{\alpha=1}^\kappa \rho_\alpha^{\sum_{\boldsymbol{m}=\boldsymbol{0}}^{+\infty} r_{\boldsymbol{m}}\, m_\alpha},
\end{equation}
\noindent
and reorganize the several terms according to their powers in $\rho_\alpha$ for each $\alpha$. Thus, for any fixed $\boldsymbol{n}$, we select those values of $\{r_{\boldsymbol{m}}\}$ which satisfy $\sum_{\boldsymbol{m}=\boldsymbol{0}}^{+\infty} r_{\boldsymbol{m}}\, \boldsymbol{m} = \boldsymbol{n}$. For example, in the case of one--parameter transseries we addressed earlier, one would simply be doing the choice of $\sum_{m=0}^{+\infty} r_m\, m = n$, or $1\cdot r_1 + 2\cdot r_2 + 3\cdot r_3 + \cdots = n$, that is, one would be selecting those $\{r_m\}$ that codify an integer partition of $n$. It then follows that
\begin{equation}
\Phi_b = \sum_{\boldsymbol{n}=\boldsymbol{0}}^{+\infty} \boldsymbol{\rho}^{\boldsymbol{n}} \sum_{\{r_{\boldsymbol{m}}\}\: :\: \sum_{\boldsymbol{m}=\boldsymbol{0}}^{+\infty} r_{\boldsymbol{m}}\, \boldsymbol{m} = \boldsymbol{n}} \varphi^{\sum_{\boldsymbol{m}=\boldsymbol{0}}^{+\infty} r_{\boldsymbol{m}}\, b^{(\boldsymbol{m})}} \prod_{\alpha,\beta=1}^\kappa E_{\alpha\beta}^{\sum_{\boldsymbol{m}=\boldsymbol{0}}^{+\infty} r_{\boldsymbol{m}}\, m_\alpha m_\beta}.
\end{equation}
\noindent
Next, one may separate, for each $\boldsymbol{n}$, the $\{r_{\boldsymbol{m}}\}$ (with $\sum_{\boldsymbol{m}=\boldsymbol{0}}^{+\infty} r_{\boldsymbol{m}}\, \boldsymbol{m} = \boldsymbol{n}$) according to their values of $\left\{\sum_{\boldsymbol{m}=\boldsymbol{0}}^{+\infty} r_{\boldsymbol{m}}\, m_\alpha m_\beta \right\}_{\alpha,\beta}$, and introduce a label $\gamma_{\boldsymbol{n}}$ that runs over these classes. Let us define
\begin{equation}
a_{\alpha\beta} \left( \boldsymbol{n}; \gamma_{\boldsymbol{n}} \right) := \sum_{\boldsymbol{m}=\boldsymbol{0}}^{+\infty} r_{\boldsymbol{m}}\, m_\alpha m_\beta
\end{equation}
\noindent
for $\{r_{\boldsymbol{m}}\}$ such that $\sum_{\boldsymbol{m}} r_{\boldsymbol{m}}\, \boldsymbol{m} = \boldsymbol{n}$. Then, $\gamma_{\boldsymbol{n}}$ is the label that identifies the classes which have different values of $\{a_{\alpha\beta}\}$. In the one--parameter example, one had $a(n;\gamma_n) = \sum_{m=0} r_m\, m^2$ with $\sum_{m=0} r_m\, m=n$. This means that $a(n;\gamma_n)$ is the sum of squares of the parts of the integer partition of $n$ specified by $\{r_m\}$. The value of $a$ breaks the set of partitions of $n$ into equivalence classes, labeled by $\gamma_n$. Now, for each $\gamma_{\boldsymbol{n}}$ we have a set of numbers $a_{\alpha\beta} \left( \boldsymbol{n}; \gamma_{\boldsymbol{n}} \right)$, and we may write our equation above as
\begin{equation}\label{eq:explicitFgenfunction}
\Phi_b = \sum_{\boldsymbol{n}=\boldsymbol{0}}^{+\infty} \boldsymbol{\rho}^{\boldsymbol{n}}\, \sum_{\left\{ \gamma_{\boldsymbol{n}} \right\}}\, \prod_{\alpha,\beta=1}^\kappa E_{\alpha\beta}^{a_{\alpha\beta} \left( \boldsymbol{n}; \gamma_{\boldsymbol{n}} \right)} \sum_{\{r_{\boldsymbol{m}}\} \in \left\{ \gamma_{\boldsymbol{n}} \right\}} \varphi^{\sum_{\boldsymbol{m}=\boldsymbol{0}}^{+\infty} r_{\boldsymbol{m}}\, b^{(\boldsymbol{m})}}.
\end{equation}
\noindent
Finally, let us introduce
\begin{equation}
\lambda_b \left( \boldsymbol{n}; \gamma_{\boldsymbol{n}} \right) := \min_{\{r_{\boldsymbol{m}}\} \in \left\{ \gamma_{\boldsymbol{n}} \right\}} \bigg\{ \sum_{\boldsymbol{m}=\boldsymbol{0}}^{+\infty} r_{\boldsymbol{m}}\, b^{(\boldsymbol{m})} \bigg\}.
\end{equation}
\noindent
This is simplest to understand in the one--parameter example with $b^{(n)}=1$. In that case, $\lambda_{b=1}(n;\gamma_n)$ looks at the partitions of $n$ that lie in the class $\gamma_n$, counts the number of parts of each partition in the class, and gives the minimum of those. In the end, given all the manipulations above, what we have shown is that one can write
\begin{equation}\label{eq:expandedgenfunction}
\Phi_b =  \sum_{\boldsymbol{n}=\boldsymbol{0}}^{+\infty} \boldsymbol{\rho}^{\boldsymbol{n}}\, \sum_{\left\{ \gamma_{\boldsymbol{n}} \right\}} E^{a \left( \boldsymbol{n}; \gamma_{\boldsymbol{n}} \right)}\, \varphi^{\lambda_b \left( \boldsymbol{n}; \gamma_{\boldsymbol{n}} \right)} \left( 1 + \mathcal{O} \left( \varphi \right) \right),
\end{equation}
\noindent
where $E^a := \prod_{\alpha,\beta=1}^\kappa E_{\alpha\beta}^{a_{\alpha\beta}}$. On this final expression, note that for every $\boldsymbol{n}$ there is a special class $\widehat{\gamma}_{\boldsymbol{n}}$ whose only element is specified by $r_{\boldsymbol{m}} = \delta_{\boldsymbol{m}\boldsymbol{n}}$, or $a_{\alpha\beta} \left( \boldsymbol{n}; \widehat{\gamma}_{\boldsymbol{n}} \right) = n_\alpha n_\beta$, $\forall \alpha,\beta$. Further notice that $\lambda_b \left( \boldsymbol{n}; \widehat{\gamma}_{\boldsymbol{n}} \right) = b^{(\boldsymbol{n})}$. For the one--parameter case when $b^{(n)}=1$, this special class has $a(n;\gamma_n)=n^2$ and $\lambda(n;\gamma_n)=1$. Recall from our discussion at the end of section \ref{sec:1parameter} that this is precisely where the holomorphic ambiguity $f^{(\boldsymbol{n})}_g$ sits for each $g$.

Let us go back one step and recall that the main idea we are exploring is that the collections of numbers $a_{\alpha\beta} ( \boldsymbol{n}; \gamma_{\boldsymbol{n}} )$ and $\lambda ( \boldsymbol{n}; \gamma_{\boldsymbol{n}} )$ are all the required data which is necessary in order to reconstruct the structure of the nonperturbative sectors $F^{(\boldsymbol{n})}_g$, for every $g$. In order to make this idea precise, the first step needed is to make the identification
\begin{equation}\label{eq:EtoExp}
E_{\alpha\beta} \longleftrightarrow \rme^{\frac{1}{2}\, \partial_z A_\alpha\, \partial_z A_\beta\, S^{zz}},
\end{equation}
\noindent
so that the numbers $a/a_{\alpha\beta}$ in (\ref{eq:Fstructure}) and (\ref{eq:expandedgenfunction}) coincide. This is somewhat similar to what happened in the one--parameter transseries example, in the previous section. The second step will be to prove that
\begin{equation}\label{eq:degreeformula}
d_b \left( \boldsymbol{n}; g; \gamma_{\boldsymbol{n}} \right) = 3 \left( g +b^{(\boldsymbol{n})} - \lambda_b \left( \boldsymbol{n}; \gamma_{\boldsymbol{n}} \right) \right).
\end{equation}
\noindent
Before proving this statement, and as an example to illustrate the general structure, take a transseries with two parameters and restrict to $b^{(\boldsymbol{n})}=1$. Let us determine the structure of the instanton sector $(2|1)$, for each $g$. We need to formally expand $\Phi_{b=1}$ in (\ref{eq:generatingfunction}) and look at terms of order $\rho_1^2 \rho_2^1$. They are
\begin{equation}
\label{eq:rhoEvarphiexample}
E_{11}^4\, E_{12}^{2\cdot 2}\, E_{22}\, \varphi + E_{11}^4\, E_{22}\, \varphi^2 + E_{11}^2\, E_{12}^{2\cdot 1}\, E_{22}\, \varphi^2 + E_{11}^2\, E_{22}\, \varphi^3.
\end{equation}
\noindent
We now have to compare this with (\ref{eq:expandedgenfunction}),
\be
\Phi_{b=1} = \cdots + \rho_1^2 \rho_2^1 \sum_{\{\gamma_{(2|1)}\}} E_{11}^{a_{11}(2|1;\gamma_{(2|1)})}\, E_{12}^{2 a_{12}(2|1;\gamma_{(2|1)})}\, E_{22}^{a_{22}(2|1;\gamma_{(2|1)})}\, \varphi^{\lambda(2|1;\gamma_{(2|1)})} \left( 1+ \mathcal{O}( \varphi ) \right) + \cdots.
\ee
\noindent
We see that there are four classes. The special class is the one with $a_{11}=2\cdot 2=4$, $a_{12}=2\cdot 1=2$, $a_{22}=1\cdot 1 =1$, which corresponds to the first term in (\ref{eq:rhoEvarphiexample}). Note that the value of $\lambda$ at this special class is indeed $b^{(2|1)}=1$, as stated before. Using (\ref{eq:degreeformula}) we can read
\begin{eqnarray}
F^{(2|1)}_g &=& \rme^{[2 \left( \partial_z A_1 \right)^2  + 2\, \partial_z A_1\, \partial_z A_2 + \frac{1}{2} \left( \partial_z A_2 \right)^2] S^{zz}}\, \textrm{Pol} \left( S^{zz}; 3g \right) + \rme^{[2 \left( \partial_z A_1 \right)^2 + \frac{1}{2} \left( \partial_z A_2 \right)^2 ]S^{zz}}\, \textrm{Pol} \left( S^{zz}; 3g-3 \right) +\nonumber \\
&&
\hspace{-20pt}
+ \;\rme^{ [\left( \partial_z A_1 \right)^2 + 2\, \partial_z A_1\, \partial_z A_2 + \frac{1}{2} \left( \partial_z A_2 \right)^2 ]S^{zz}}\, \textrm{Pol} \left( S^{zz}; 3g-3 \right) + \rme^{[\left( \partial_z A_1 \right)^2 + \frac{1}{2} \left( \partial_z A_2 \right)^2 ]S^{zz}}\, \textrm{Pol} \left( S^{zz}; 3g-6 \right).\nonumber\\
&&
\end{eqnarray}
\noindent
Here it is important to note that if the degree of the polynomial is negative, we take it to be identically zero. This means that, for low $g$, we may not see all the exponentials one can write down. In the example above, for $g=0$ we only have the first term, for $g=1$ we have the first three, and for $g\geq 2$ all four terms appear.

We are now ready to state and prove the following theorem:

\begin{thm}\label{structuraltheorem}
For any $\boldsymbol{n}\neq \boldsymbol{0}$ and $g\geq 0$, the structure of the nonperturbative free energies has the form
\begin{equation}\label{eq:1paramStructure}
F^{(\boldsymbol{n})}_g = \sum_{\{\gamma_{\boldsymbol{n}}\}} \rme^{\frac{1}{2} \sum_{\alpha,\beta=1}^\kappa a_{\alpha\beta} \left( \boldsymbol{n}; \gamma_{\boldsymbol{n}} \right) \partial_z A_\alpha\, \partial_z A_\beta\, S^{zz}}\, \mathrm{Pol} \left( S^{zz}; 3 \left( g + b^{(\boldsymbol{n})} - \lambda_b \left( \boldsymbol{n}; \gamma_{\boldsymbol{n}} \right) \right) \right),
\end{equation}
\noindent
where the set of numbers $\left\{ a_{\alpha\beta} \left( \boldsymbol{n}; \gamma_{\boldsymbol{n}} \right) \right\}$ and $\left\{ \lambda_b \left( \boldsymbol{n}; \gamma_{\boldsymbol{n}} \right) \right\}$ are read off from the generating function\footnote{Because we are not focusing on the coefficients of the generating function, but only on the combinatorial structure of the formal variables, note that this generating function is actually not unique.}
\bea
\Phi_b &=& \sideset{}{'}\prod_{\boldsymbol{m}=\boldsymbol{0}}^{+\infty} \frac{1}{1 - \varphi^{b^{(\boldsymbol{m})}}\, \prod_{\alpha,\beta=1}^\kappa E_{\alpha\beta}^{m_\alpha m_\beta}\, \prod_{\alpha=1}^\kappa \rho_\alpha^{m_\alpha}} = \nonumber \\
&=&\sum_{\boldsymbol{n}=\boldsymbol{0}}^{+\infty} \boldsymbol{\rho}^{\boldsymbol{n}}\, \sum_{\left\{ \gamma_{\boldsymbol{n}} \right\}} \prod_{\alpha,\beta=1}^\kappa E_{\alpha\beta}^{a_{\alpha\beta} \left( \boldsymbol{n}; \gamma_{\boldsymbol{n}} \right)}\, \varphi^{\lambda_b \left( \boldsymbol{n}; \gamma_{\boldsymbol{n}} \right)} \left( 1+ \mathcal{O} \left( \varphi \right) \right).
\label{eq:thmcurlyFstructure}
\eea
\noindent
In here $\mathrm{Pol} \left(S^{zz};d\right)$ stands for a polynomial of degree $d$ in the variable $S^{zz}$ (and whose coefficients have a holomorphic dependence on $z$). Whenever $d<0$, the polynomial is taken to be identically zero. We are assuming that $b^{(\boldsymbol{m})} + b^{(\boldsymbol{n}-\boldsymbol{m})} - b^{(\boldsymbol{n})}\geq 0$.
\end{thm}
\noindent
In order to prove this theorem we will need to make use of the following lemma, whose proof we have included in appendix \ref{ap:proofofrecursionlemma}.

\begin{lem}\label{lem:setrecurrence}
The set of numbers $\left\{ a_{\alpha\beta} \left( \boldsymbol{n}; \gamma_{\boldsymbol{n}} \right) \right\}$ and $\left\{ \lambda_b \left( \boldsymbol{n}; \gamma_{\boldsymbol{n}} \right)Ê\right\}$, and the range of the labels $\left\{ \gamma_{\boldsymbol{n}} \right\}$, which appear in (\ref{eq:thmcurlyFstructure}) are determined by the recursions
\begin{equation}\label{eq:arecursion}
\left\{ a_{\alpha\beta} \left( \boldsymbol{n}; \gamma_{\boldsymbol{n}} \right) \right\}_{\gamma_{\boldsymbol{n}} \neq \widehat{\gamma}_{\boldsymbol{n}}} = \sideset{}{'}\bigcup_{\boldsymbol{m}, \gamma_{\boldsymbol{m}}, \gamma_{\boldsymbol{n}-\boldsymbol{m}}} \left\{ a_{\alpha\beta} \left( \boldsymbol{m}; \gamma_{\boldsymbol{m}} \right) + a_{\alpha\beta} \left( \boldsymbol{n}-\boldsymbol{m}; \gamma_{\boldsymbol{n}-\boldsymbol{m}} \right) \right\}
\end{equation}
\noindent
and
\begin{equation}\label{eq:lambdarecursion}
\lambda_b \left( \boldsymbol{n}; \gamma_{\boldsymbol{n}} \right)= \min \left\{ \lambda_b \left( \boldsymbol{m}; \gamma_{\boldsymbol{m}} \right) + \lambda_b \left( \boldsymbol{n}-\boldsymbol{m}; \gamma_{\boldsymbol{n}-\boldsymbol{m}} \right) \right\}, \qquad \forall \gamma_{\boldsymbol{n}} \neq \widehat{\gamma}_{\boldsymbol{n}},
\end{equation}
\noindent
where $\min$ ranges over $\boldsymbol{m} \in \{\boldsymbol{0},\ldots,\boldsymbol{n}\}'$, and $\gamma_{\boldsymbol{m}}$, $\gamma_{\boldsymbol{n}-\boldsymbol{m}}$ are such that $a_{\alpha\beta} \left( \boldsymbol{m};\gamma_{\boldsymbol{m}} \right) + a_{\alpha\beta}\left( \boldsymbol{n}-\boldsymbol{m}; \right.$ $\left. \gamma_{\boldsymbol{n}-\boldsymbol{m}} \right) = a_{\alpha\beta} \left( \boldsymbol{n};\gamma_{\boldsymbol{n}} \right)$, $\forall \alpha,\beta$. As always, the prime means that $\boldsymbol{m}=\boldsymbol{0}$ and $\boldsymbol{m}=\boldsymbol{n}$ are excluded. Further, we have to specify the initial data:
\begin{eqnarray}
\label{eq:ainitialdata}
a_{\alpha\beta} \left( \boldsymbol{n}; \widehat{\gamma}_{\boldsymbol{n}} \right) &=& n_\alpha n_\beta, \qquad \forall\, \boldsymbol{n},\alpha,\beta, \\
\label{eq:lambdainitialdata}
\lambda_b \left( \boldsymbol{n}; \widehat{\gamma}_{\boldsymbol{n}} \right) &=& b^{(\boldsymbol{n})}, \qquad \forall\, \boldsymbol{n}.
\end{eqnarray}
\end{lem}

\begin{proof}[Proof of the Theorem]
The proof uses induction in $\boldsymbol{n}$ and $g$. Let us assume that the structure of the free energies, $F^{(\boldsymbol{m})}_h$ is known for $\boldsymbol{m}\leq \boldsymbol{n}$ and $h<g$. Then we can calculate the complete anti--holomorphic structure of $G^{(\boldsymbol{n})}_g$ appearing in \eqref{eq:integratingG}, and thus integrate the holomorphic anomaly equations. In order to see how this may be explicitly done recall the relation \eqref{eq:covariantderivativeonpropagator}: a $z$--derivative on a polynomial raises the degree by $1$. Also, if the $z$--derivative acts on an exponential such as \eqref{eq:EtoExp}, the result is the same exponential times a polynomial of degree $2$. If we acted twice with the derivative, the degree would be $4$. We will also use that $\partial_z F^{(\boldsymbol{0})}_h$ is a polynomial of degree $3h-2$ for $h\geq 1$. Finally, in order to ease the notation a little, we are going to use $E_{\alpha\beta}$ instead of $\exp \frac{1}{2}\, \partial_z A_\alpha \partial_z A_\beta\, S^{zz}$ and we are going to drop the $S^{zz}$ in $\textrm{Pol} \left(S^{zz};d\right)$ since it is always understood.

Recall from \eqref{eq:NPeq} that $G^{(\boldsymbol{n})}_g$ is composed of two terms. The first may be calculated as
\be
\sum_{h=1}^g \mathcal{D}^{(\boldsymbol{n})}_h F^{(\boldsymbol{n})}_{g-h} = \sum_{\{\gamma_{\boldsymbol{n})}\}} E^{a \left( \boldsymbol{n};\gamma_{\boldsymbol{n}} \right)}\, \textrm{Pol} \left( 3\left( g+b^{(n)}-\lambda_b \left(\boldsymbol{n};\gamma_{\boldsymbol{n}} \right) \right) -1 \right),
\ee
\noindent
where the dominant contribution arises from the term $\mathcal{D}^{(\boldsymbol{n})}_1 F^{(\boldsymbol{n})}_{g-1}$, that is, it is the contribution with the highest degree polynomial accompanying each exponential.

The second term in $G^{(\boldsymbol{n})}_g$ is not so straightforward as it is given by a product. Let us evaluate it explicitly,
\begin{eqnarray}
&&
\sideset{}{'}\sum_{\boldsymbol{m}=\boldsymbol{0}}^{\boldsymbol{n}} \sum_{h=0}^{g-B(\boldsymbol{n},\boldsymbol{m})} \left( \partial_z F^{(\boldsymbol{m})}_{h-1} - \partial_z A^{(\boldsymbol{m})} F^{(\boldsymbol{m})}_h \right) \left( \partial_z F^{(\boldsymbol{n}-\boldsymbol{m})}_{g-1-B(\boldsymbol{n},\boldsymbol{m})-h} - \partial_z A^{(\boldsymbol{n}-\boldsymbol{m})} F^{(\boldsymbol{n}-\boldsymbol{m})}_{g-B(\boldsymbol{n},\boldsymbol{m})-h} \right) = \nonumber \\
&=& \sideset{}{'}\sum_{\boldsymbol{m}=\boldsymbol{0}}^{\boldsymbol{n}}\, \sum_{h=0}^{g-B} \left\{ \sum_{\left\{ \gamma_{\boldsymbol{m}} \right\}} E^{a \left( \boldsymbol{m}; \gamma_{\boldsymbol{m}} \right)}\, \textrm{Pol} \left( 3 \left( h + b^{(\boldsymbol{m})} - \lambda_b \left( \boldsymbol{m}; \gamma_{\boldsymbol{m}} \right) \right) \right) \right\} \times \label{eq:proofquadraticterm} \\
&&
\hspace{5pt}
\times \left\{ \sum_{\left\{ \gamma_{\boldsymbol{n}-\boldsymbol{m}} \right\}} E^{a \left( \boldsymbol{n}-\boldsymbol{m}; \gamma_{\boldsymbol{n}-\boldsymbol{m}} \right)}\,  \textrm{Pol} \left( 3 \left( g - h - B(\boldsymbol{n},\boldsymbol{m}) + b^{(\boldsymbol{n}-\boldsymbol{m})} - \lambda_b \left( \boldsymbol{n}-\boldsymbol{m}; \gamma_{\boldsymbol{n}-\boldsymbol{m}} \right) \right) \right) \right\} = \nonumber \\
&=& \sideset{}{'}\sum_{\boldsymbol{m}=\boldsymbol{0}}^{\boldsymbol{n}} \sum_{\left\{ \gamma_{\boldsymbol{m}}, \gamma_{\boldsymbol{n}-\boldsymbol{m}} \right\}} E^{a \left( \boldsymbol{m}; \gamma_{\boldsymbol{m}} \right) + a \left( \boldsymbol{n}-\boldsymbol{m}; \gamma_{\boldsymbol{n}-\boldsymbol{m}} \right)}\, \textrm{Pol} \left( 3 \left( g + b^{(\boldsymbol{n})} - \lambda_b \left( \boldsymbol{m}; \gamma_{\boldsymbol{m}} \right) - \lambda_b \left( \boldsymbol{n}-\boldsymbol{m}; \gamma_{\boldsymbol{n}-\boldsymbol{m}} \right) \right) \right), \nonumber
\end{eqnarray}
\noindent
where we have used the definition of $B(\boldsymbol{n},\boldsymbol{m})$ and dropped the sum in $h$ since there is no dependence on it anymore. Now one makes use of Lemma \ref{lem:setrecurrence} and obtains the final result,
\begin{equation}
(\ref{eq:proofquadraticterm}) = \sideset{}{'}\sum_{\left\{ \gamma_{\boldsymbol{n}} \right\}} E^{a \left( \boldsymbol{n}; \gamma_{\boldsymbol{n}} \right)}\, \textrm{Pol} \left( 3 \left( g + b^{(\boldsymbol{n})} - \lambda_b \left( \boldsymbol{n}; \gamma_{\boldsymbol{n}} \right) \right) \right),
\end{equation}
\noindent
where the prime means that $\widehat{\gamma}_{\boldsymbol{n}}$ is not included in the sum. Putting both ingredients of $G^{(\boldsymbol{n})}_g$ together, one arrives at
\bea
G^{(\boldsymbol{n})}_g &=& E^{a \left( \boldsymbol{n}; \widehat{\gamma}_{\boldsymbol{n}} \right)}\, \textrm{Pol} \left( 3 \left( g + b^{(\boldsymbol{n})} - \lambda_b \left( \boldsymbol{n}; \widehat{\gamma}_{\boldsymbol{n}} \right) \right) - 1 \right) + \nonumber \\
&&
\hspace{80pt}
+ \sideset{}{'}\sum_{\left\{ \gamma_{\boldsymbol{n}} \right\}} E^{a \left( \boldsymbol{n}; \gamma_{\boldsymbol{n}} \right)}\, \textrm{Pol} \left( 3 \left( g + b^{(\boldsymbol{n})} - \lambda_b \left( \boldsymbol{n}; \gamma_{\boldsymbol{n}} \right) \right) \right),
\label{eq:thmalmostGng}
\eea
\noindent
where we have explicitly separated the $\widehat{\gamma}_{\boldsymbol{n}}$ term, and assembled all others together. This separation is convenient since
\begin{equation}
E^{a \left( \boldsymbol{n}; \widehat{\gamma}_{\boldsymbol{n}} \right)} = \rme^{\frac{1}{2} \sum_{\alpha,\beta=1}^\kappa a_{\alpha\beta} \left( \boldsymbol{n}; \widehat{\gamma}_{\boldsymbol{n}} \right) \partial_z A_\alpha\, \partial_z A_\beta\, S^{zz}} = \rme^{\frac{1}{2} \sum_{\alpha,\beta=1}^\kappa n_\alpha n_\beta\, \partial_z A_\alpha\, \partial_z A_\beta\, S^{zz}} = \rme^{\frac{1}{2} \left( \partial_z A^{(\boldsymbol{n})} \right)^2 S^{zz}},
\end{equation}
\noindent
and what we have to integrate with respect to $S^{zz}$ is precisely $\rme^{-\frac{1}{2} \left( \partial_z A^{(\boldsymbol{n})} \right)^2 S^{zz}}\, G^{(\boldsymbol{n})}_g$. Consequently, the first term in (\ref{eq:thmalmostGng}) becomes a polynomial within the integral, while the others keep an exponential part. In fact, using that
\bea
\int^{S^{zz}} \rmd \widetilde{S}^{zz}\, \textrm{Pol}\, ( \widetilde{S}^{zz};d ) &=& \textrm{Pol} \left( S^{zz};d+1 \right), \\ 
\int^{S^{zz}} \rmd \widetilde{S}^{zz}\, \rme^{c\, \widetilde{S}^{zz}}\, \textrm{Pol}\, ( \widetilde{S}^{zz};d ) &=& \rme^{c \, S^{zz}}\, \textrm{Pol} \left( S^{zz};d \right), \qquad c \neq 0,
\eea
\noindent
we finally obtain the result we were looking for
\begin{eqnarray}
F^{(\boldsymbol{n})}_g &=& \rme^{\frac{1}{2} \left( \partial_z A^{(\boldsymbol{n})} \right)^2 S^{zz}} \left( f^{(\boldsymbol{n})}_g + \int^{S^{zz}} \rmd \widetilde{S}^{zz}\, \rme^{-\frac{1}{2} \left( \partial_z A^{(\boldsymbol{n})} \right)^2 \widetilde{S}^{zz}}\, G^{(\boldsymbol{n})}_g (z,\widetilde{S}^{zz}) \right) = \nonumber \\
&=& E^{a \left( \boldsymbol{n}; \widehat{\gamma}_{\boldsymbol{n}} \right)} \Bigg( f^{(\boldsymbol{n})}_g + \textrm{Pol} \left( 3 \left( g + b^{(\boldsymbol{n})} - \lambda_b \left( \boldsymbol{n}; \widehat{\gamma}_{\boldsymbol{n}} \right) \right) - 1 + 1 \right) + \nonumber \\
&&
\hspace{110pt}
+ \sideset{}{'}\sum_{\left\{ \gamma_{\boldsymbol{n}} \right\}} E^{a \left( \boldsymbol{n}; \gamma_{\boldsymbol{n}} \right) - a \left( \boldsymbol{n}; \widehat{\gamma}_{\boldsymbol{n}} \right)}\, \textrm{Pol} \left( 3 \left( g + b^{(\boldsymbol{n})} - \lambda_b \left( \boldsymbol{n}; \gamma_{\boldsymbol{n}} \right) \right) \right) \Bigg) = \nonumber \\
&=& \sum_{\left\{ \gamma_{\boldsymbol{n}} \right\}} E^{a \left( \boldsymbol{n}; \gamma_{\boldsymbol{n}} \right)}\, \textrm{Pol} \left( 3 \left( g + b^{(\boldsymbol{n})} - \lambda_b \left( \boldsymbol{n}; \gamma_{\boldsymbol{n}} \right) \right) \right).
\end{eqnarray}

To conclude the proof one has to check the induction base case. For $g=0$ and any $\boldsymbol{n}$,
\be\label{eq:basecaseeq}
\left( \partial_{S^{zz}} - \frac{1}{2} \left( \partial_z A^{(\boldsymbol{n})} \right)^2 \right) F^{(\boldsymbol{n})}_0 = \frac{1}{2} \sideset{}{'}\sum_{\substack{ \boldsymbol{m} = \boldsymbol{0}, \\ B(\boldsymbol{n},\boldsymbol{m})=0 }}^{\boldsymbol{n}} \partial_z A^{(\boldsymbol{m})}\, \partial_z A^{(\boldsymbol{n}-\boldsymbol{m})}\, F^{(\boldsymbol{m})}_0\, F^{(\boldsymbol{n}-\boldsymbol{m})}_0.
\ee
\noindent
For one--instanton sectors, of the form $(0| \cdots | 1 | \cdots |0)$, the right--hand--side of \eqref{eq:basecaseeq} is zero and we can integrate directly
\be
F^{(0| \cdots | 1 | \cdots |0)}_0 = \rme^{\frac{1}{2} \left(\partial_z A^{(0| \cdots | 1 | \cdots |0)}\right)^2 S^{zz}}\, f^{(0| \cdots | 1 | \cdots |0)}_0 = E^{a \left( (0| \cdots | 1 | \cdots |0); \gamma_{(0| \cdots | 1 | \cdots |0)} \right)}\, \textrm{Pol} \left(S^{zz};0\right),
\ee
\noindent
which is exactly the statement of the theorem \eqref{eq:1paramStructure} for the one--instanton sector and $g=0$. Indeed, because it is an one--instanton sector, there is only one $\gamma_{\boldsymbol{n}}$ precisely given by $\widehat{\gamma}_{\boldsymbol{n}}$. For that class $\sum_{\alpha\beta} a_{\alpha\beta}\, \partial_z A_\alpha\, \partial_z A_\beta = \sum_{\alpha\beta} n_\alpha n_\beta\, \partial_z A_\alpha\, \partial_z A_\beta = \left( \partial_z A^{(\boldsymbol{n})}\right)^2$, where we have used \eqref{eq:ainitialdata}. Further using \eqref{eq:lambdainitialdata} we obtain $\textrm{Pol} \left(S^{zz};0\right)$. Now one uses induction in the instanton number $\boldsymbol{n}$ in order to show that
\be\label{eq:basecasesol}
F^{(\boldsymbol{n})}_0 = \sum_{\substack{\{\gamma_{\boldsymbol{n}}\} \\ \lambda_b \left( \boldsymbol{n};\gamma_{\boldsymbol{n}} \right) = b^{(\boldsymbol{n})}}} E^{a \left( \boldsymbol{n};\gamma_{\boldsymbol{n}} \right)}\, \textrm{Pol} \left(S^{zz};0\right),
\ee
\noindent
which is the expression \eqref{eq:1paramStructure} when $g=0$ (recall that if the degree of $\textrm{Pol}$ is negative the polynomial is taken to be zero). Now assume that \eqref{eq:basecasesol} is true for $\boldsymbol{m} < \boldsymbol{n}$. Because\footnote{If we define $\lambda_B \left( \boldsymbol{n};\gamma_{\boldsymbol{n}} \right) := \lambda_b \left( \boldsymbol{n};\gamma_{\boldsymbol{n}} \right) - b^{(\boldsymbol{n})}$, then \eqref{eq:lambdarecursion} and \eqref{eq:lambdainitialdata} become
\be
\lambda_B \left( \boldsymbol{n};\gamma_{\boldsymbol{n}} \right) = \min \left\{ \lambda_B \left( \boldsymbol{m};\gamma_{\boldsymbol{m}} \right) + \lambda_B \left( \boldsymbol{n}-\boldsymbol{m};\gamma_{\boldsymbol{n}-\boldsymbol{m}} \right) + B \left( \boldsymbol{n},\boldsymbol{m} \right) \right\},
\ee
\noindent
with $\lambda_B \left( \boldsymbol{n};\hat{\gamma}_{\boldsymbol{n}} \right) = 0$. Since we assume $B\geq 0$, $\forall \boldsymbol{n},\boldsymbol{m}$, induction shows that $\lambda_B \left( \boldsymbol{n};\gamma_{\boldsymbol{n}} \right) \geq 0$.} $\lambda_b \geq b^{(\boldsymbol{n})}$, we can write $F^{(\boldsymbol{m})}_0$ as
\be
\sum_{\{\gamma_{\boldsymbol{m}}\}} E^{a \left( \boldsymbol{m};\gamma_{\boldsymbol{m}} \right)}\, \textrm{Pol} \left( S^{zz};b^{(\boldsymbol{m})} - \lambda_b \left( \boldsymbol{m};\gamma_{\boldsymbol{m}} \right) \right).
\ee
Then, the right--hand--side of \eqref{eq:basecaseeq} is
\bea
&&
\sideset{}{'}\sum_{\substack{ \boldsymbol{m} = \boldsymbol{0} \\ B(\boldsymbol{n},\boldsymbol{m})=0 }}^{\boldsymbol{n}}\, \sum_{\{\gamma_{\boldsymbol{m}}\},\{\gamma_{\boldsymbol{n}-\boldsymbol{m}}\}} E^{a \left( \boldsymbol{m};\gamma_{\boldsymbol{m}} \right) + a \left( \boldsymbol{n}-\boldsymbol{m};\gamma_{\boldsymbol{n}-\boldsymbol{m}} \right)} \times \nonumber \\
&&
\hspace{100pt}
\times\, \textrm{Pol} \left( b^{(\boldsymbol{m})} + b^{(\boldsymbol{n}-\boldsymbol{m})} - \lambda_b \left( \boldsymbol{m};\gamma_{\boldsymbol{m}} \right) - \lambda_b \left( \boldsymbol{n}-\boldsymbol{m};\gamma_{\boldsymbol{n}-\boldsymbol{m}} \right) \right) = \nonumber\\
&&
= \sideset{}{'}\sum_{\{\gamma_{\boldsymbol{n}}\}} E^{a \left( \boldsymbol{n};\gamma_{\boldsymbol{n}} \right)}\, \textrm{Pol} \left( b^{(\boldsymbol{n})} - \lambda_b(\boldsymbol{n};\gamma_{\boldsymbol{n}}) \right) = \sideset{}{'}\sum_{\substack{\{\gamma_{\boldsymbol{n}}\} \\ \lambda_b \left( \boldsymbol{n};\gamma_{\boldsymbol{n}} \right) = b^{(\boldsymbol{n})}}} E^{a \left( \boldsymbol{n};\gamma_{\boldsymbol{n}} \right)}\, \textrm{Pol} \left(0\right),
\eea
\noindent
where we have used that $B(\boldsymbol{n},\boldsymbol{m})=0$ means $b^{(\boldsymbol{m})}+b^{(\boldsymbol{n}-\boldsymbol{m})}=b^{(\boldsymbol{n})}$, and the recursions \eqref{eq:arecursion}, \eqref{eq:lambdarecursion}. Integrating \eqref{eq:basecaseeq} concludes the proof of the base case and the proof of the theorem.
\end{proof}

A generalization of the results in this section to the case of Calabi--Yau geometries with multi--dimensional complex moduli space is included in appendix \ref{app:multidimcomplmodspace}.

\subsection{Resonance and Logarithmic Sectors}\label{sec:resonanceandlogarithmicsectors}

As we have discussed in section \ref{sec:transseriesintro}, in some situations our models might be resonant. Exactly when does a Calabi--Yau geometry yields a resonant solution is an interesting question for future research; in here we shall only be interested in assuming resonance and investigating its consequences on the solutions to the (nonperturbative) holomorphic anomaly equations. In the matrix model setting we already know that we should expect logarithmic sectors and a ``reflective'' pattern within the instanton actions. In the following we will assume that resonance is still manifested by logarithmic blocks in the transseries, but in fact any other non--analytic function of $g_s$ could take up their place and the calculation would still go through. In this subsection we shall explore what can be said about these logarithmic sectors.

Recall that the generalization of the transseries \textit{ansatz} \eqref{eq:preNPansatz}, \eqref{eq:multiparameteransatz}, to incorporate logarithmic sectors is quite straightforward:
\be\label{eq:preNPlogansatz}
F (\boldsymbol{\sigma},g_s) = \sum_{\boldsymbol{n} \in \mathbb{N}_0^\kappa} \boldsymbol{\sigma}^{\boldsymbol{n}}\, \rme^{-A^{(\boldsymbol{n})}/g_s} \sum_{k=0}^{k_{\text{max}}^{(\boldsymbol{n})}} \log^k \left(g_s\right)\, F^{(\boldsymbol{n})[k]}(g_s),
\ee
\noindent
with
\be\label{eq:multiparameterlogansatz}
F^{(\boldsymbol{n})[k]}(g_s) \simeq \sum_{g=0}^{+\infty} g_s^{g+b^{(\boldsymbol{n})[k]}}F^{(\boldsymbol{n})[k]}_g.
\ee
\noindent
Here $k_{\text{max}}$ is a function that specifies the maximum power of $\log \left(g_s\right)$ that could appear in each instanton sector. As it happened with the (generalized) starting order $b^{(\boldsymbol{n})[k]}$, $k_{\text{max}}$ is not fixed by the holomorphic anomaly equations. We will assume some generic condition for $k_{\text{max}}^{(\boldsymbol{n})}$ to carry on with the calculations (see Theorem \ref{thm:multiparameterlogstructure}) and leave the special cases for a separate analysis.

If we plug \eqref{eq:preNPlogansatz} in the holomorphic anomaly equation we find
\bea
\sum_{k=0}^{k_{\text{max}}^{(\boldsymbol{n})}} \log^k \left(g_s\right) \mathcal{D}^{(\boldsymbol{n})} F^{(\boldsymbol{n})[k]} &=&  \frac{1}{2} g_s^2\, \sideset{}{'}\sum_{\boldsymbol{m}=\boldsymbol{0}}^{\boldsymbol{n}}\, \sum_{k=0}^{k_{\text{max}}^{(\boldsymbol{m})} + k_{\text{max}}^{(\boldsymbol{n}-\boldsymbol{m})}} \log^k \left(g_s\right) \times \\
&&
\hspace{-90pt}
\times \sum_{\ell = \max \left( 0, k - k_{\text{max}}^{(\boldsymbol{n}-\boldsymbol{m})} \right)}^{\min \left( k, k_{\text{max}}^{(\boldsymbol{m})} \right)} \left(\partial_z -\frac{1}{g_s} \partial_z A^{(\boldsymbol{m})} \right) F^{(\boldsymbol{m})[\ell]} \left(\partial_z -\frac{1}{g_s} \partial_z A^{(\boldsymbol{n}-\boldsymbol{m})} \right) F^{(\boldsymbol{n}-\boldsymbol{m})[k-\ell]}. \nonumber
\eea
\noindent
Let us focus on the situation in which $k_{\text{max}}^{(\boldsymbol{n})} \geq  k_{\text{max}}^{(\boldsymbol{m})} + k_{\text{max}}^{(\boldsymbol{n}-\boldsymbol{m})}$. This case avoids the appearance of any extra constraints, and it is cleaner to handle. It then follows
\be\label{eq:logpreeq}
\mathcal{D}^{(\boldsymbol{n})} F^{(\boldsymbol{n})[k]} = \sideset{}{'}\sum_{\boldsymbol{m}=\boldsymbol{0}}^{\boldsymbol{n}} \sum_{\ell = \max \left( 0, k-k_{\text{max}}^{(\boldsymbol{n}-\boldsymbol{m})} \right)}^{\min \left( k, k_{\text{max}}^{(\boldsymbol{m})} \right)} \left(\partial_z -\frac{1}{g_s} \partial_z A^{(\boldsymbol{m})} \right) F^{(\boldsymbol{m})[\ell]} \left(\partial_z -\frac{1}{g_s} \partial_z A^{(\boldsymbol{n}-\boldsymbol{m})} \right) F^{(\boldsymbol{n}-\boldsymbol{m})[k-\ell]},
\ee
\noindent
for $\boldsymbol{n}\neq 0$, $k=0,1,2, \ldots, k_{\text{max}}^{(\boldsymbol{n})}$. Using the asymptotic expansion \eqref{eq:multiparameterlogansatz} we obtain
\begin{eqnarray}
-\partial_{S^{zz}} A^{(\boldsymbol{n})}\, F^{(\boldsymbol{n})[k]}_{g+1} + \sum_{h=0}^g \mathcal{D}^{(\boldsymbol{n})}_h F^{(\boldsymbol{n})[k]}_{g-h} &=&  \frac{1}{2}\, \sideset{}{'}\sum_{\boldsymbol{m}=\boldsymbol{0}}^{\boldsymbol{n}}\,\sum_{\ell} \sum_{h=0}^{g-B(\boldsymbol{n},\boldsymbol{m})[k,\ell]} \left( \partial_z F^{(\boldsymbol{m})[\ell]}_{h-1} - \partial_z A^{(\boldsymbol{m})}\, F^{(\boldsymbol{m})[\ell]}_h \right) \times \nonumber \\
&&
\hspace{-35pt}
\times \left( \partial_z F^{(\boldsymbol{n}-\boldsymbol{m})[k-\ell]}_{g-1-B(\boldsymbol{n},\boldsymbol{m})[k,\ell]-h} - \partial_z A^{(\boldsymbol{n}-\boldsymbol{m})}\, F^{(\boldsymbol{n}-\boldsymbol{m})[k-\ell]}_{g-B(\boldsymbol{n},\boldsymbol{m})[k,\ell]-h} \right),
\label{eq:almostNPlogeq}
\end{eqnarray}
\noindent
where the sum in $\ell$ has the same limits as in the previous equations, and
\be
B(\boldsymbol{n},\boldsymbol{m})[k,\ell] := b^{(\boldsymbol{m})[\ell]}+b^{(\boldsymbol{n}-\boldsymbol{m})[k-\ell]} - b^{(\boldsymbol{n})[k]}.
\ee
\noindent
As usual we shall assume that $B\geq 0$ in the following.

Once more, if we look at the one--instanton sectors $(0| \cdots |1| \cdots | 0)$, $g=-1$, and also $k=0$ for simplicity, we find that the instanton actions are holomorphic,
\be
\partial_{S^{zz}} A_\alpha = 0, \qquad \forall\alpha.
\ee
\noindent
It is important that the equations can be used to conclude the holomorphicity of $A_\alpha$ in every circumstance (in fact this is also true even if $B$ is not greater or equal to zero for some sectors).

Now notice that the equations \eqref{eq:almostNPlogeq} are only slightly more complicated than those in \eqref{eq:almostNPeq}, so it is natural to expect a straightforward generalization of both Theorem \ref{structuraltheorem} and the counterpart of Lemma \ref{lem:setrecurrence}. This is in fact true as we shall now show.

\begin{thm}\label{thm:multiparameterlogstructure}
For any $\boldsymbol{n}\neq \boldsymbol{0}$, $k \in \{ 0,\ldots,k_{\rm{max}}^{(\boldsymbol{n})} \}$, and $g\geq 0$, the structure of the nonperturbative free energies has the form
\begin{equation}\label{eq:1paramLogStructure}
F^{(\boldsymbol{n})[k]}_g = \sum_{\{\gamma_{\boldsymbol{n}}\}} \rme^{\frac{1}{2}\sum_{\alpha,\beta=1}^\kappa a_{\alpha\beta} \left( \boldsymbol{n}; \gamma_{\boldsymbol{n}} \right) \partial_z A_\alpha\, \partial_z A_\beta\, S^{zz}}\, \mathrm{Pol} \left( S^{zz}; 3 \left( g + b^{(\boldsymbol{n})[k]} - \lambda_{b,k_{\rm{max}}}^{[k]} \left( \boldsymbol{n}; \gamma_{\boldsymbol{n}} \right) \right) \right),
\end{equation}
\noindent
where the set of numbers $\left\{ a_{\alpha\beta} \left( \boldsymbol{n}; \gamma_{\boldsymbol{n}} \right) \right\}$ and $\left\{ \lambda_{b,k_{\rm{max}}}^{[k]} \left( \boldsymbol{n}; \gamma_{\boldsymbol{n}} \right) \right\}$ are read off from the generating function
\bea
\Phi_{b,k_{\rm{max}}} &=& \sideset{}{'}\prod_{\boldsymbol{m}=\boldsymbol{0}}^{+\infty} \, \prod_{\ell=0}^{k_{\rm{max}}^{(\boldsymbol{m})}} \frac{1}{1 - \varphi^{b^{(\boldsymbol{m})[\ell]}}\, \psi^{\ell}\, \prod_{\alpha,\beta=1}^\kappa E_{\alpha\beta}^{m_\alpha m_\beta}\, \prod_{\alpha=1}^\kappa \rho_\alpha^{m_\alpha}} = \nonumber \\
&=& \sum_{\boldsymbol{n}=\boldsymbol{0}}^{+\infty} \boldsymbol{\rho}^{\boldsymbol{n}}\, \sum_{k=0}^{k_{\rm{max}}^{(\boldsymbol{n})}} \psi^k\, \sum_{\left\{ \gamma_{\boldsymbol{n}} \right\}} \prod_{\alpha,\beta=1}^\kappa E_{\alpha\beta}^{a_{\alpha\beta} \left( \boldsymbol{n}; \gamma_{\boldsymbol{n}} \right)}\, \varphi^{\lambda_{b,k_{\rm{max}}}^{[k]} \left( \boldsymbol{n}; \gamma_{\boldsymbol{n}} \right)} \left( 1+ \mathcal{O} \left( \varphi \right) \right).
\eea
\noindent
In here $\mathrm{Pol} \left(S^{zz};d\right)$ stands for a polynomial of degree $d$ in the variable $S^{zz}$ (and whose coefficients have a holomorphic dependence on $z$). Whenever $d<0$, the polynomial is taken to be identically zero. We are assuming that $b^{(\boldsymbol{m})[\ell]} + b^{(\boldsymbol{n}-\boldsymbol{m})[k-\ell]} - b^{(\boldsymbol{n})[k]}\geq 0$, and $k_{\rm{max}}^{(\boldsymbol{n})} - k_{\rm{max}}^{(\boldsymbol{m})} - k_{\rm{max}}^{(\boldsymbol{n}-\boldsymbol{m})} \geq 0$.
\end{thm}

Note in particular that the exponential structure of the free energies has no dependence on $b^{(\boldsymbol{n})[k]}$ or $k_{\text{max}}^{(\boldsymbol{n})}$, since this only comes from the quadratic nature of the differential equation and the differential operator $\mathcal{D}$. The proof of this theorem is completely analogous to the proof for Theorem \ref{structuraltheorem}, modulo the extra index $[k]$, and thus we omit it.

Throughout this section we have been explicitly constructing the anti--holomorphic structure of rather general transseries \textit{ans\"atze}. The appearance of exponentials and polynomials, specified by combinatorial data, represents the generic situation one may encounter in many particular examples. Nevertheless, models enjoying certain extra symmetries may be expected to depart from this generic situation, which means that Theorem \ref{thm:multiparameterlogstructure} may not apply directly. A very important example is the one in which there are generalized instanton sectors as in \cite{gikm10, asv11, sv13}, \textit{i.e.}, we find instanton actions $A_1 = - A_2$ or generalizations of this structure. We shall discuss this case in detail in the next subsection. Another obvious  situation where Theorem \ref{thm:multiparameterlogstructure} may have to be modified is when the starting genus $b^{(\boldsymbol{n})[k]}$ or the function $k_{\text{max}}^{(\boldsymbol{n})}$ do not satisfy the assumptions of the theorem (specified at the end). When this happens, equations \eqref{eq:NPeq} and \eqref{eq:logpreeq} give constraints on the free energies which have already been calculated in the recursion, or they give extra conditions on the instanton actions. The reason is simply because in those equations the left--hand--side is zero while the right--hand--side is not. Note that this is not any inconsistency; particular models are indeed expected to satisfy certain extra relations between their free energies and to have special values for the instanton actions, across moduli space, which are compatible with these constraints. In such cases Theorem \ref{thm:multiparameterlogstructure} will not provide the exact anti--holomorphic structure for the sectors satisfying the extra constraints, and one should be on the lookout for these effects.

Let us discuss some of these aforementioned possibilities. For example, in a logarithm--free transseries for which the starting genus satisfies $B(\boldsymbol{n},\boldsymbol{m}) = b^{(\boldsymbol{m})} + b^{(\boldsymbol{n}-\boldsymbol{m})} - b^{(\boldsymbol{n})} < 0$, for some $\boldsymbol{n}$ and one particular sector $\boldsymbol{m}$, one will have an equation coming from \eqref{eq:NPeq} which is
\be
0 = \partial_z A^{(\boldsymbol{m})}\, \partial_z A^{(\boldsymbol{n}-\boldsymbol{m})}\,  F^{(\boldsymbol{m})}_0\, F^{(\boldsymbol{n}-\boldsymbol{m})}_0.
\ee
\noindent
This means that one combination of instanton actions is \textit{constant} (recall $F^{(\boldsymbol{m})}_0 \neq 0$ by construction). Something similar happens for a transseries in which $k_{\text{max}}^{(\boldsymbol{n})} < k_{\text{max}}^{(\boldsymbol{m})} + k_{\text{max}}^{(\boldsymbol{n}-\boldsymbol{m})}$, for some $\boldsymbol{n}$ and $\boldsymbol{m}$. For example, if, for simplicity, it is only one sector $\boldsymbol{m}$ satisfying this inequality, \eqref{eq:logpreeq} will yield
\be
0 = \left( \partial_z - \frac{1}{g_s} \partial_z A^{(\boldsymbol{m})} \right) F^{(\boldsymbol{m})[k_{\text{max}}^{(\boldsymbol{m})}]} \left( \partial_z - \frac{1}{g_s} \partial_z A^{(\boldsymbol{n}-\boldsymbol{m})} \right) F^{(\boldsymbol{n}-\boldsymbol{m})[k_{\text{max}}^{(\boldsymbol{n}-\boldsymbol{m})}]}.
\ee
So either one of the terms, or both, cancel. Say the first one does. Then, when using the asymptotic expansion in $g_s$ for the free energy one will find that $A^{(\boldsymbol{m})}$ is constant and $\partial_z F^{(\boldsymbol{m})[k_{\text{max}}^{(\boldsymbol{m})}]}_g = 0$ for every $g$. In spite of this, and as the propagator variable $S^{zz}$ includes in itself both $z$ and $\bar{z}$, the free energy may still be $z$--dependent due to \eqref{eq:covariantderivativeonpropagator}.

So, we have seen that these non--generic situations are associated to finding constant instanton actions (see, \textit{e.g.}, \cite{ps09} for examples). Since the holomorphic anomaly equations depend on the \textit{derivative} of the instanton action this will naturally affect the structure of the solutions, both in the exponential and polynomial parts. For the sake of simplicity, consider a logarithmic--free one--parameter transseries with $b^{(n)}=+1$. In this case, integrating the equations gives the structure
\be\label{eq:Aconststructure}
F^{(n)}_g = \textrm{Pol} \left( 3 \left[\frac{g}{2} \right] -1 \right), \qquad g\geq 1,
\ee
\noindent
and $F^{(n)}_0 = \textrm{Pol} \left(0\right)$. What we find is that the free energies come in doublets with the same structure, and that the exponentials have disappeared. The degree dependence resembles that of the perturbative sector. If the transseries has more parameters but still only one instanton action is constant, say $A_1$, then we will find the same dependence \eqref{eq:Aconststructure} for the corresponding instanton sectors $(n|0|\cdots|0)$. For the other sectors we find the generic structure given by Theorem \ref{structuraltheorem}, provided we set $E_{1\beta}=1$, $\forall \beta$, in the generating function (note \eqref{eq:EtoExp}).

In conclusion, it is important to remark that in each particular example to be worked out one will first have to understand what the structure of the solutions is, taking into account any possible deviations from Theorem \ref{thm:multiparameterlogstructure} which we have just described.

\subsection{Genus Expansions within Transseries Solutions}

As described in section \ref{sec:transseriesintro}, resonance is usually associated to the fact that the (generalized) instanton actions are not independent of each other. We have already addressed one consequence of such structure, logarithms, let us now address another. Our transseries are constructed via multi--instanton sectors $F^{\boldsymbol{(n)}} (g_s)$, labeled by integers $\boldsymbol{n} = ( n_1, \ldots, n_\kappa) \in \mathbb{N}^\kappa_0$. Generically, these sectors are computed as asymptotic power series in $g_s$, but not always: for instance, the perturbative sector $\boldsymbol{n} = \boldsymbol{0}$ has a well--known genus expansion, \textit{i.e.}, it is given by an asymptotic power series in $g_s^2$. Are there other sectors with similar  genus expansions? It is at least known that in the examples studied in \cite{asv11, sv13}---dealing with two--parameters transseries---whenever the actions in the exponential term would cancel each other out, the corresponding nonperturbative sector had an expansion in $g_s^2$. Let us  discuss how this may be generalized away from explicit examples, by exploring the holomorphic anomaly equations.

Recall from, \textit{e.g.}, \eqref{eq:preNPansatz}, the general transseries form
\be
F ( \boldsymbol{\sigma}, g_s ) = \sum_{\boldsymbol{n} \in \BN^\kappa_0} \boldsymbol{\sigma}^{\boldsymbol{n}} \, \rme^{-A^{(\boldsymbol{n})}/g_s}\, F^{(\boldsymbol{n})} (g_s).
\ee
\noindent
where $\boldsymbol{\sigma}^{\boldsymbol{n}} = \prod_{\alpha = 1}^\kappa \sigma_\alpha^{n_\alpha}$ and $A^{(\boldsymbol{n})} = \sum_{\alpha=1}^\kappa n_\alpha A_\alpha$. The results that follow are valid in general, but, as usual, we shall ease the notation by focusing on one--dimensional complex moduli spaces.

Via the holomorphic anomaly equations, the free energy in a fixed multi--instanton sector satisfies \eqref{eq:preNPeq}
\be\label{eq:DFT}
\mathcal{D}^{(\boldsymbol{n})}(g_s)\, F^{(\boldsymbol{n})}(g_s) = T^{(\boldsymbol{n})}(g_s).
\ee
\noindent
Here $\mathcal{D}^{(\boldsymbol{n})}(g_s)$ is given by 
\be\label{eq:topcalD}
\mathcal{D}^{(\boldsymbol{n})}(g_s) = \partial_{S^{zz}} - \frac{1}{2} g_s^2 \left( D_z - \frac{1}{g_s}\partial_z A^{(\boldsymbol{n})}+2\, \partial_z \widetilde{F}^{(\boldsymbol{0})}(g_s) \right)\left( \partial_z - \frac{1}{g_s}\partial_z A^{(\boldsymbol{n})} \right),
\ee
\noindent
where, recall, $\widetilde{F}^{(\boldsymbol{0})}(g_s) \equiv F^{(\boldsymbol{0})}(g_s) - \frac{1}{g_s^2} F_0^{(\boldsymbol{0})}(g_s) \simeq \sum_{g=1}^{+\infty} g_s^{2g-2} F_g^{(\boldsymbol{0})}$. As to the right--hand--side of \eqref{eq:DFT}, $T^{(\boldsymbol{n})}(g_s)$ is the familiar quadratic term
\bea
T^{(\boldsymbol{n})}(g_s) &:=& \frac{1}{2}g_s^2\, \sideset{}{'} \sum_{\boldsymbol{m}=\boldsymbol{0}}^{\boldsymbol{n}} \left( \partial_z - \frac{1}{g_s} \partial_z A^{(\boldsymbol{m})} \right) F^{(\boldsymbol{m})} \left(\partial_z - \frac{1}{g_s} \partial_z A^{(\boldsymbol{n}-\boldsymbol{m})} \right) F^{(\boldsymbol{n}-\boldsymbol{m})} \equiv \nonumber \\
&\equiv& \frac{1}{2}g_s^2\, \sideset{}{'} \sum_{\boldsymbol{m}=\boldsymbol{0}}^{\boldsymbol{n}} H^{(\boldsymbol{m})}(g_s)\, H^{(\boldsymbol{n}-\boldsymbol{m})}(g_s),
\eea
\noindent
where we have defined the $H^{(\boldsymbol{m})}(g_s)$ functions for later convenience.

Our main goal in here is to show that if a sector $\boldsymbol{r}$ is such that $A^{(\boldsymbol{r})}=0$, then $F^{(\boldsymbol{r})}(-g_s) = + F^{(\boldsymbol{r})}(g_s)$. The strategy to prove this will consist of studying $F^{(\boldsymbol{n})}(-g_s)$, for any sector $\boldsymbol{n}$, and see that, under some conditions, it is equal to $F^{(\boldsymbol{n}^\star)}(+g_s)$ for some other sector $\boldsymbol{n}^\star$. It will turn out that the sectors with vanishing instanton action $A^{(\boldsymbol{r})}=0$ are those for which $\boldsymbol{r}^\star = \boldsymbol{r}$, implying that the corresponding free energies are even in $g_s$.

Let us start with the simplest example. Consider a two--parameters transseries in which the instanton actions satisfy the relation
\be
A_1 = - A_2.
\ee
\noindent
The transseries has a $\mathbb{Z}_2$--symmetry if this is the case. Then, for any instanton sector $(n_1| n_2)$,
\be\label{eq:2paramsymmetry}
A^{(n_1|n_2)}+A^{(n_2|n_1)} = \left( n_1 A_1+n_2 A_2 \right)+\left( n_2 A_1+n_1 A_2 \right) = \left( n_1+n_2 \right)\left( A_1+A_2 \right) = 0.
\ee
\noindent
Note also that 
\be\label{eq:2paramfixed}
A^{(r|r)} = r \left( A_1+A_2 \right) = 0, \qquad \forall\, r.
\ee
\noindent
It will prove useful later to rewrite this in the following way. Let $\star$ be defined by
\be
\boldsymbol{n}^\star \equiv (n_1 | n_2)^\star := (n_2 | n_1).
\ee
\noindent
Then we can write \eqref{eq:2paramsymmetry} and \eqref{eq:2paramfixed} as $A^{(\boldsymbol{n})} + A^{(\boldsymbol{n}^\star)}=0$, and $A^{(\boldsymbol{r})}=0$ if and only if $\boldsymbol{r}^\star = \boldsymbol{r}$. Further, because of the $g_s$--dependence of \eqref{eq:topcalD},
\be
\mathcal{D}^{(\boldsymbol{n})}(-g_s) = \mathcal{D}^{(\boldsymbol{n}^\star)}(g_s),
\ee
\noindent
\textit{i.e.}, the change of sign in $g_s$ is compensated by the change of sign of the instanton actions, which in turn relates to the sector $\boldsymbol{n}^\star$. Suppose for a moment that the $T^{(\boldsymbol{n})}(g_s)$ functions satisfy the following relation
\be\label{eq:Tassumption}
T^{(\boldsymbol{n})}(-g_s) = \varepsilon^{(\boldsymbol{n})}\, T^{(\boldsymbol{n}^\star)}(g_s),
\ee
\noindent
where $\varepsilon^{(\boldsymbol{n})}$ is a number which is equal to $+1$ when $\boldsymbol{n} = \boldsymbol{r} = (r | r) = \boldsymbol{r}^\star$. Then we can calculate
\bea
\mathcal{D}^{(\boldsymbol{n})}(-g_s) \left( F^{(\boldsymbol{n})}(-g_s) - \varepsilon^{(\boldsymbol{n})}\, F^{(\boldsymbol{n}^\star)}(g_s) \right) &=& \mathcal{D}^{(\boldsymbol{n})}(-g_s) F^{(\boldsymbol{n})}(-g_s) - \varepsilon^{(\boldsymbol{n})}\, \mathcal{D}^{(\boldsymbol{n}^\star)}(g_s) F^{(\boldsymbol{n}^\star)}(g_s) = \nonumber \\
&=& T^{(\boldsymbol{n})}(-g_s) - \varepsilon{(\boldsymbol{n})} \, T^{(\boldsymbol{n}^\star)}(g_s) = 0.
\label{eq:calDcombinationzero}
\eea
\noindent
This means that the combination $F^{(\boldsymbol{n})}(-g_s) - \varepsilon^{(\boldsymbol{n})}\, F^{(\boldsymbol{n}^\star)}(g_s)$ is annihilated by a linear differential operator. This operator admits the trivial solution, so we can have
\be
F^{(n_1|n_2)}(-g_s) = \varepsilon^{(n_1|n_2)}\, F^{(n_2|n_1)}(g_s).
\ee
\noindent
Finally setting $n_1=n_2=r$ we obtain the desired result
\be
F^{(r|r)}(-g_s) = + F^{(r|r)}(g_s).
\ee
\noindent
There are two things we have not shown in this example. One is the assumption \eqref{eq:Tassumption}, with which we shall deal in a more general setting to be described below. The other is explaining how the trivial solution is selected in \eqref{eq:calDcombinationzero}. We take care of this in the following lemma.

\begin{lem}\label{lem:kernelD}
For a function $y (g_s;z,S^{zz})$ such that
\be
\mathcal{D}^{(\boldsymbol{n})}(g_s)\, y (g_s;z,S^{zz}) = 0,
\ee
\noindent
we have that $y$ is the trivial solution, $y(g_s;z,S^{zz})=0$, if the ambiguity part $y(g_s;z,0)$, obtained by setting all the propagators to zero, is also zero.
\end{lem}
\begin{proof}
Expanding both $\mathcal{D}^{(\boldsymbol{n})}(g_s)$ and $y(g_s)$ in a $g_s$--series we have
\be
\sum_{h=0}^g \mathcal{D}^{(\boldsymbol{n})}_h\, y_{g-h} = 0, \qquad g=0,1,2,\ldots.
\ee
\noindent
For $g=0$,
\be
\mathcal{D}^{(\boldsymbol{n})}_0\, y_0 = 0 \qquad \Rightarrow \qquad y_0 (z,S^{zz}) = \rme^{\frac{1}{2} \left( \partial_z A^{(\boldsymbol{n})} \right)^2 \, S^{zz}}\, v_0 (z),
\ee
\noindent
where $v_0(z)$ is the corresponding holomorphic ambiguity. If $v_0 = 0$ then $y_0 = 0$. By induction in $g$, if $y_h=0$, $\forall\, h<g$, then
\be
\mathcal{D}^{(\boldsymbol{n})}_0\, y_g = 0 \qquad \Rightarrow \qquad y_g (z,S^{zz}) = \rme^{\frac{1}{2} \left( \partial_z A^{(\boldsymbol{n})} \right)^2\, S^{zz}}\, v_g (z).
\ee
\noindent
But $v_g = 0$ implies $y_g = 0$, and the proof is complete.
\end{proof}

Let us now describe the generalization to the $(\mathbb{Z}_2)^p$--symmetric case. Consider a $2p$--parameters transseries such that the instanton actions satisfy
\be
A_{p+1} = - A_1, \quad A_{p+2} = - A_2, \quad \ldots , \quad A_{2p} = - A_p.
\ee
\noindent
This means that $A^{(\boldsymbol{n})}=0$ if $n_{p+1}=n_1$, $n_{p+2}=n_2$, \ldots, $n_{2p}=n_p$. But we will also require the converse, so that we do not have situations such as $A_1+A_2+A_3=0$ (this would be a $\mathbb{Z}_3$--symmetry). We are imposing this extra condition in order to make sure that there are no other sectors with vanishing instanton action, for which our arguments would not apply. If we define
\be
\boldsymbol{n}^\star \equiv (n_1| n_2 | \cdots | n_p | n_{p+1} | n_{p+2} | \cdots | n_{2p} )^\star := (  n_{p+1} | n_{p+2} | \cdots | n_{2p} | n_1| n_2 | \cdots | n_p ),
\ee
\noindent
a $2p$--parameters transseries has a $(\mathbb{Z}_2)^p$--symmetry when
\be
A^{(\boldsymbol{n})}=0 \quad \Leftrightarrow \quad \boldsymbol{n}^\star = \boldsymbol{n}.
\ee
\noindent
The example we discussed at the beginning is the particular case $p=1$.

\begin{lem}\label{lem:topinduction}
For a $2p$--parameters transseries with $(\mathbb{Z}_2)^p$--symmetry, and for any instanton sector $\boldsymbol{n}$, one has
\be
F^{(\boldsymbol{n})}(-g_s) = \varepsilon^{(\boldsymbol{n})}\, F^{(\boldsymbol{n}^\star)}(g_s),
\ee
\noindent
provided that the associated holomorphic ambiguities satisfy the corresponding relation
\be\label{eq:ambiguitysymmetry}
f^{(\boldsymbol{n})}(-g_s) = \varepsilon^{(\boldsymbol{n})}\, f^{(\boldsymbol{n}^\star)}(g_s).
\ee
\noindent
In here the $\varepsilon$ are numbers which satisfy $\varepsilon^{(\boldsymbol{m})}\varepsilon^{(\boldsymbol{n}-\boldsymbol{m})} = \varepsilon^{(\boldsymbol{n})}$, for any $\boldsymbol{m}<\boldsymbol{n}$.
\end{lem}
\begin{proof}
We proceed by induction on the instanton sector, and leave the base case for last. Assume that the statement of the lemma holds for every $\boldsymbol{m}<\boldsymbol{n}$, then
\bea
H^{(\boldsymbol{m})}(-g_s) &=& \left( \partial_z -\frac{1}{(-g_s)} \partial_z A^{(\boldsymbol{m})} \right) F^{(\boldsymbol{m})}(-g_s) = \left( \partial_z + \frac{1}{g_s} \partial_z A^{(\boldsymbol{m})} \right) \varepsilon^{(\boldsymbol{m})}\, F^{(\boldsymbol{m}^\star)}(g_s) = \nonumber \\
&=& \varepsilon^{(\boldsymbol{m})} \left(\partial_z - \frac{1}{g_s} \partial_z A^{(\boldsymbol{m}^\star)} \right) F^{(\boldsymbol{m}^\star)}(g_s) = \varepsilon^{(\boldsymbol{m})}\, H^{(\boldsymbol{m}^\star)}(g_s).
\label{eq:Hrelation}
\eea
\noindent
Using \eqref{eq:Hrelation} above, we can calculate
\bea
T^{(\boldsymbol{n})}(-g_s) &=& \frac{1}{2} (-g_s)^2\, \sideset{}{'}\sum_{\boldsymbol{m}=\boldsymbol{0}}^{\boldsymbol{n}} H^{(\boldsymbol{m})}(-g_s)\, H^{(\boldsymbol{n}-\boldsymbol{m})}(-g_s) = \nonumber \\
&=& \frac{1}{2}g_s^2\, \sideset{}{'}\sum_{\boldsymbol{m}=\boldsymbol{0}}^{\boldsymbol{n}} \varepsilon^{(\boldsymbol{m})}\, \varepsilon^{(\boldsymbol{n}-\boldsymbol{m})}\, H^{(\boldsymbol{m}^\star)}(g_s)\, H^{((\boldsymbol{n}-\boldsymbol{m})^\star)}(g_s) = \nonumber\\
&=& \varepsilon^{(\boldsymbol{n})}\, \frac{1}{2}g_s^2\, \sideset{}{'}\sum_{\boldsymbol{m}^\star=\boldsymbol{0}}^{\boldsymbol{n}^\star} H^{(\boldsymbol{m}^\star)}(g_s)\, H^{(\boldsymbol{n}^\star-\boldsymbol{m}^\star)}(g_s) = \varepsilon^{(\boldsymbol{n})}\, T^{(\boldsymbol{n}^\star)}(g_s),
\label{eq:Trelation}
\eea
\noindent
where in the third line we have used that $\varepsilon^{(\boldsymbol{m})}\varepsilon^{(\boldsymbol{n}-\boldsymbol{m})} = \varepsilon^{(\boldsymbol{n})}$, that $\star$ is linear, and that 
\be
\sideset{}{'}\sum_{\boldsymbol{m}=\boldsymbol{0}}^{\boldsymbol{n}} = \sideset{}{'}\sum_{\boldsymbol{m}^\star=\boldsymbol{0}}^{\boldsymbol{n}^\star}.
\ee
\noindent
From the $g_s$--dependence of the differential operator \eqref{eq:topcalD} we conclude that
\be
\mathcal{D}^{(\boldsymbol{n})}(-g_s) = \mathcal{D}^{(\boldsymbol{n}^\star)}(g_s),
\ee
\noindent
which may be used, together with \eqref{eq:Trelation}, to calculate
\bea
\mathcal{D}^{(\boldsymbol{n})}(-g_s) \left( F^{(\boldsymbol{n})}(-g_s) - \varepsilon^{(\boldsymbol{n})}\, F^{(\boldsymbol{n}^\star)}(g_s) \right) &=& \mathcal{D}^{(\boldsymbol{n})}(-g_s) F^{(\boldsymbol{n})}(-g_s) - \varepsilon^{(\boldsymbol{n})}\, \mathcal{D}^{(\boldsymbol{n}^\star)}(g_s) F^{(\boldsymbol{n}^\star)}(g_s) = \nonumber \\
&=& T^{(\boldsymbol{n})}(-g_s) - \varepsilon{(\boldsymbol{n})} \, T^{(\boldsymbol{n}^\star)}(g_s) = 0.
\eea
\noindent
Of course this is the exact same calculation as in \eqref{eq:calDcombinationzero}. Lemma \ref{lem:kernelD} now ensures that if $f^{(\boldsymbol{n})}(-g_s) - \varepsilon^{(\boldsymbol{n})}\, f^{(\boldsymbol{n}^\star)}(g_s)=0$ then
\be
F^{(\boldsymbol{n})}(-g_s) = \varepsilon^{(\boldsymbol{n})}\, F^{(\boldsymbol{n}^\star)}(g_s).
\ee
\noindent
To finish this proof by induction we have to check the base cases, which are the one--instanton sectors (those with $\|\boldsymbol{n} \| = 1$). But for these the $T$ function is zero, there is no obstruction, and the proof carries through as in the general case, making use again of Lemma \ref{lem:kernelD} for $\|\boldsymbol{n} \| = 1$.
\end{proof}

Now we can state the main result of this subsection.

\begin{thm}
Consider a $2p$--parameters transseries with $(\mathbb{Z}_2)^p$--symmetry. If some sector $\boldsymbol{r}$ has a vanishing total instanton action, $A^{(\boldsymbol{r})}=0$, then $F^{(\boldsymbol{r})}(g_s)$ has a topological genus expansion in the string coupling $g_s$ (an asymptotic expansion in powers of $g_s^2$), provided the holomorphic ambiguities respect the symmetry (as in \eqref{eq:ambiguitysymmetry} with $\varepsilon^{(\boldsymbol{r})}=+1$).
\end{thm}
\begin{proof}
A sector $\boldsymbol{r}$ with $A^{(\boldsymbol{r})}=0$ satisfies $\boldsymbol{r}^\star = \boldsymbol{r}$. Then Lemma \ref{lem:topinduction} says
\be
F^{(\boldsymbol{r})}(-g_s) = + F^{(\boldsymbol{r})}(g_s),
\ee
\noindent
where we have used $\varepsilon^{(\boldsymbol{r})}=+1$.
\end{proof}

It is important to remark that the whole argumentation of this section and the result above is indifferent to the number of propagators. Sub--indices $ij$ could have been used where appropriate and everything would have worked in the same way. Also, note that the validity of \eqref{eq:ambiguitysymmetry} should be confirmed from understanding the fixing of the holomorphic ambiguity for the corresponding sectors. Nonetheless, it is quite natural to expect that the $\BZ_2$--symmetry of the model is shared by the ambiguities.

One interesting special case of a $(\mathbb{Z}_2)^p$--symmetric transseries is when the instanton actions are given by
\be
A_\alpha = \omega^{\alpha -1} A, \qquad \alpha = 1,2, \ldots, 2^s,
\ee
\noindent
where $s$ is an integer and $p=2^{s-1}$, and $\omega$ is the root of unity $\rme^{2\pi\rmi/2^s}$. Note that $A^{(\boldsymbol{n})}=0$ if and only if $n_1 = n_{2^{s-1}+1}$, $n_2 = n_{2^{s-1}+2}$, \ldots, $n_{2^{s-1}}=n_{2^s}$, and so this is indeed a $(\mathbb{Z}_2)^p$--symmetry with $p=2^{s-1}$. Unfortunately, it is not straightforward to go any further than this: the reason we have restricted ourselves to $\mathbb{Z}_2$ symmetries, and not something larger such as $\mathbb{Z}_3$ or more, is that for the latter case the type of argument presented in Lemma \ref{lem:topinduction} does not work. Using the change $g_s \to -g_s$ precisely worked because there was a corresponding $\mathbb{Z}_2$ symmetry at the level of the instanton actions that counterbalanced the minus sign.

Let us also give a brief description of the anti--holomorphic structure of the higher--instanton free energies, in the present case with $\mathbb{Z}_2$--symmetry. If $n_1 \neq n_2$, $F^{(n_1|n_2)}_g$ is described by the generating function of Theorem \ref{structuraltheorem}, provided we take into account that $E_{12} = E_{11}^{-1}$ and $E_{22}= E_{11}$ because $A_1 = -A_2$ (recall \eqref{eq:EtoExp}). The sectors with vanishing total instanton action, $(r|r)$, have a more complicated structure. Assume we are in the case $b^{(\boldsymbol{n})}=+1$, so $F^{(r|r)}_{2h} = 0$. We then find that $F^{(r|r)}_{2h+1}$ is formed by a polynomial of degree $3h-1$ plus a combination of exponentials times polynomials, as in the generic case. The coefficients $a$ in the exponentials are still given by the adapted generating function, but the degrees of the corresponding polynomials now become more complicated than usual. The generating function helps us with the structure, though. It turns out that if $\lambda \left( (r|r);\gamma_{(r|r)} \right) = 3,4,5,\ldots$, then the degree of the polynomial follows the usual rule: it is $3(g+1-\lambda)$, where here $g = 2h+1$. An exception appears if $\lambda = 4$ and the corresponding value of $a/2$ is a sum of two squares\footnote{Moreover, the value of $a/2$ can be neither in the list $\{m_1^2-m_1 m_2+m_2^2\}_{1\leq m_1<m_2} = \{3,7,12,13,19,21,27,\ldots\}$, which corresponds to terms with $\lambda=3$, nor can it be a perfect square, which corresponds to $\lambda=2$.}. In this case the degree becomes $5(h-1)$. If $\lambda = 2$ then degree of the polynomial is $5h$. The value of $\lambda$ is equal to $1$ only for $a=0$, which corresponds to the bare polynomial in $F^{(r|r)}_{2h+1}$. This description follows from the study of solutions to high genus and instanton number, but we provide no proof. Modifications therefore may occur.

Let us make a final comment. The argument for the appearance of genus expansions can be extended to logarithmic transseries provided $k_{\text{max}}^{(\boldsymbol{n}^\star)} = k_{\text{max}}^{(\boldsymbol{n})}$, so that one can carry through the same kind of proof. On the other hand, one expects (from examples such as Painlev\'e I, II) that the transseries is logarithmic--free along $(r|r)$ sectors, which is certainly compatible with the symmetry conditions on $k_{\text{max}}$. However, the equations are not constraining enough to show this.

\section{A Glimpse of Large--Order Analysis}\label{sec:largeorder}

We have seen at the end of section \ref{sec:transseriesintro} that the transseries we study in this paper are expected to be {\em resurgent}, meaning that they satisfy bridge equations such as (\ref{eq:bridgeeq}). What this means in practice is that we expect to find relations linking the \textit{perturbative} coefficients in a transseries to its \textit{nonperturbative} content. These relations can be tested and therefore lead to useful checks and insight on our general formalism. Of course, such tests are most interesting when studying particular examples, which will be the purpose of the next papers in this series \cite{cesv14, cesv15}. Nevertheless, as we shall see, it is still the case that the general theory includes model--independent results which, in particular, provide strong resurgent tests on our present overall approach.

Generically, the bridge equations \eqref{eq:bridgeeq} lead to so--called \textit{large--order relations} between the transseries coefficients. Let us begin with the simple example of an one--parameter log--free resurgent transseries,
\be
F(x,\sigma) = \sum_{n=0}^{+\infty} \sum_{g=0}^{+\infty} \sigma^n\, \rme^{-nA/x}\, F^{(n)}_g\, x^g,
\ee
\noindent
where we also took all the starting orders $b^{(n)}$ of the $g$--series to vanish. One can derive from the bridge equation for such a transseries (see, \textit{e.g.}, \cite{asv11}) that its perturbative coefficients satisfy the following large--order relation
\be
F^{(0)}_g \simeq \frac{1}{2\pi\rmi} \sum_{k=1}^{+\infty} S_1^k\, \sum_{h=0}^{+\infty} \frac{\Gamma \left(g-h\right)}{\left(kA\right)^{g-h}}\, F^{(k)}_h.
\label{eq:largeorder}
\ee
\noindent
Here, the \textit{leading Stokes coefficient} $S_1$ is the constant term in the $\sigma$--dependent prefactor $S_\omega(\sigma)$ that appears in the bridge equations \eqref{eq:bridgeeq}. Its exact value differs from problem to problem, and therefore has to be calculated for each problem separately.

The meaning of the above large--order relation is most easily seen by expanding the $h$ and $k$--sums. For example, the leading term ($k=1$, $h=0$) tells us that, at large $g$, the $F^{(0)}_g$ behave as
\be
F^{(0)}_g \simeq \frac{1}{2\pi\rmi}\, S_1\, \frac{\Gamma \left(g\right)}{A^{g}},
\label{eq:largeorder0}
\ee
\noindent
where we used that we set $F^{(1)}_0=1$. First of all, we see from this expression that the perturbative coefficients $F^{(0)}_g$ grow like $g!$. This means that they define an \textit{asymptotic} series---something which is generically true for resurgent transseries of this form \cite{asv11}. More importantly, we see that from the precise large--order behavior of the $F^{(0)}_g$ we can read off the instanton action $A$ and the Stokes coefficient $S_1$. Continuing this process by including further terms in the $h$--sum from \eqref{eq:largeorder}, we see that from $1/g$ corrections to \eqref{eq:largeorder0} we can read off the coefficient $F^{(1)}_1$, from $1/g^2$ corrections we can read off $F^{(1)}_2$, and so on. When all of these corrections can be summed (or, in practice, Borel resummed), one can continue including further terms from the $k$--sum and, for example, read off $F^{(2)}_0$ from $2^{-g}$ corrections and so on. In other words, all of the nonperturbative transseries coefficients ``resurge'' in the large--order behavior of the perturbative coefficients $F^{(0)}_g$. Here, we will only use large--order relations in their simplest form \eqref{eq:largeorder0}.

In topological string theories, there are a few subtleties that we should keep in mind when writing down the analogue of the above large--order formulae. First of all, it is customary to write the perturbative free energy of a topological string theory as
\be
F(g_s) \simeq \sum_{g=0}^{+\infty} g_s^{2g-2} F_g.
\label{eq:topstrclass}
\ee
\noindent
Here, the expansion parameter is the closed string coupling constant, $g_s^2$, but in the nonperturbative transseries, where D--brane effects and therefore open strings are involved, the appropriate expansion parameter is the open string coupling constant, $g_s$. However, we still want to use the traditional form \eqref{eq:topstrclass} for the perturbative sub--series, which leads to the appearance of some factors of $2$ in our final formulae. Another subtlety that we should keep in mind when comparing to the simple example above, is that general topological string theories have multi--parameter transseries solutions with logarithm--sectors and nonzero starting orders $g$. In what follows, we will not need to deal with the complications that arise due to log--sectors (but see \cite{asv11, sv13} for some examples of large--order formulae when such sectors are present), but we will need to consider multi--parameter transseries with arbitrary starting orders. Taking all of these subtleties into account, and assuming a bridge equation of the form \eqref{eq:bridgeeq}, one finds that the leading large--order formula we need in the topological string case takes the following form
\be\label{eq:fullresurgenceformula}
F^{({\boldsymbol{0}})}_g \simeq \frac{1}{2\pi\rmi}\, \sum_{\alpha=1}^{\kappa} \sum_{n=1}^{+\infty} S_{1,\alpha}^n\, \sum_{h=0}^{+\infty} \frac{\Gamma \left( 2g - 2 - b^{({\boldsymbol{n}_\alpha})} - h \right)}{\left( n A_\alpha \right)^{2g - 2 - b^{({\boldsymbol{n}_\alpha})} - h}}\, F^{({\boldsymbol{n}_\alpha})}_h.
\ee
\noindent
Here, $1 \leq \alpha \leq \kappa$ labels the different instanton actions $A_\alpha$, meaning that we have a $\kappa$--parameters transseries. We have also used a shorthand where $({\boldsymbol{n}_\alpha}) := (0| \cdots |0| n |0| \cdots |0)$, with an $n$ in the $\alpha$th place. The $S_{1,\alpha}$ are the leading Stokes constants associated to the different instanton sectors; they need to be determined for each problem separately. If however the bridge equation is not of the form \eqref{eq:bridgeeq}, then  \eqref{eq:fullresurgenceformula} may naturally generalize---although large deviations are not expected.

Let us now see how the above resurgent formulae provide checks, within the perturbative sector, of our nonperturbative results in the previous sections. Either large--order or resurgent analysis (see, \textit{e.g.}, \cite{z81, asv11}) yields a relation determining the instanton action, obtained by an appropriate ratio of \eqref{eq:largeorder0},
\begin{equation}\label{eq:A2resurgence}
A^2 \left(z^i,S^{ij}\right) = \lim_{g\to+\infty} \frac{\Gamma \left(2g-b+2\right)}{\Gamma \left(2g-b\right)}\, \frac{F^{(0)}_g \left(z^i,S^{ij} \right)}{F^{(0)}_{g+1} \left(z^i,S^{ij} \right)},
\end{equation}
\noindent
where $A$ is the dominant instanton action at the point in moduli space where the right--hand--side is being evaluated. We shall now explicitly prove that this immediately implies holomorphicity of the instanton action; $\partial_{S^{ij}} A^2 = 0$.

In order to prove this statement, take a propagator derivative on both sides of \eqref{eq:A2resurgence}
\bea
\partial_{S^{ij}} A^2 &=& \lim_{g\to+\infty} \frac{\Gamma \left(2g-b+2\right)}{\Gamma \left(2g-b\right)} \left( \frac{\partial_{S^{ij}}F^{(0)}_g}{F^{(0)}_{g+1}} - \frac{\partial_{S^{ij}}F^{(0)}_{g+1}\, F^{(0)}_g}{(F^{(0)}_{g+1})^2} \right) = \nonumber \\
&=& \lim_{g\to+\infty} \partial_{S^{ij}}F^{(0)}_g \left( \frac{\Gamma \left(2g-b+2\right)}{\Gamma \left(2g-b\right)}\, \frac{1}{F^{(0)}_{g+1}}- \frac{\Gamma \left(2g-b\right)}{\Gamma \left(2g-b-2\right)}\, \frac{F^{(0)}_{g-1}}{(F^{(0)}_g)^2} \right),
\label{eq:propagatorderivativeA}
\eea
\noindent
where we have relabelled $g\to g-1$ in the second term. Now, for large enough $g$, the perturbative free energies grow factorially and controlled by the dominant instanton action, as in
\be\label{eq:largeorderperturbative}
F^{(0)}_g \simeq c(z_i,S^{ij})\, A(z^i,S^{ij})^{-(2g-b)}\, \Gamma \left(2g-b\right) \left( 1+\mathcal{O}(g^{-1}) \right),
\ee
\noindent
where $c(z_i,S^{ij})$ is a function independent of $g$. Let us first focus on the term in brackets in \eqref{eq:propagatorderivativeA}. For large enough $g$ we can calculate it to give
\be
\frac{A^{2g-b}}{\Gamma \left(2g-b\right) c} \left( 1+\mathcal{O}(g^{-1}) \right) - \frac{A^{2g-b}}{\Gamma \left(2g-b\right) c} \left( 1+\mathcal{O}(g^{-1}) \right) =  \frac{A^{2g-b}}{\Gamma \left(2g-b\right)}\, \mathcal{O}(g^{-1}),
\ee
\noindent
as the leading terms cancel each other out. Let us next turn to the factor of $\partial_{S^{ij}} F^{(0)}_g$ in \eqref{eq:propagatorderivativeA}. The usual holomorphic anomaly equations allow us to express this in terms of $z$--derivatives,
\be
\partial_{S^{ij}}F^{(0)}_g = \frac{1}{2} \left( \partial_i \partial_j F^{(0)}_{g-1} - \Gamma^{k}_{ij}\, \partial_k F^{(0)}_{g-1} + \sum_{h=1}^{g-1} \partial_i F^{(0)}_{g-h}\,\partial_j F^{(0)}_h \right).
\ee
\noindent
In particular, the large--order behavior of the derivatives of the free energies are obtained from \eqref{eq:largeorderperturbative},
\bea
\partial_i F^{(0)}_h &\sim& \partial_i c\, \frac{\Gamma \left(2h-b\right)}{A^{2h-b}} - c\, \partial_i A\, \frac{\Gamma \left(2h-b+1\right)}{A^{2h-b+1}}, \\
\partial_i \partial_j F^{(0)}_h &\sim& \partial_i \partial_j c\, \frac{\Gamma \left(2h-b\right)}{A^{2h-b}} - \left( \partial_i c\, \partial_j A + \partial_j c\, \partial_i A + c\, \partial_i \partial_j A \right) \frac{\Gamma \left(2h-b+1\right)}{A^{2h-b+1}} +  \nonumber \\
&+&
c\, \partial_i A\, \partial_j A\, \frac{\Gamma \left(2h-b+2\right)}{A^{2h-b+2}}.
\eea
\noindent
It follows that, at large $g$,
\bea
\partial_i \partial_j F^{(0)}_{g-1} &\sim& \frac{\Gamma \left(2g-b\right)}{A^{2g-b}} \left\{ c\, \partial_i A\, \partial_j A - \left( \partial_i c\, \partial_j A + \partial_j c\, \partial_i A + c\, \partial_i \partial_j A \right) A\, \frac{\Gamma \left(2g-b-1\right)}{\Gamma \left(2g-b\right)} + \right. \nonumber \\
&&
\hspace{60pt}
\left. + \partial_i \partial_j c\, A^2\, \frac{\Gamma \left(2g-b-2\right)}{\Gamma \left(2g-b\right)} \right\} = \nonumber \\
&=& \frac{\Gamma \left(2g-b\right)}{A^{2g-b}}\, \left\{ \mathcal{O}(g^0) + \mathcal{O}(g^{-1}) + \mathcal{O}(g^{-1})Ê\right\},
\label{eq:largeordersecondderivative}
\eea
\noindent
and
\bea
\Gamma^{k}_{ij}\, \partial_k F^{(0)}_{g-1} &\sim& \frac{\Gamma \left(2g-b\right)}{A^{2g-b}}\, \Gamma^{k}_{ij} \left\{ - c\, A\, \partial_k A\, \frac{\Gamma \left(2g-b-1\right)}{\Gamma \left(2g-b\right)} + \partial_k c\, A^2\, \frac{\Gamma \left(2g-b-2\right)}{\Gamma\left(2g-b\right)} \right\} = \nonumber \\
&=& \frac{\Gamma \left(2g-b\right)}{A^{2g-b}} \left\{ \mathcal{O}(g^{-1}) + \mathcal{O}(g^{-2}) \right\}.  
\label{eq:largeorderchristoffelterm}
\eea
\noindent
The quadratic term of the holomorphic anomaly equations is trickier because some factors are not really large (for either small or large $h$). However, we can use an inequality that the free energies satisfy since they are asymptotic coefficients. In fact, in this case, there must be a function $\widetilde{c}(z_i, S^{ij})$ such that, for every $g$,
\be
\left| F^{(0)}_g \right| \leq \widetilde{c} \left|A\right|^{-(2g-b)}\, \Gamma \left(2g-b\right).
\ee
\noindent
Therefore, one can evaluate the quadratic term as
\bea
\left| \sum_{h=1}^{g-1} \partial_i F^{(0)}_{g-h}\, \partial_j F^{(0)}_h \right| &\leq& \sum_{h=1}^{g-1} \left|\partial_i F^{(0)}_{g-h}\right| \left|\partial_j F^{(0)}_h\right| \leq \nonumber \\
&\leq& \sum_{h=1}^{g-1} \left( \partial_i \widetilde{c}\, \frac{\Gamma \left(2g-2h-b\right)}{\left|A\right|^{2g-2h-b}} + \widetilde{c}\, \partial_i |A|\, \frac{\Gamma \left(2g-2h-b+1\right)}{\left|A\right|^{2g-2h-b+1}} \right) \times \nonumber \\
&&
\hspace{50pt}
\times \left( \partial_j \widetilde{c}\, \frac{\Gamma \left(2h-b\right)}{\left|A\right|^{2h-b}} + \widetilde{c}\, \partial_j |A|\, \frac{\Gamma \left(2h-b+1\right)}{\left|A\right|^{2h-b+1}} \right).
\eea
\noindent
Next, we make use of the inequalities
\bea
\Gamma \left(2g-2h+p\right) \Gamma \left(2h+p\right) \leq \Gamma \left(2g+p-2\right) \Gamma \left(p+2\right), \\
\Gamma \left(2g-2h+p\right) \Gamma \left(2h+p+1\right) \leq \Gamma \left(2h+p-1\right) \Gamma \left(p+2\right),
\eea
\noindent
in order to find
\be
\left| \sum_{h=1}^{g-1} \partial_i F^{(0)}_{g-h}\, \partial_j F^{(0)}_h \right| \leq \frac{\Gamma \left(2g-b\right)}{\left|A\right|^{2g-b}} \left\{ \mathcal{O}(g^{-1}) + \mathcal{O}(g^0) + \mathcal{O}(g^0) \right\},
\label{eq:largeorderquadraticterm}
\ee
\noindent
where we have essentially done the same type of manipulations as before. Note that the sum in $h$ has $\sim g$ terms which must also be taken into account. Now plugging \eqref{eq:largeordersecondderivative}, \eqref{eq:largeorderchristoffelterm} and \eqref{eq:largeorderquadraticterm} into \eqref{eq:propagatorderivativeA} we finally obtain
\be
\left| \partial_{S^{ij}}A^2 \right| \leq \mathcal{O} \left(g^{-1}\right),
\ee
\noindent
and as $g$ goes to infinity we prove the holomorphicity of $A$.

It is important to notice that in the above reasoning we have only used the \textit{perturbative} formulation of the holomorphic anomaly equations, \eqref{eq:propagatorsholanomeqs}, and nothing else besides completely generic large--order arguments. Yet, we have obtained a general result which precisely reproduces what we had computed earlier via the nonperturbative formulation of the holomorphic anomaly equations we discussed in section \ref{sec:holoeqs}. In particular, this large--order test independently validates our approach. As we mentioned earlier, this holomorphicity tells us that, in cases where there is a known matrix model holographic dual, the instanton action keeps its usual geometric interpretation and may be directly computed as an adequate period of the corresponding spectral curve \cite{msw07, ps09, dmp11}. Otherwise, strictly within the closed string theoretic approach, this has the natural generalization of giving the instanton action a threefold geometric interpretation as a combination of holomorphic cycles in the  Calabi--Yau geometry.

Having checked our nonperturbative result for the instanton action using a perturbative large--order analysis, let us move to the next--to--leading order and address the one--loop one--instanton sector. As we have shown in the previous sections, the anti--holomorphic dependence of the higher--instanton free energies is essentially independent of the particular model (geometry) we are dealing with. For the perturbative sector that structure was a polynomial of a certain degree in the propagators (the coefficients of those polynomials are holomorphic functions that \textit{do} depend on the model). A similar situation happens for the nonperturbative sectors, where we now encounter not only polynomials but also exponentials in the propagators. The degrees and exponents are almost independent of the model (a mild dependence on the starting genus, recall \eqref{eq:lambdainitialdata}), but, again, the coefficients will be holomorphic and dependent on the model at hand.

This robust anti--holomorphic structure must be compatible with the large--order behavior of the \textit{perturbative} free energies \eqref{eq:fullresurgenceformula} and, in general, with the resurgence relations briefly mentioned at the beginning of this section. A detailed check of this has to be done for particular models. Here, as an example, consider the leading term in \eqref{eq:fullresurgenceformula}
\be\label{eq:largeorderleading}
F^{(0)}_g \simeq \frac{S_1}{2\pi\rmi}\, \frac{\Gamma \left(2g-2 - b^{(1)}\right)}{A^{2g-2-b^{(1)}}}\, F^{(1)}_0 + \cdots,
\ee
\noindent
where $A$ is now the instanton action with the smallest absolute value at the point in moduli space we are evaluating our model, \textit{i.e.}, the dominant instanton action. By $(1)$ we denote the associated one--instanton sector $(0| \cdots |1| \cdots |0)$. From \eqref{eq:largeorderleading} we can easily obtain
\be
\frac{S_1}{2\pi\rmi}\, F^{(1)}_0 = \lim_{g\to+\infty} \frac{A^{2g-2-b^{(1)}}}{\Gamma \left(2g-2-b^{(1)}\right)}\, F^{(0)}_g.
\ee
\noindent
Now comes the key point. On the left--hand--side we have an \textit{exponential} in the propagator $S^{zz}$, since according to our earlier results $F^{(1)}_0$ is of the form
\be
F^{(1)}_0 = \rme^{\frac{1}{2} \left(\partial_z A\right)^2 S^{zz}}\, f^{(1)}_0 (z),
\ee
\noindent
where $f^{(1)}_0$ is the corresponding holomorphic ambiguity. On the other hand, on the right--hand--side we have a \textit{polynomial}, albeit of increasing degree $3g-3$ as $g \to +\infty$. Of course that if our nonperturbative results are to be validated by resurgence, and in particular by the resurgence of the perturbative sector, the propagator structure of both sides has to match. If we denote
\be
F^{(0)}_g = \sum_{k=0}^{3g-3} p_{g,k}(z)\, (S^{zz})^k,
\ee
\noindent
we must have
\be\label{eq:largeorder1inst}
\frac{S_1}{2\pi\rmi}\, f^{(1)}_0(z)\, \frac{1}{k!} \left( \frac{1}{2} (\partial_z A)^2 \right)^k = \lim_{g\to+\infty}\frac{A^{2g-2-b^{(1)}}}{\Gamma \left(2g-2-b^{(1)}\right)}\,  p_{g,k} (z).
\ee
\noindent
The left--hand--side of \eqref{eq:largeorder1inst} is essentially the same for every $k$, since it arises from the exponential dependence in $S^{zz}$ of $F^{(1)}_0$, while on the right--hand--side one has the large--order values of every coefficient in $F^{(0)}_g$, $g \gg 1$. One could go on and include more terms in \eqref{eq:largeorderleading}, in order to see what the anti--holomorphic structure of the solutions imposes. More precisely, the left--hand--side of \eqref{eq:fullresurgenceformula} is a polynomial, it has finite degree in the propagators, while the right--hand--side is an infinite sum involving every possible power of the propagators. This means, that there must be a cancellation of powers higher than $3g-3$. Checking this structure in its full generality is beyond the scope of this paper. Nonetheless, an equivalent approach is to confirm it within particular examples and for specified numerical precision. We have done this in the context of local $\BP^2$ around the conifold point and obtained a very precise numerical check of \eqref{eq:largeorder1inst} \cite{cesv14}, again validating our nonperturbative approach with purely perturbative data.

Finally, notice that since $p_{g,0} = f^{(0)}_g$, the perturbative holomorphic ambiguity, the $k=0$ equation in \eqref{eq:largeorder1inst} is
\be
\frac{S_1}{2\pi\rmi}\, f^{(1)}_0(z) = \lim_{g\to+\infty} \frac{A^{2g-2-b^{(1)}}}{\Gamma \left(2g-2-b^{(1)}\right)}\, f^{(0)}_g.
\ee
\noindent
This implies that $f^{(1)}_0$, the first non--perturbative holomorphic ambiguity, can be fixed, up to the Stokes constant $S_1$, by the large--order relation above. This is just an example of what one obtains from resurgent relations such as \eqref{eq:fullresurgenceformula} when the propagators are set to zero. In fact, when $S^{ij}=0$, $F^{(\boldsymbol{n})}_g$ reduces to a combination of $f^{(\boldsymbol{n})}_g$ and all the previous ambiguities. In particular, this means that the resurgence of the free energies, \eqref{eq:fullresurgenceformula}, induces a resurgence of the holomorphic ambiguities, immediately implying that the fixing of the perturbative ambiguity is equivalent, via resurgence, to the fixing of all nonperturbative ambiguities. Note, however, that while this is in principle a correct strategy to fix all holomorphic ambiguities, in practice it requires a lot of computation and becomes a bit cumbersome. In the next section we shall discuss another, more practical, approach to fixing  perturbative and nonperturbative holomorphic ambiguities.

\section{Fixing the Holomorphic Ambiguities}\label{sec:fixingambiguities}

So far we have understood in detail how the holomorphic anomaly equations fix the anti--holomorphic dependence of all (generalized) multi--instanton sectors and, thus, the full $S^{ij}$ dependence of its transseries solution. Akin to the perturbative case, one question remains: is it possible, in each example, to further fix the holomorphic dependence and thus explicitly solve the holomorphic anomaly equations? We shall now discuss, generically, how to fix the remaining holomorphic ambiguity---but one should further point out that only in concrete examples will this fixing become fully explicit \cite{cesv14, cesv15}.

\subsection{Fixing the Perturbative Ambiguity}

To first set the stage let us review, following \cite{hkr08}, how the perturbative holomorphic ambiguity is fixed for each genus, in the case of local Calabi--Yau geometries. For the sake of simplicity we will assume that the complex moduli space is one--dimensional, but the results are general.

As discussed earlier, the perturbative free energies, $F^{(0)}_g$ for $g\geq 2$, are polynomials of degree $3g-3$ in the propagator. The undetermined term in this polynomial is the holomorphic ambiguity that appears after integration of the holomorphic anomaly equations, and this is what we need to fix. Taking the holomorphic limit of the free energies, which we shall denote by $\mathcal{F}^{(0)}_g$, one finds regular functions on the complex moduli space except at a particular set of points at its boundary. At those points the Calabi--Yau geometry acquires a node. For example, in the case of geometries which are fibrations of an elliptic curve, the so called local curves, the singular locus is given by the vanishing of the discriminant of the curve, $\Delta$. Proceeding with this example, it is known that the singular locus, called the conifold divisor, corresponds in the A--model to a phase transition point in which the free energies can undergo a transition whose corresponding double--scaling limit is associated to some universality class. If $t_{\text{c}}$ is the flat coordinate around the conifold point, then one has the following behavior for the holomorphic limit of the free energies near $t_{\text{c}} = 0$,
\be\label{eq:gapcondition}
\left. \mathcal{F}^{(0)}_g \right|_{\text{c}} = \frac{{\mathfrak{c}}^{g-1}\, B_{2g}}{2g \left(2g-2\right) t_{\text{c}}^{2g-2}} +\mathcal{O }(t_{\text{c}}^0), \qquad g \geq 2,
\ee
\noindent
where $\mathfrak{c}$ is a constant depending on the model and $B_{2g}$ are the Bernoulli numbers. This behavior is called the \textit{gap condition} because there is only one singular term with a pole or order $2g-2$ at the conifold point, and then a regular tail \cite{bcov93, gv95, agnt95, gv98a, gv98c, hk06, hkq06}. Having this in mind, consider the expression for $F^{(0)}_g$ integrated from the holomorphic anomaly equations,
\be
F^{(0)}_g = C^{2g-2} \sum_{k=1}^{3g-3} p_{g,k}(z)\, (S^{zz})^k + f^{(0)}_g (z),
\ee
\noindent
where $C= C_{zzz}$ is the Yukawa coupling, $p_{g,k}(z)$ are (rational) functions of the complex modulus $z$ which depend on the particular model, and $f^{(0)}_g(z)$ is the perturbative holomorphic ambiguity. It is a general property that $C \sim \Delta^{-1}$, and $\DeltaÊ\sim t_{\text{c}}$. Note that near the conifold point the Yukawa coupling reproduces the singular behavior of the gap condition, \eqref{eq:gapcondition}, but not the correct coefficient. To make things right one has to consider a holomorphic ambiguity of the form
\be\label{eq:ambiguityansatz}
f^{(0)}_g(z) = \frac{p_{g,0} (z)}{\Delta^{2g-2}}.
\ee
\noindent
The function $p_{g,0}(z)$ is then determined by imposing that $\mathcal{F}^{(0)}_g$ is regular at every point in moduli space---besides at the conifold point---in particular when $z\to 0$ and $z\to \infty$. Assuming for simplicity that $\Delta$ is linear in $z$, we find that $p_{g,0}(z)$ cannot have poles (this is regularity at $z=0$, the large--radius point), and it can only grow with $z$ as fast as $\Delta^{2g-2}$, in which case $p_{g,0}(z)$ is a polynomial in $z$. A counting argument then shows that all the coefficients of $p_{g,0}(z)$, except for one, can be fixed by the gap condition \eqref{eq:gapcondition}. The remaining coefficient is obtained by imposing the universal behavior of the free energies at the large--radius point, $z=0$, which in the A--model corresponds to the large volume limit. This is called the constant map contribution \cite{bcov93, mm98, gv98a, fp98, gv98c}, 
\be
\left. \mathcal{F}^{(0)}_g \right|_{\text{cm}} = \frac{(-1)^g\, \chi\, B_{2g-2}B_{2g}}{4g \left(2g-2\right) \left(2g-2\right)!} + \mathcal{O}(z), \qquad g\geq 2,
\ee
\noindent
where $\chi$ is the Euler characteristic of the Calabi--Yau manifold. This algorithm for fixing the ambiguities, with its generalization to multi--dimensional moduli spaces, works for every genus and has been used to solve several topological string theories up to high genus in many examples.

The gap condition represents the phase transition singularity that the free energies undergo when the model at hand is in the universality class of the $c=1$ string at self--dual radius \cite{gv95}. Models in different universality classes will have phase transitions with a different singular behavior, controlled by a characteristic exponent $\gamma$, in the form
\be\label{eq:freeenergyHgamma}
\mathcal{F}^{(0)}_g \sim H_g\, \tau^{\left(1-g\right)\left(2-\gamma\right)}, \qquad g\geq 2,
\ee
\noindent
where $\tau$ is a proper coordinate around the phase transition point, and $H_g$ are some particular (fixed) coefficients which play the counterpart of the Bernoulli numbers in \eqref{eq:gapcondition}.  Usually, in ``conventional'' topological string theory, we have $\gamma=0$ recovering \eqref{eq:gapcondition}. But for other classes $\gamma$ can be different; \textit{e.g.}, for the class of $c=0$ string theory, or pure 2d gravity, one has $\gamma=-\frac{1}{2}$ \cite{cgmps06}. In general these classes are in correspondence to minimal models which have central charge $c_{p,q}=1-6 \frac{(p-q)^2}{pq}$ and
\be
\gamma_{p,q} = -\frac{2}{p+q-1},
\ee
\noindent
see \cite{fgz93} for a calculation in terms of multi--matrix models. Topological string theories in these more general classes have, however, been scarcely studied. One exception is \cite{cgmps06} which address (equivariant) topological strings on the local curve $X_p = \mathcal{O}(p-2)\oplus \mathcal{O}(-p) \rightarrow \mathbb{P}^1$, $p\geq 2$, which are in fact in the $c=0$ class. In this case it was found that the free energies have a structure somewhat similar to the one described above, in particular resembling the fixing of the holomorphic ambiguity. However, without ``nice'' coordinates around the phase transition point, we cannot be sure a gap condition exists. It would be extremely interesting to study topological string theory in more general universality classes and find out how more general boundary conditions fix the holomorphic ambiguities.

\subsection{Fixing Holomorphic Ambiguities at the Conifold}\label{subsec:fixingambiguityconifold}

The phase transition behavior is a universal property of the given topological string theory class. In this sense, making use of the resurgent equations that relate the large--order behavior of the perturbative free energies to the higher instanton sectors \cite{asv11}, one should be able to conclude something about those higher instanton free energies near the phase transition point. It turns out that this is the case under some circumstances and for certain instanton sectors.

Let us be more concrete by considering a model that lies in the universality class of $c=1$ string theory, and assume that there is only one complex modulus. Near the conifold point, $t_{\text{c}}=0$, the holomorphic limit of the perturbative free energies satisfies the gap condition \eqref{eq:gapcondition}, which we now write in the form
\be
\left. \mathcal{F}^{(0)}_g \right|_{\text{c}} = \frac{{\mathfrak{c}}^{g-1}\, B_{2g}}{2g \left(2g-2\right) t_{\text{c}}^{2g-2}} + {\mathfrak{a}}_g + {\mathfrak{b}}_g\, t_{\text{c}} +\cdots,
\ee
\noindent
where $\{ {\mathfrak{a}}_g, {\mathfrak{b}}_g, \ldots\}$ are numbers that depend on the particular model we are working with. Near this phase transition point we expect the large--order behavior of the $\mathcal{F}^{(0)}_g$ to be controlled by an instanton action proportional to $t_{\text{c}}$ \cite{ps09}, which is actually a period and vanishes at the conifold point. Indeed, as we have seen before from large--order behavior in section \ref{sec:largeorder}, one can extract the dominant instanton action via
\bea
A^2 &=& \lim_{g\to \infty} 4 g^2\, \frac{\mathcal{F}^{(0)}_g}{\mathcal{F}^{(0)}_{g+1}} = \lim_{g\to+\infty} 4g^2\, \frac{G_g\, t_{\text{c}}^{2-2g} + {\mathfrak{a}}_g + {\mathfrak{b}}_g\, t_{\text{c}} + \cdots}{G_{g+1}\, t_{\text{c}}^{2-2(g+1)} + {\mathfrak{a}}_{g+1} + {\mathfrak{b}}_{g+1}\, t_{\text{c}} + \cdots} = \nonumber \\
&=& \lim_{g\to+\infty} 4g^2 \left( \frac{G_g}{G_{g+1}}\, t_{\text{c}}^2 + \frac{{\mathfrak{a}}_g}{G_{g+1}}\, t_{\text{c}}^{2g} + \frac{{\mathfrak{b}}_g}{G_{g+1}}\, t_{\text{c}}^{2g+1} + \cdots \right),
\eea
\noindent
where $G_g \equiv \frac{{\mathfrak{c}}^{g-1}\, B_{2g}}{2g \left(2g-2\right)}$. Note that the first term is of order $t_{\text{c}}^2$ and then we have a tail starting at $t_{\text{c}}^{2g}$. When we take $g$ to infinity that tail will disappear leaving a result which is model independent (but still depending on the universality class). Now, since for large $g$ one has
\be\label{eq:Bernoulliasymptotics}
B_{2g}=(-1)^{g-1}\, \frac{2(2g)!}{(2\pi)^{2g}} \left( 1+\frac{1}{2^{2g}} + \cdots \right),
\ee
\noindent
one immediately finds
\be\label{eq:conifoldinstantonactionsquared}
A_{\text{c}}^2 = -\frac{4\pi^2}{{\mathfrak{c}}}\, t_{\text{c}}^2.
\ee
\noindent
So that indeed $A_{\text{c}} \propto t_{\text{c}}$ and it vanishes at the conifold point. As we shall see in the following, this type of large--order argument in which the model dependence, encoded in the constants $\{ {\mathfrak{a}}_g, {\mathfrak{b}}_g, \ldots \}$, goes away in the large $g$ limit, may be applied to extract the universal behavior of the higher--instanton free energies, associated to that instanton action (in the holomorphic limit).

As we discussed in section \ref{sec:largeorder}, the leading large--order behavior of $\mathcal{F}^{(0)}_g$ is expected to be of the form
\be\label{eq:largeorderfreeenergyconifold}
\mathcal{F}^{(0)}_g \simeq \sum_{k=1}^{+\infty} \frac{S_1^k}{2\pi\rmi}\, \sum_{h=0}^{+\infty} \frac{\Gamma \left( 2g - b_{(k)} - h \right)}{\left( k A_{\text{c}} \right)^{2g - b_{(k)} - h}}\, \mathcal{F}^{(k)}_h + \cdots.
\ee
\noindent
Here $S_1$ is the Stokes constant associated to the specific model, $(k)$ refers to the instanton sector $(0| \cdots |k| \cdots |0)$ associated to the conifold instanton action $A_{\text{c}}$, given by \eqref{eq:conifoldinstantonactionsquared}, $b_{(k)}$ is related to the starting power (similar to $b^{(k)}$ but not quite the same---thus the resembling notation), and the dots at the end stand for contributions from other (multi) instanton sectors including mixed/generalized sectors. These contributions are necessarily subleading as, around the conifold point, $A_{\text{c}}$ has the smallest absolute value being zero at the phase transition point.

We have already seen how to calculate $A_{\text{c}}$ from the large--order limit of the perturbative free energies. In the same spirit we may now calculate $b_{(1)}$,
\bea
b_{(1)} &=& \lim_{g\to+\infty} 2g \left(1 - \frac{A_{\text{c}}^2}{4g^2}\, \frac{\mathcal{F}^{(0)}_{g+1}}{\mathcal{F}^{(0)}_g} \right) = \lim_{g\to+\infty} 2g \left( 1 - \frac{A_{\text{c}}^2}{4g^2}\, \frac{G_{g+1}\, t_{\text{c}}^{2-2(g+1)} + {\mathfrak{a}}_{g+1} + \cdots}{G_g\, t_{\text{c}}^{2-2g} + {\mathfrak{a}}_g + \cdots}\right) = \nonumber \\
&=& \lim_{g\to+\infty} 2g \left( 1 - \frac{-\frac{4\pi^2}{{\mathfrak{c}}}\, t_{\text{c}}^2}{4g^2} \left( \frac{G_{g+1}}{G_g}\, \frac{1}{t_{\text{c}}^2} - \frac{{\mathfrak{a}}_g\, G_{g+1}}{G_g^2}\, t_{\text{c}}^{g+1} + \cdots \right) \right) = \nonumber\\
&=& \lim_{g\to+\infty} \frac{4g^2+2 \left(1+g-2g^2\right)}{2g} = 1,
\label{eq:betaresurgence}
\eea
\noindent
where, again, the model--dependent data is removed as $g$ goes to infinity. The next step is to calculate $\mathcal{F}^{(1)}_h$, $h=0,1,2,\ldots$. One finds
\bea
\frac{S_1}{2\pi\rmi}\, \mathcal{F}^{(1)}_0 &=& \lim_{g\to+\infty} \frac{A_{\text{c}}^{2g-b_{(1)}}}{\Gamma \left(2g-b_{(1)}\right)}\, \mathcal{F}^{(0)}_g = \nonumber\\
&=& \lim_{g\to+\infty} t_{\text{c}}\, \frac{\rmi \left(2g-1\right)}{2\pi\sqrt{{\mathfrak{c}}} \left(g-1\right)} = \frac{\rmi}{\pi\sqrt{{\mathfrak{c}}}}\, t_{\text{c}} = \frac{1}{2\pi^2}\, A_{\text{c}}.
\eea
\noindent
Next,
\bea
\frac{S_1}{2\pi\rmi}\, \mathcal{F}^{(1)}_1 &=& \lim_{g\to+\infty} \frac{A_{\text{c}}^{2g-2}}{\Gamma \left(2g-2\right)} \left( \mathcal{F}^{(0)}_g - \frac{\Gamma \left(2g-1\right)}{A_{\text{c}}^{2g-1}}\, \frac{S_1}{2\pi\rmi}\, \mathcal{F}^{(1)}_0 \right) = \nonumber\\
&=& \lim_{g\to+\infty} \left( \frac{2g-1}{2\pi^2} - \frac{2g-2}{2\pi^2} \right) = \frac{1}{2\pi^2}.
\eea
\noindent
The next free energy is different,
\bea
\frac{S_1}{2\pi\rmi}\, \mathcal{F}^{(1)}_2 &=&  \lim_{g\to+\infty}  \frac{\pi A_{\text{c}}^{2g-3}}{\Gamma \left(2g-3\right)} \left( \mathcal{F}^{(0)}_g - \frac{\Gamma \left(2g-1\right)}{\pi A_{\text{c}}^{2g-1}}\, \frac{S_1}{2\pi\rmi}\, \mathcal{F}^{(1)}_0 - \frac{\Gamma \left(2g-2\right)}{\pi A_{\text{c}}^{2g-2}}\, \frac{S_1}{2\pi\rmi}\, \mathcal{F}^{(1)}_1 \right) = \nonumber\\
&=& \lim_{g\to+\infty} \frac{\sqrt{{\mathfrak{c}}}}{4\pi^3\rmi\, t_{\text{c}}}\, \frac{1}{\Gamma \left(2g-3\right)} \left( \frac{\Gamma \left(2g\right)}{2 \left(g-1\right)} - \Gamma \left(2g-1\right) - \Gamma \left(2g-2\right) \right) = 0,
\eea
\noindent
since the term in brackets is identically zero. By induction one proves that $\mathcal{F}^{(1)}_h = 0$ for $h\geq 2$.

One could now go to the next instanton sector, $n=2$, and keep calculating. For that we would have to use \eqref{eq:Bernoulliasymptotics} including the term $2^{-2g}$. Proceeding in this way we find the same independence on the data associated to the specific model; as above all model--dependent data drops out after taking $g$ to infinity. Doing the same type of calculations, but taking into account the new factor of $2$ from \eqref{eq:Bernoulliasymptotics}, we obtain results which are very similar to the one--instanton sector. Going further we can actually calculate the holomorphic limit of the full $n$--instanton sector as
\bea
b_{(n)} &=& 1, \\
\frac{S_1^n}{2\pi\rmi}\, \mathcal{F}^{(n)}_0 &=& \frac{1}{n}\, \frac{\rmi}{\pi \sqrt{{\mathfrak{c}}}}\, t_{\text{c}}, \\
\frac{S_1^n}{2\pi\rmi}\, \mathcal{F}^{(n)}_1 &=& \frac{1}{n^2}\, \frac{1}{2\pi^2}, \\
\mathcal{F}^{(n)}_{g\geq 2} &=& 0.
\eea
\noindent
What we have found, based on a large--order analysis of the perturbative free energies, is that, in the holomorphic limit, the higher--instanton free energies associated to the conifold instanton action are universal, in the sense that they only depend on the universality class of the $c=1$ string and the Stokes constant. It is thus no surprise that the above results turn out to be essentially the same as the ones obtained in \cite{ps09}. What is important to notice is that this procedure provides a rather general method for fixing the nonperturbative holomorphic ambiguity, $f^{(n)}_g$, of all these sectors (and up to the Stokes constant): simply consider the holomorphic limit of the free energies appearing in the large--order relation and match them against the equations above, solving for $f^{(n)}_g$. Furthermore, it is important to remark that the coordinate $t_{\text{c}}$, which can be calculated from the Picard--Fuchs equations and has the interpretation of a period, is completely dependent upon the model we are trying to solve. In that way, the results above are universal in form but the actual expressions \textit{do} depend on the model.

\subsection{Fixing Ambiguities beyond Conifold Points}

If the model we are trying to solve lies in an universality class other than $c=1$, one might expect to be able to do a similar calculation as above. In that case, if we had a gap condition, we would have
\be\label{eq:generalizedgapcondition}
\mathcal{F}^{(0)}_g = H_g\, \tau^{\left(1-g\right)\left(2-\gamma\right)} + {\mathfrak{a}}_g + {\mathfrak{b}}_g\, \tau + \cdots,
\ee
\noindent
where the coefficients $H_g$ should grow factorially in order to identify the usual dominant instanton action, and $\tau$ is the appropriate coordinate around the phase transition point (which should also have a period interpretation, otherwise the gap condition becomes trivial). However, in order to be as thorough as for the conifold case, we would first need to know the exact form of $H_g$. For the $c=1$ class, the coefficients $G_g$ arise from the Gaussian matrix model which appears as the double--scaling limit at the phase transition point. For other classes one will have to address the double--scaling limit of the model, which naturally depends on the universality class. These double--scaling limits lie in hierarchies, such as the KdV hierarchy \cite{gm90}, but in full generality it is a complicated problem by itself to identify what the numbers $H_g$ are, and how they behave for large $g$, in analogy to \eqref{eq:Bernoulliasymptotics}. Without this information it is impossible to carry on the type of calculation we did in the previous subsection \ref{subsec:fixingambiguityconifold}. Another issue which also needs to be taken into account is the following. While it might be conceivable that topological string theories in the $c=1$ class will not show resonance or logarithmic sectors, a detailed study of the double--scaling limit of 2d gravity, described by the Painlev\'e I equation, \cite{gikm10,asv11}, or of 2d supergravity, described by the Painlev\'e II equation, \cite{sv13}, indicates that other classes will have them. Consequently, large--order expressions such as \eqref{eq:largeorderfreeenergyconifold} would need careful revision along the lines of \cite{asv11, sv13}.

It would also be interesting to address the aforementioned problem of more general phase transition points from a purely geometric point--of--view, and for example classify the universal types of large--order behavior of the perturbative series starting only from a classification of geometric Calabi--Yau threefold singularities. Unfortunately, while such a classification is known for K3 surfaces (\textit{i.e.}, Calabi--Yau twofolds), not only is it not known in the three--dimensional case, but it is actually believed that such singularities are not classifiable \cite{h04}. As a result, it seems to be quite difficult to prove in full generality that one can always find enough boundary conditions to completely solve the holomorphic anomaly equations, beyond the $c=1$ case. Nevertheless, in practical examples \cite{cesv14, cesv15}, we find that at least in the $c=1$ cases the available boundary conditions are indeed sufficient to make the problem completely integrable.

\section{Conclusions and Outlook}

In this work, by making use of resurgent analysis and transseries solutions, we have shown that the holomorphic anomaly equations are completely integrable at the nonperturbative level. In particular, we have suggested a prescription to solve, at full nonperturbative level, any closed topological string theory within the local Calabi--Yau B--model setting. This prescription is based on constructing classes of nonperturbative solutions, assembled as resurgent transseries, and we have shown very explicitly how to build two classes of solutions (``multi--instantonic'' transseries, with or without logarithmic sectors). These are the simplest classes of solutions one can think of, but, of course, it is conceivable that many other (more complicated) classes of solutions may exist and that they will be constructed along very similar fashions. Of course the next natural question is to ask how our methods apply in practice to specific examples, not only at the level of the explicit construction of the transseries solutions but also performing strong checks on these solutions by using resurgent analysis as in \cite{asv11, sv13}, finding Stokes constants and so on. We have applied our methods to two natural Calabi--Yau examples, local $\BP^2$ and local $\BP^1 \times \BP^1$, and we will report on these complete results some time soon \cite{cesv14, cesv15}.

For several reasons, the aforementioned examples are the first natural testing grounds for our general procedure. Recall that one of the simplest examples, precisely the conifold, was addressed at the nonperturbative level in \cite{ps09} (alongside the resolved conifold). In that paper it was suggested that the techniques used therein, mainly dealing with Borel analysis\footnote{The relation to resurgent methods was later pointed out in \cite{asv11}.}, should be applicable to topological strings in any Calabi--Yau geometry with a \textit{finite} number of non--zero Gopakumar--Vafa invariants. On the other hand, it was clear that the methods in \cite{ps09} would not work for threefolds with an infinite number of non--zero Gopakumar--Vafa invariants and it was already suggested in that reference that one should devise a new approach to tackle local $\BP^2$; possibly the simplest background to address in this class. The approach in the present paper precisely answers what the correct approach is, and in fact shows how general transseries methods are as they are essentially applicable to every possible geometry. In the examples of local $\BP^2$ and local $\BP^1 \times \BP^1$ \cite{cesv14, cesv15} we find many new features as compared to the case of the conifold \cite{ps09}. For example, the mirror curves have genus one and modularity plays an important role. In particular, one may now understand when the instanton action is dictated by A--cycle instantons (much like the conifold was \cite{ps09}) and when is it dictated by B--cycle instantons along the lines of \cite{msw07}. Generically, there will be several competing instanton actions in the game much along the lines of \cite{dmp11} and the resurgent analysis of these models is extremely rich and interesting \cite{cesv14, cesv15}. Another point to explore is what is the proper role of the generating functions which were so central in our present construction. There is an open possibility that these generating functions actually turn into ``real'' functions within specific examples, providing a resummation of the full nonperturbative solution. Whether this is possible and whether this would relate to the theta--function resummations in \cite{bde00, e08, em08} is an extremely interesting route of future research.

Moreover, the issue of modularity is by itself worth further investigation. For example, it would be very interesting to extend the results in \cite{abk06, asyz13} and understand the general modular properties of the many multi--instanton components of the transseries solutions. At the same time, the question of modularity relates to that of holomorphicity as addressed in \cite{em08}. Recall that when working on the gauge theory side, and computing the matrix model partition function, one starts off with a holomorphic object to which nonperturbative contributions further turn modular \cite{em08}. It is thus very natural to ask what happens in the present closed string picture. Of course that choosing a background entails breaking holomorphicity (while maintaining modularity) \cite{bcov93}, but could it be that the full set of transseries multi--instanton nonperturbative corrections still conspire in order to recover a holomorphic partition function as the final result? We believe these are very interesting questions for future research.

As one thinks about possible applications of our methods, many ideas quickly come to mind. For example, we have worked on the B--model side and it would be very interesting to investigate what would be the mirror analogue of our results, on the A--model. It is conceivable that this would relate to the topological vertex approach in \cite{adkmv03, akmv03}. Still concerning the B--model side, the extension of the matrix model large $N$ duality in \cite{dv02} from planar level to arbitrary (perturbative) genus, following \cite{emo07}, entailed a rather complete parallel between the closed string holomorphic anomaly equations and the matrix model topological recursion \cite{eo07, eo08}. In this way, our present construction seems to be pointing in the direction that there should also be a nonperturbative version of the topological recursion, this time around based on resurgent transseries solutions and on our general equation \eqref{eq:1parameterNPeq} (and its multi--parameter/multi--dimensional generalizations). Finally, it would be very interesting to extend our results to the compact Calabi--Yau case, perhaps along the lines of \cite{yy04, al07}.

\acknowledgments
We would like to thank Marcos Mari\~no for most insightful comments. RCS would further like to thank Murad Alim, Jos\'e Manuel Queiruga and Emanuel Scheidegger for useful discussions. RCS would like to thank the Institute for Theoretical Physics at the University of Amsterdam and the Department of Mathematics of Instituto Superior T\'ecnico, for their hospitality while this paper was being written. RS would like to thank the Department of Physics of the University of Geneva and CERN--TH Division for hospitality. RCS was supported by a Spanish FPU fellowship. RCS and JDE are supported in part by MICINN and FEDER (grant FPA2011--22594), Xunta de Galicia (Conseller\'{\i}a de Educaci\'on and grant PGIDIT10PXIB206075PR), and the Spanish Consolider--Ingenio 2010 Programme CPAN (CSD2007--00042). The research of RS was partially supported by the European Science Foundation Grant HoloGrav--4076 and by FCT--Portugal grants PTDC/MAT/119689/2010 and EXCL/MAT-GEO/0222/2012. The research of MV was partially supported by the European Research Council Advanced Grant EMERGRAV. The Centro de Estudios Cient\'{\i}ficos (CECs) is funded by the Chilean Government through the Centers of Excellence Base Financing Program of Conicyt.

\newpage

\appendix

\section{Multi--Dimensional Complex Moduli Space}\label{app:multidimcomplmodspace}

In the main body of the paper we have restricted ourselves to the case where the complex moduli space of our Calabi--Yau manifold is one--dimensional. A similar analysis of the equations (and solutions) can be carried out in the completely general case where the moduli space has arbitrary dimension $h^{2,1}$. We find that no new ingredients are required in this generalization which we shall present in the following, as we will make explicit (but schematic) the derivation of the equations for the higher--instanton free energies, assuming a general transseries \textit{ansatz} and including logarithmic sectors. Recall that indices $i$, $j$ run from $1$ to $h^{2,1}$.

We start off with the holomorphic anomaly equation
\be\label{eq:multidimholanomeq}
\frac{\partial F}{\partial S^{ij}} +\frac{1}{2}\left( U_i \: D_j F +U_j \: D_i F \right) - \frac{1}{2}g_s^2 \left( D_i D_j F + D_i F \: D_j F \right) = g_s^{-2} W_{ij}+V_{ij},
\ee
\noindent
and the usual transseries \textit{ansatz}
\be
F = \sum_{\boldsymbol{n}\in \mathbb{N}^\kappa_0} \boldsymbol{\sigma}^{\boldsymbol{n}}\,  \rme^{-A^{(\boldsymbol{n})}/g_s} \sum_{k=0}^{k_{\text{max}}^{(\boldsymbol{n})}} \log^k \left(g_s\right)\, F^{(\boldsymbol{n})[k]}(g_s).
\ee
\noindent
Next, one first calculates the covariant derivatives and the quadratic terms,
\bea
D_j F &=& \sum_{\boldsymbol{n}} \boldsymbol{\sigma}^{\boldsymbol{n}}\, \rme^{-A^{(\boldsymbol{n})}/g_s} \sum_{k=0}^{k_{\text{max}}^{(\boldsymbol{n})}} \log^k \left(g_s\right) \left( \partial_j - \frac{1}{g_s} \partial_j A^{(\boldsymbol{n})}\right) F^{(\boldsymbol{n})[k]}, \\
D_i D_j F &=& \sum_{\boldsymbol{n}} \boldsymbol{\sigma}^{\boldsymbol{n}}\, \rme^{-A^{(\boldsymbol{n})}/g_s} \sum_{k=0}^{k_{\text{max}}^{(\boldsymbol{n})}} \log^k \left(g_s\right) \times \nonumber \\
&&
\times \left\{ \left( \partial_i - \frac{1}{g_s} \partial_i A^{(\boldsymbol{n})} \right)\left( \partial_j - \frac{1}{g_s} \partial_j A^{(\boldsymbol{n})} \right) - \Gamma^k_{ij} \left( \partial_k - \frac{1}{g_s} \partial_k A^{(\boldsymbol{n})} \right)  \right\} F^{(\boldsymbol{n})[k]}, \\
D_i F \, D_j F &=& \sum_{\boldsymbol{n}} \boldsymbol{\sigma}^{\boldsymbol{n}}\, \rme^{-A^{(\boldsymbol{n})}/g_s} \sum_{\boldsymbol{m}=\boldsymbol{0}}^{\boldsymbol{n}} \sum_{k=0}^{k_{\text{max}}^{(\boldsymbol{m})}+k_{\text{max}}^{(\boldsymbol{n}-\boldsymbol{m})}} \log^k \left(g_s\right) \times \\
&&
\times \sum_{\ell=\max \left( 0,k-k_{\text{max}}^{(\boldsymbol{n}-\boldsymbol{m})} \right)}^{\min \left( k,k_{\text{max}}^{(\boldsymbol{m})} \right)}\left( \partial_i - \frac{1}{g_s} \partial_i A^{(\boldsymbol{m})} \right) F^{(\boldsymbol{m})[\ell]}  \left( \partial_j - \frac{1}{g_s} \partial_j A^{(\boldsymbol{n}-\boldsymbol{m})} \right) F^{(\boldsymbol{n}-\boldsymbol{m})[k-\ell]}. \nonumber
\eea
\noindent
In the last equation we have used $A^{(\boldsymbol{m})} + A^{(\boldsymbol{n}-\boldsymbol{m})} = A^{(\boldsymbol{n})}$ in the exponential term. Using these explicit expressions, the holomorphic anomaly equation \eqref{eq:multidimholanomeq} reads
\bea
\sum_{\boldsymbol{n}} \boldsymbol{\sigma}^{\boldsymbol{n}}\, \rme^{-A^{(\boldsymbol{n})}/g_s} && \Bigg\{ \sum_{k=0}^{k_{\text{max}}^{(\boldsymbol{n})}} \log^k \left(g_s\right) \left[ \left( \partial_{S^{ij}} - \frac{1}{g_s} \partial_{S^{ij}}A^{(\boldsymbol{n})} \right) F^{(\boldsymbol{n})[k]} + \right.  \nonumber \\
&&
+ \frac{1}{2} U_i \left( \partial_j - \frac{1}{g_s} \partial_j A^{(\boldsymbol{n})}\right) F^{(\boldsymbol{n})[k]} + \frac{1}{2} U_j \left( \partial_i - \frac{1}{g_s} \partial_i A^{(\boldsymbol{n})} \right) F^{(\boldsymbol{n})[k]} - \nonumber \\
&&
\hspace{-20pt}
\left.- \frac{1}{2} g_s^2 \left( \left( \partial_i - \frac{1}{g_s} \partial_i A^{(\boldsymbol{n})}\right) \left( \partial_j - \frac{1}{g_s} \partial_j A^{(\boldsymbol{n})}\right) - \Gamma^k_{ij} \left( \partial_k - \frac{1}{g_s} \partial_k A^{(\boldsymbol{n})}\right) \right) F^{(\boldsymbol{n})[k]} \right] - \nonumber \\
&&
\hspace{-85pt}
- \frac{1}{2} g_s^2 \sum_{\boldsymbol{m}=\boldsymbol{0}}^{\boldsymbol{n}} \sum_{k=0}^{k_{\text{max}}^{(\boldsymbol{m})}+k_{\text{max}}^{(\boldsymbol{n}-\boldsymbol{m})}} \log^k \left(g_s\right) \sum_\ell \left( \partial_i - \frac{1}{g_s} \partial_i A^{(\boldsymbol{m})}\right) F^{(\boldsymbol{m})[\ell]} \left( \partial_j - \frac{1}{g_s} \partial_j A^{(\boldsymbol{n}-\boldsymbol{m})}\right) F^{(\boldsymbol{n}-\boldsymbol{m})[k-\ell]} - \nonumber \\
&&
- \delta_{\boldsymbol{n}\boldsymbol{0}} \left( \frac{1}{g_s^2} W_{ij} + V_{ij} \right) \Bigg\} = 0.
\eea
\noindent
As explained in section \ref{sec:holoeqs}, the perturbative sector $\boldsymbol{n}=\boldsymbol{0}$ gives back the usual holomorphic anomaly equations \eqref{eq:propagatorsholanomeqs}, when (\ref{eq:Uconstraint}--\ref{eq:Wconstraint}) are satisfied. As such, we shall now focus on $\boldsymbol{n}\neq \boldsymbol{0}$. Separate the terms corresponding to $\boldsymbol{m}=\boldsymbol{0}$ and $\boldsymbol{m}=\boldsymbol{n}$ from the sum in $\boldsymbol{m}$, and collect them along with the $U_i$ terms. One finds
\bea
&&
\sum_{k=0}^{k_{\text{max}}^{(\boldsymbol{n})}} \log^k \left(g_s\right)\left[ \left( \partial_{S^{ij}} - \frac{1}{g_s} \partial_{S^{ij}}A^{(\boldsymbol{n})} \right) F^{(\boldsymbol{n})[k]} \right. - \nonumber \\
&&
- \frac{1}{2} g_s^2 \left\{ \partial_i \widetilde{F}^{(\boldsymbol{0})} \left( \partial_j - \frac{1}{g_s} \partial_j A^{(\boldsymbol{n})}\right) + \partial_j \widetilde{F}^{(\boldsymbol{0})} \left( \partial_i - \frac{1}{g_s} \partial_i A^{(\boldsymbol{n})}\right) + \right. \nonumber \\
&&
\hspace{35pt}
\left.\left. + \left( \partial_i - \frac{1}{g_s} \partial_i A^{(\boldsymbol{n})} \right)\left( \partial_j - \frac{1}{g_s} \partial_j A^{(\boldsymbol{n})} \right) - \Gamma^k_{ij} \left( \partial_k - \frac{1}{g_s} \partial_k A^{(\boldsymbol{n})} \right) \right\} F^{(\boldsymbol{n})[k]} \right] = \\
&&
= \frac{1}{2} g_s^2  \sideset{}{'}\sum_{\boldsymbol{m}=\boldsymbol{0}}^{\boldsymbol{n}}\, \sum_{k=0}^{k_{\text{max}}^{(\boldsymbol{m})}+k_{\text{max}}^{(\boldsymbol{n}-\boldsymbol{m})}} \log^k \left(g_s\right) \sum_\ell \left( \partial_i - \frac{1}{g_s} \partial_i A^{(\boldsymbol{m})} \right) F^{(\boldsymbol{m})[\ell]} \left( \partial_j - \frac{1}{g_s} \partial_j A^{(\boldsymbol{n}-\boldsymbol{m})} \right) F^{(\boldsymbol{n}-\boldsymbol{m})[k-\ell]}, \nonumber
\eea
\noindent
where we have defined $\partial_i \widetilde{F}^{(\boldsymbol{0})} := \partial_i F^{(\boldsymbol{0})}- \frac{1}{g_s^2}U_i$. Noticing that $F^{(\boldsymbol{n})[k]} = 0$ if $k<0$ or $k> k_{\text{max}}^{(\boldsymbol{n})}$, and defining the symmetrization of two indices as usual $T_{(ij)}:= \frac{1}{2} \left( T_{ij} + T_{ji} \right)$, one can write
\bea
&&
\left( \partial_{S^{ij}} - \frac{1}{g_s} \partial_{S^{ij}}A^{(\boldsymbol{n})} \right) F^{(\boldsymbol{n})[k]} - \nonumber \\
&&
- \frac{1}{2} g_s^2 \left( D_{(i|} - \frac{1}{g_s} \partial_{(i|}A^{(\boldsymbol{n})} + 2\, \partial_{(i|} \widetilde{F}^{(\boldsymbol{0})} \right) \left( \partial_{|j)} - \frac{1}{g_s} \partial_{|j)}A^{(\boldsymbol{n})} \right) F^{(\boldsymbol{n})[k]} = \nonumber \\
&&
= \frac{1}{2}g_s^2\, \sideset{}{'}\sum_{\boldsymbol{m}=\boldsymbol{0}}^{\boldsymbol{n}} \sum_\ell \left( \partial_i - \frac{1}{g_s} \partial_i A^{(\boldsymbol{m})} \right) F^{(\boldsymbol{m})[\ell]} \left( \partial_j - \frac{1}{g_s} \partial_j A^{(\boldsymbol{n}-\boldsymbol{m})} \right) F^{(\boldsymbol{n}-\boldsymbol{m})[k-\ell]}.
\eea
\noindent
In here we have also used that the Christoffel symbols are symmetric in the lower indices. Now, the structure of the differential operators in this equation suggests the following definition of covariant derivatives
\bea
\nabla^{(\boldsymbol{n})}_i F^{(\boldsymbol{n})[k]} &:=&  \left( \partial_i - \frac{1}{g_s} \partial_i A^{(\boldsymbol{n})} \right) F^{(\boldsymbol{n})[k]}, \\
\nabla^{(\boldsymbol{n})}_i \nabla^{(\boldsymbol{n})}_j F^{(\boldsymbol{n})[k]} &:=& \left( D_{(i|} - \frac{1}{g_s} \partial_{(i|}A^{(\boldsymbol{n})} + 2\, \partial_{(i|} \widetilde{F}^{(\boldsymbol{0})} \right) \left( \partial_{|j)} - \frac{1}{g_s} \partial_{|j)} A^{(\boldsymbol{n})} \right) F^{(\boldsymbol{n})[k]} = \\
&=&
\left( \delta^k_{(i|} \left( \partial_{|j)} -\frac{1}{g_s} \partial_{|j)}A^{(\boldsymbol{n})} + 2\, \partial_{|j)} \widetilde{F}^{(\boldsymbol{0})} \right) - \Gamma^k_{ij} \right) \left( \partial_k - \frac{1}{g_s} \partial_k A^{(\boldsymbol{n})}\right) F^{(\boldsymbol{n})[k]}, \nonumber \\
\nabla^{(\boldsymbol{n})}_{S^{ij}} F^{(\boldsymbol{n})[k]} &:=& \left( \partial_{S^{ij}} - \frac{1}{g_s} \partial_{S^{ij}} A^{(\boldsymbol{n})} \right) F^{(\boldsymbol{n})[k]}.
\eea
\noindent
With this compact notation the equations have the final form:
\be\label{eq:multidimcompacteq}
\left( \nabla^{(\boldsymbol{n})}_{S^{ij}} - \frac{1}{2} g_s^2\, \nabla^{(\boldsymbol{n})}_i \nabla^{(\boldsymbol{n})}_j \right) F^{(\boldsymbol{n})[k]} = \frac{1}{2} g_s^2\, \sideset{}{'}\sum_{\boldsymbol{m}=\boldsymbol{0}}^{\boldsymbol{n}}\, \sum_\ell \nabla^{(\boldsymbol{m})}_i F^{(\boldsymbol{m})[\ell]}\, \nabla^{(\boldsymbol{n}-\boldsymbol{m})}_i F^{(\boldsymbol{n}-\boldsymbol{m})[k-\ell]}.
\ee

Up to this point we have not made explicit the $g_s$--expansion of the free energies. Before we do that, let us note that the structure of our original holomorphic anomaly equation \eqref{eq:multidimholanomeq}, as well as the traditional (perturbative) formulation of these equations as in \eqref{eq:propagatorsholanomeqs}, including a propagator--derivative, a second--order derivative and a quadratic term, is precisely preserved in here and it will remain so even once we expand the free energies in $g_s$. Further notice that the contribution from lower instanton sectors in \eqref{eq:multidimcompacteq} comes only from the quadratic terms.

Just as we did in the main body of the paper, the operator on the left--hand--side of \eqref{eq:multidimcompacteq} will be denoted by $\mathcal{D}^{(\boldsymbol{n})}_{ij}$. This operator, which includes the propagator--derivative and second--order derivatives, has a $g_s$--expansion starting at $g=-1$,
\be
\left( \nabla^{(\boldsymbol{n})}_{S^{ij}} - \frac{1}{2} g_s^2\, \nabla^{(\boldsymbol{n})}_i \nabla^{(\boldsymbol{n})}_j \right) \equiv \mathcal{D}^{(\boldsymbol{n})}_{ij} (g_s) = \sum_{g=-1}^{+\infty} g_s^g\, \mathcal{D}^{(\boldsymbol{n})}_{ij;g}.
\ee
\noindent
Here $\mathcal{D}^{(\boldsymbol{n})}_{ij;-1}= - \partial_{S^{ij}}A^{(\boldsymbol{n})}$, which will turn out to be the zero operator once we find that the instanton actions are holomorphic. Let the asymptotic expansions for the higher instanton sectors be
\be
F^{(\boldsymbol{n})[k]}(g_s) \simeq \sum_{g=0}^{+\infty} g_s^{g+b^{(\boldsymbol{n})[k]}} F^{(\boldsymbol{n})[k]}_g (z_i, S^{ij}),
\ee
\noindent
and calculate the left and right--hand sides of \eqref{eq:multidimcompacteq}. One finds
\be
\mathcal{D}^{(\boldsymbol{n})}_{ij} F^{(\boldsymbol{n})[k]} = g^{b^{(\boldsymbol{n})[k]}} \sum_{g=-1}^{+\infty} g_s^g \left( -\partial_{S^{ij}} A^{(\boldsymbol{n})}\,  F^{(\boldsymbol{n})[k]}_{g+1} + \sum_{h=0}^g \mathcal{D}^{(\boldsymbol{n})}_{ij;h} F^{(\boldsymbol{n})[k]}_{g-h} \right),
\ee
\noindent
and
\bea
&&
g_s^2\, \sideset{}{'}\sum_{\boldsymbol{m}=\boldsymbol{0}}^{\boldsymbol{n}}\, \sum_\ell \left( \partial_i - \frac{1}{g_s} \partial_i A^{(\boldsymbol{m})} \right) F^{(\boldsymbol{m})[\ell]} \left( \partial_j - \frac{1}{g_s} \partial_j A^{(\boldsymbol{n}-\boldsymbol{m})} \right) F^{(\boldsymbol{n}-\boldsymbol{m})[k-\ell]} = \\
&&
= \sideset{}{'}\sum_{\boldsymbol{m}=\boldsymbol{0}}^{\boldsymbol{n}}\, \sum_\ell g_s^{b^{(\boldsymbol{m})[\ell]}+b^{(\boldsymbol{n}-\boldsymbol{m})[k-\ell]}} \times \nonumber \\
&&
\hspace{30pt}
\times \sum_{g=0}^{+\infty} g_s^g\, \sum_{h=0}^g \left( \partial_i F^{(\boldsymbol{m})[\ell]}_{h-1} - \partial_i A^{(\boldsymbol{m})} F^{(\boldsymbol{m})[\ell]}_h \right)\left( \partial_j F^{(\boldsymbol{n}-\boldsymbol{m})[k-\ell]}_{g-1-h} - \partial_j A^{(\boldsymbol{n}-\boldsymbol{m})} F^{(\boldsymbol{n}-\boldsymbol{m})[k-\ell]}_{g-h} \right). \nonumber
\eea
\noindent
Combining both calculations and multiplying everything by $g_s^{-b^{(\boldsymbol{n})}}$,  \eqref{eq:multidimcompacteq} becomes
\bea\label{eq:multidimNPeq}
&&
- \partial_{S^{ij}}A^{(\boldsymbol{n})}\, F^{(\boldsymbol{n})[k]}_{g+1} + \sum_{h=0}^g \mathcal{D}^{(\boldsymbol{n})}_{ij;h} F^{(\boldsymbol{n})[k]}_{g-h} = \\
&&
= \frac{1}{2} \sideset{}{'}\sum_{\boldsymbol{m}=\boldsymbol{0}}^{\boldsymbol{n}}\, \sum_\ell\, \sum_{h=0}^{g-B} \left( \partial_i F^{(\boldsymbol{m})[\ell]}_{h-1} - \partial_i A^{(\boldsymbol{m})} F^{(\boldsymbol{m})[\ell]}_h \right) \left( \partial_j F^{(\boldsymbol{n}-\boldsymbol{m})[k-\ell]}_{g-1-B-h} - \partial_j A^{(\boldsymbol{n}-\boldsymbol{m})} F^{(\boldsymbol{n}-\boldsymbol{m})[k-\ell]}_{g-B-h} \right) = 0, \nonumber
\eea
\noindent
for $g=-1,0,1,2,\ldots$, $k=0,1,2,\dots$, $1\leq i\leq j \leq h^{2,1}$, and $\boldsymbol{n}\neq \boldsymbol{0}$. Here $B$ stands for $B(\boldsymbol{n},\boldsymbol{m})[k,\ell] := b^{(\boldsymbol{m})[\ell]} + b^{(\boldsymbol{n}-\boldsymbol{m})[k-\ell]} - b^{(\boldsymbol{n})[k]}$. One has:
\bea
\mathcal{D}^{(\boldsymbol{n})}_{ij;0} &=& \partial_{S^{ij}} - \frac{1}{2} \partial_i A^{(\boldsymbol{n})} \partial_j A^{(\boldsymbol{n})}, \\
\mathcal{D}^{(\boldsymbol{n})}_{ij;1} &=& \frac{1}{2} D_i D_j A^{(\boldsymbol{n})} +\partial_{(i|} A^{(\boldsymbol{n})}\left( D_{|j)} + \partial_{|j)} F^{(\boldsymbol{0})}_1 \right), \\
\mathcal{D}^{(\boldsymbol{n})}_{ij;2} &=& - \frac{1}{2} D_i D_j - \partial_{(i|}F^{(\boldsymbol{0})}_1 D_{|j)}, \\
\mathcal{D}^{(\boldsymbol{n})}_{ij;2h-1} &=& \partial_{(i|}F^{(\boldsymbol{0})}_h \partial_{|j)}A^{(\boldsymbol{n})}, \qquad h=2,3,\ldots, \\
\mathcal{D}^{(\boldsymbol{n})}_{ij;2h} &=& -\partial_{(i|}F^{(\boldsymbol{0})}_h D_{|j)}, \qquad h=2,3,\ldots.
\eea

Going back to \eqref{eq:multidimNPeq}, let us consider the equation for $\boldsymbol{n} = ( 0 | \cdots |1| \cdots | 0)$ (an one--instanton sector), $k=0$, and $g=-1$. The second and third terms are zero and one immediately obtains
\be
\partial_{S^{ij}} A^{(0 | \cdots |1| \cdots | 0)}\, F^{(0 | \cdots |1| \cdots | 0)[0]}_0  = 0.
\ee
\noindent
Since by construction $F^{(0 | \cdots |1| \cdots | 0)[0]}_0 \neq 0$, we find the by now familiar result that the instanton actions are all holomorphic. With this important piece of information the nonperturbative holomorphic anomaly equations finally read (here we separate the $\mathcal{D}_0$ term)
\bea\label{eq:multidimrecursiveeq}
&&
\left( \partial_{S^{ij}} - \frac{1}{2} \partial_i A^{(\boldsymbol{n})} \partial_j A^{(\boldsymbol{n})} \right) F^{(\boldsymbol{n})[k]}_g = -\sum_{h=1}^g \mathcal{D}^{(\boldsymbol{n})}_{ij;h} F^{(\boldsymbol{n})[k]}_{g-h} + \\
&&
+ \frac{1}{2}\, \sideset{}{'}\sum_{\boldsymbol{m}=\boldsymbol{0}}^{\boldsymbol{n}}\,  \sum_\ell\, \sum_{h=0}^{g-B} \left( \partial_i F^{(\boldsymbol{m})[\ell]}_{h-1} - \partial_i A^{(\boldsymbol{m})} F^{(\boldsymbol{m})[\ell]}_h \right) \left( \partial_j F^{(\boldsymbol{n}-\boldsymbol{m})[k-\ell]}_{g-1-B-h} - \partial_j A^{(\boldsymbol{n}-\boldsymbol{m})} F^{(\boldsymbol{n}-\boldsymbol{m})[k-\ell]}_{g-B-h} \right). \nonumber
\eea
\noindent
Let us make a few remarks concerning this result. First notice the familiar structure which includes a propagator--derivative (now enlarged with the instanton action) equated to a second--order derivative (to be found inside $\mathcal{D}_2$) added with quadratic terms. Further, the equations are recursive in the instanton sector $\boldsymbol{n}$, the logarithmic power $k$, and the genus $g$. However, because of the functions $B$ and $k_{\text{max}}$, it may be the case that for some equations the left--hand--side of \eqref{eq:multidimrecursiveeq} is zero while the right--hand--side is not. This does not break the recursiveness of the equations but introduces some extra constraints that need to be satisfied in the particular model under study (see the discussion at the end of subsection \ref{sec:resonanceandlogarithmicsectors}). In the case in which we have no constraints and the instanton actions are generic (\textit{i.e.}, they are not constant nor linearly dependent over the integers) we can prove that the structure of the free energies obtained as solutions from \eqref{eq:multidimrecursiveeq} is given by the following theorem.

\begin{thm}
For any $\boldsymbol{n}\neq \boldsymbol{0}$, $k\in\{0,\ldots,k_{\rm{max}}^{(\boldsymbol{n})}\}$, and $g\geq 0$, the structure of the nonperturbative free energies has the form
\begin{equation}
F^{(\boldsymbol{n})[k]}_g = \sum_{\{\gamma_{\boldsymbol{n}}\}} \rme^{\frac{1}{2}\sum_{\alpha,\beta=1}^\kappa a_{\alpha\beta} \left( \boldsymbol{n}; \gamma_{\boldsymbol{n}} \right) \sum_{i,j=1}^{h^{2,1}} \partial_i A_\alpha\, \partial_j A_\beta\, S^{ij}}\, \mathrm{Pol} \left( S^{ij}; 3 \left( g + b^{(\boldsymbol{n})[k]} - \lambda_{b,k_{\rm{max}}}^{[k]} \left( \boldsymbol{n}; \gamma_{\boldsymbol{n}} \right) \right) \right),
\end{equation}
\noindent
where the set of numbers $\left\{ a_{\alpha\beta} \left( \boldsymbol{n}; \gamma_{\boldsymbol{n}} \right) \right\}$ and $\left\{ \lambda_{b,k_{\rm{max}}}^{[k]} \left( \boldsymbol{n}; \gamma_{\boldsymbol{n}} \right) \right\}$ are read off from the generating function
\bea
\Phi_{b,k_{\rm{max}}} &=& \sideset{}{'}\prod_{\boldsymbol{m}=\boldsymbol{0}}^{+\infty}\,  \prod_{\ell=0}^{k_{\rm{max}}^{(\boldsymbol{m})}} \frac{1}{1 - \varphi^{b^{(\boldsymbol{m})[\ell]}}\, \psi^\ell\, \prod_{\alpha,\beta=1}^\kappa E_{\alpha\beta}^{m_\alpha m_\beta}\, \prod_{\alpha=1}^\kappa \rho_\alpha^{m_\alpha}} = \nonumber \\
&=& \sum_{\boldsymbol{n}=\boldsymbol{0}}^{+\infty} \boldsymbol{\rho}^{\boldsymbol{n}}\, \sum_{k=0}^{k_{\rm{max}}^{(\boldsymbol{n})}} \psi^k\, \sum_{\left\{ \gamma_{\boldsymbol{n}} \right\}} \prod_{\alpha,\beta=1}^\kappa E_{\alpha\beta}^{a_{\alpha\beta} \left( \boldsymbol{n}; \gamma_{\boldsymbol{n}} \right)}\, \varphi^{\lambda_{b,k_{\rm{max}}}^{[k]} \left( \boldsymbol{n}; \gamma_{\boldsymbol{n}} \right)} \left( 1 + \mathcal{O} \left( \varphi \right) \right).
\eea
\noindent
Here $\mathrm{Pol} \left(S^{ij};d\right)$ stands for a polynomial of total degree $d$ in the variables $\{S^{ij}\}$ (and whose coefficients have a holomorphic dependence in $\{z_i \}$). Whenever $d<0$, the polynomial is taken to be identically zero. We are assuming that $b^{(\boldsymbol{m})[\ell]} + b^{(\boldsymbol{n}-\boldsymbol{m})[k-\ell]} - b^{(\boldsymbol{n})[k]}\geq 0$, and $k_{\rm{max}}^{(\boldsymbol{n})}-k_{\rm{max}}^{(\boldsymbol{m})}-k_{\rm{max}}^{(\boldsymbol{n}-\boldsymbol{m})} \geq 0$.
\end{thm}

The proof of this theorem is completely analogous to that of Theorem \ref{structuraltheorem}. The only novelty one has to take into account are the indices $i$, $j$, and the presence of logarithmic sectors. It is also important to notice that the generating function is the same, independently of the dimension of the complex moduli space. This is because all the propagators are on equal footing with each other---albeit the same thing cannot be said about the instanton sectors. This is why the coefficients $a$ carry indices $\alpha$, $\beta$, but not $i$, $j$.

\section{Proof of Set Recursion Lemma}\label{ap:proofofrecursionlemma}

In this appendix we prove Lemma \ref{lem:setrecurrence}, which provides a bridge between the recursions defining the coefficients which appear in the proof of Theorem \ref{structuraltheorem}, and the description through the generating function \eqref{eq:thmcurlyFstructure}.

\begin{proof}[Proof of Lemma \ref{lem:setrecurrence}]
The proof essentially consists in calculating $\left( \Phi_b - 1 \right)^2$ in two different ways, and then equating the results. In the first way we simply square a formal power series as
\begin{eqnarray}
\label{eq:FminusOne1}
\left( \Phi_b - 1 \right)^2 &=& \left( \sideset{}{'}\sum_{\boldsymbol{n}=\boldsymbol{0}}^{+\infty} \boldsymbol{\rho}^{\boldsymbol{n}}\, \sum_{\left\{ \gamma_{\boldsymbol{n}} \right\}} E^{a \left( \boldsymbol{n}; \gamma_{\boldsymbol{n}} \right)}\,  \varphi^{\lambda_b \left( \boldsymbol{n}; \gamma_{\boldsymbol{n}} \right)} \left( 1+ \mathcal{O} \left( \varphi \right) \right) \right)^2 = \\
&& 
\hspace{-45pt}
= \sideset{}{''}\sum_{\boldsymbol{n}=\boldsymbol{0}}^{+\infty} \boldsymbol{\rho}^{\boldsymbol{n}}\, \sideset{}{'}\sum_{\boldsymbol{m}=\boldsymbol{0}}^{\boldsymbol{n}} \sum_{\left\{ \gamma_{\boldsymbol{m}}, \gamma_{\boldsymbol{n}-\boldsymbol{m}} \right\}} E^{a \left( \boldsymbol{m}; \gamma_{\boldsymbol{m}} \right) + a \left( \boldsymbol{n}-\boldsymbol{m}; \gamma_{\boldsymbol{n}-\boldsymbol{m}} \right)}\, \varphi^{\min \left\{ \lambda_b \left( \boldsymbol{m}; \gamma_{\boldsymbol{m}} \right) + \lambda_b \left( \boldsymbol{n}-\boldsymbol{m}; \gamma_{\boldsymbol{n}-\boldsymbol{m}} \right) \right\}} \left(1+\mathcal{O}\left( \varphi \right) \right), \nonumber
\end{eqnarray}
\noindent
where the double--prime ($\:''$) means that $\boldsymbol{n}$ with norm equal to $0$ or $1$ are excluded from the sum.

The second way to compute the above quantity is by simply expanding the square
\begin{equation}\label{eq:FminusOne2}
\left( \Phi_b - 1 \right)^2 = \Phi_b^2 - 2 \Phi_b +1,
\end{equation}
\noindent
with
\begin{equation}
\Phi_b^2 = \sideset{}{'}\prod_{\boldsymbol{m}=\boldsymbol{0}}^{+\infty} \frac{1}{\left( 1 - \varphi^{b^{(\boldsymbol{m})}}\, \prod_{\alpha,\beta=1}^\kappa E_{\alpha\beta}^{m_\alpha m_\beta}\, \prod_{\alpha=1}^\kappa \rho_\alpha^{m_\alpha} \right)^2}.
\end{equation}
\noindent
After using the formal expansion $(1-x)^{-2} = 1 + 2 x + 3 x^2 + 4 x^3 + \cdots$ and some familiar manipulations from our earlier analysis in the main body of the text, one obtains
\begin{equation}\label{eq:explicitF2genfunction}
\Phi_b^2 = \sum_{\boldsymbol{n}=\boldsymbol{0}}^{+\infty} \boldsymbol{\rho}^{\boldsymbol{n}}\, \sum_{\left\{ \gamma_{\boldsymbol{n}} \right\}} E^{a \left( \boldsymbol{n}; \gamma_{\boldsymbol{n}} \right)} \sum_{\left\{ r_{\boldsymbol{m}} \right\} \in \left\{ \gamma_{\boldsymbol{n}} \right\}} \varphi^{\sum_{\boldsymbol{m}=\boldsymbol{0}}^{+\infty} r_{\boldsymbol{m}}\, b^{(\boldsymbol{m})}} \prod_{\boldsymbol{m}=\boldsymbol{0}}^{+\infty} \left( r_{\boldsymbol{m}} + 1 \right).
\end{equation}
\noindent
Now, since (\ref{eq:FminusOne1}) does note have any terms with $\|\boldsymbol{n}\|=0,1$ we should check that (\ref{eq:FminusOne2}) does not have them either. Let us do this explicitly. For $\|\boldsymbol{n}\|=0$ one finds
\begin{equation}
\left. \Phi_b^2 \right|_{\|\boldsymbol{n}\|=0}=1, \quad \left. \Phi_b \right|_{\|\boldsymbol{n}\|=0}=1 \qquad \Rightarrow \qquad \left. \left( \Phi_b - 1 \right)^2 \right|_{\|\boldsymbol{n}\|=0} = 1-2\cdot 1+1 = 0.
\end{equation}
\noindent
For $\|\boldsymbol{n}\|=1$ note that there is only one class ($\widehat{\gamma}_{\boldsymbol{n}}$) and we thus find
\begin{equation}
\left. \Phi_b^2 \right|_{\|\boldsymbol{n}\|=1} = \sum_{\alpha=1}^\kappa \rho_\alpha\, E_{\alpha\alpha}\, \varphi^{b^{(\boldsymbol{n})}} \left( 1 + 1 \right), \,\, \left. \Phi_b \right|_{\|\boldsymbol{n}\|=1} = \sum_{\alpha=1}^\kappa \rho_\alpha\, E_{\alpha\alpha}\, \varphi^{b^{(\boldsymbol{n})}} \quad \Rightarrow \quad \left. \left( \Phi_b - 1 \right)^2 \right|_{\|\boldsymbol{n}\|=1} = 0.
\end{equation}
\noindent
In this way, we have obtained a second way of writing $\left( \Phi - 1 \right)^2$, as
\begin{equation}
\sideset{}{''}\sum_{\boldsymbol{n}=\boldsymbol{0}}^{+\infty} \boldsymbol{\rho}^{\boldsymbol{n}}\, \sum_{\left\{ \gamma_{\boldsymbol{n}} \right\}} E^{a \left( \boldsymbol{n}; \gamma_{\boldsymbol{n}} \right)} \sum_{\left\{ r_{\boldsymbol{m}} \right\} \in \left\{ \gamma_{\boldsymbol{n}} \right\}} \varphi^{\sum_{\boldsymbol{m}=\boldsymbol{0}}^{+\infty} r_{\boldsymbol{m}}\, b^{(\boldsymbol{m})}} \left( \prod_{\boldsymbol{m}=\boldsymbol{0}}^{+\infty} \left( r_{\boldsymbol{m}} + 1 \right) - 2 \right).
\end{equation}
\noindent
All we have left to do is to notice that $\prod_{\boldsymbol{m}=\boldsymbol{0}}^{+\infty} \left( r_{\boldsymbol{m}} + 1 \right) - 2 \geq 0$ with equality holding only for the class $\widehat{\gamma}_{\boldsymbol{n}}$. Therefore,
\begin{equation}\label{eq:FminusOne3}
\left( \Phi_b - 1 \right)^2 = \sideset{}{''}\sum_{\boldsymbol{n}=\boldsymbol{0}}^{+\infty} \boldsymbol{\rho}^{\boldsymbol{n}}\, \sideset{}{'}\sum_{\left\{ \gamma_{\boldsymbol{n}} \right\}} E^{a \left( \boldsymbol{n}; \gamma_{\boldsymbol{n}} \right)}\,  \varphi^{\lambda_b \left( \boldsymbol{n}; \gamma_{\boldsymbol{n}} \right)}\, \mathcal{O} \left( \varphi^0 \right).
\end{equation}
\noindent
Comparison of (\ref{eq:FminusOne1}) and (\ref{eq:FminusOne3}) concludes the proof.
\end{proof}

\newpage


\bibliographystyle{plain}

\end{document}